\documentclass[12pt]{iopart}
\usepackage{graphicx}
\usepackage{subcaption}
\usepackage{xcolor,cite}
\usepackage{amsthm}
\usepackage{amssymb}
\usepackage[utf8]{inputenc}
\newtheorem{theorem}{Theorem}[section]
\newtheorem{definition}{Definition}[section]

\newtheorem{proposition}{Proposition}
\newtheorem{lemma}[theorem]{Lemma}
\newtheorem*{remark}{Remark}

\begin{document}

\title[Paperfolding Structures as Templates for Horseshoes]{Paperfolding Structures as Templates for Horseshoes in Multidimensional Hénon Maps}

\author{Jizhou Li$^1$, Keisuke Fujioka$^2$, and Akira Shudo$^2$}

\address{$^1$Research Institute for Electronic Science, Hokkaido University, Sapporo, Hokkaido 001-0020, Japan}
\address{$^2$Department of Physics, Tokyo Metropolitan University, Tokyo 192-0397, Japan}

\ead{jizhouli@es.hokudai.ac.jp \quad shudo@tmu.ac.jp}
\vspace{10pt}
\begin{indented}
\item[]September 2025
\end{indented}

\begin{abstract}
We propose a novel framework for analyzing the geometric structure of horseshoes arising in three- and four-dimensional Hénon-type maps by introducing paperfolding structures as geometric templates. These structures capture the folding and stacking mechanisms characteristic of high-dimensional chaotic dynamics, offering a combinatorial and visual language to describe complex horseshoe formations. By systematically relating iterated paperfolding patterns to the observed geometries in multidimensional maps, our approach provides a concrete method for visualizing and classifying the topological features of chaotic sets in dimensions higher than two. This framework offers new insight into the organization of chaos in higher-dimensional discrete dynamical systems.
\end{abstract}

%
%
%
%
%

\section{Introduction}

In a recent work~\cite{Fujioka25}, we derived a sufficient condition for the existence of a topological horseshoe and uniform hyperbolicity in a four-dimensional symplectic map obtained by coupling two classical two-dimensional H\'{e}non maps through linear terms. This coupled system served as a minimal model for higher-dimensional horseshoe dynamics. In the present paper, we extend this approach to explore new classes of horseshoe topologies in three- and four-dimensional settings by constructing and analyzing several variants of coupled H\'{e}non maps. The resulting configurations, which cannot occur in two or lower dimensions, are intended as geometric \emph{templates} for future studies of multidimensional chaotic structures.

The two essential mechanisms underlying the classical Smale horseshoe~\cite{Smale65}—stretching and folding of phase-space volumes—lead to the divergence of nearby trajectories and to mixing via re-injection. In two dimensions, the H\'{e}non map~\cite{Henon76} is the canonical polynomial example exhibiting this mechanism. Its topological horseshoe has been established through both elementary~\cite{Devaney79} and more advanced~\cite{bedford2004real} arguments. Higher-dimensional extensions, particularly coupled H\'{e}non maps~\cite{mao1988standard,baier1990maximum,ding1990algebraic,todesco1994analysis,astakhov2001multistability,anastassiou2018recent,backer2020elliptic}, have been studied to realize hyperchaos~\cite{rossler1979equation} and to reveal dynamical phenomena absent in two dimensions.

In physics and chemistry, interest in higher-dimensional Hamiltonian dynamics has grown steadily. Scenarios with coexisting regular and chaotic orbits arise in diverse contexts such as celestial mechanics, chemical reaction dynamics~\cite{contopoulos2002order,laskar1994large,wiggins1994normally}, and Lagrangian fluid dynamics~\cite{wiggins2005dynamical}. While much attention has been paid to invariant structures in regular regions and their bifurcations, the possible variety of horseshoe configurations in chaotic regions has received far less exploration.

Although Smale originally formulated the horseshoe in arbitrary dimensions~\cite{Smale67}, the classical picture effectively reduces to one-dimensional unstable and stable directions, even in higher-dimensional settings. This restricts the topology to singly folded or multiply folded configurations with creases in the same direction. A natural generalization is to allow foldings along independent crease directions in a multidimensional hypercube—an arrangement that has been qualitatively illustrated~\cite{Wiggins88} or modeled via piecewise linear maps~\cite{zhang2019chaotic}, but not realized in an explicit minimal form.

Here we propose explicit H\'{e}non-type maps that generate such multidimensional foldings. Depending on parameters, these systems produce five distinct horseshoe configurations, each associated with a \emph{paperfolding template} that encodes the dimensionality of the folded ``sheet'' and the choices of crease and stacking directions. These examples are not exhaustive; rather, they serve as representative cases toward building a broader catalogue of multidimensional horseshoe geometries and understanding their dynamical consequences.

The remainder of the paper is organized as follows. Section~\ref{Topological horseshoes via horizontal and vertical slabs} reviews the fundamentals of topological horseshoes, focusing on the geometric framework of horizontal and vertical slabs. Section~\ref{Paperfolding structures} introduces the notion of paperfolding structures and establishes the notation used throughout. Section~\ref{Horseshoe structures in Henon-type maps} presents four H\'{e}non-type maps in three and four dimensions, each exhibiting horseshoe configurations described by specific templates. Finally, Section~\ref{Conclusion} summarizes the findings and outlines directions for future research.

\section{Topological horseshoes via horizontal and vertical slabs}
\label{Topological horseshoes via horizontal and vertical slabs}

This section introduces the geometric and topological structures—namely, horizontal slices, horizontal slabs, and vertical slabs—that underpin our construction of symbolic dynamics. The concepts and results presented here are not new; they are adapted from the classical theory of symbolic dynamics associated with chaotic invariant sets, particularly the Moser-Conley framework. Standard references for these ideas include the monographs of Moser~\cite{Moser01} and Wiggins~\cite{Wiggins88}, where such constructions are rigorously developed in the context of uniformly hyperbolic systems.

Our treatment differs from these classical works in a fundamental way: while Moser and Wiggins impose additional metric conditions—such as Lipschitz conditions of the slices, cone conditions, or the shrinking of diameters of slabs under iteration—to establish conjugacy to a full shift and uniform hyperbolicity, we do not. In this paper, we are primarily concerned with the topological aspects of phase-space geometry. Our goal is to establish the existence of \emph{topological horseshoes}, which suffices to guarantee a \emph{semi-conjugacy} between a subset of the invariant set and the full shift on a finite number of symbols. The absence of metric conditions limits us to topological conclusions, but is better suited to the settings we wish to study, including higher-dimensional systems where uniform hyperbolicity may not hold.

Let \( \mathbb{R}^n = \mathbb{R}^{n_u} \oplus \mathbb{R}^{n_s} \) denote a decomposition of phase space into two complementary linear subspaces: the subspace \( \mathbb{R}^{n_u} \) represents the directions of dominant expansion (referred to as the \emph{horizontal} directions), while \( \mathbb{R}^{n_s} \) represents the directions of dominant contraction (referred to as the \emph{vertical} directions). This splitting is not unique, but is typically chosen to align with the local dynamics of a given map. For instance, when analyzing the behavior of a homeomorphism near a hyperbolic fixed point in \( \mathbb{R}^n \), one may take \( \mathbb{R}^{n_u} \) to be the unstable subspace (tangent to the unstable manifold) and \( \mathbb{R}^{n_s} \) to be the stable subspace (tangent to the stable manifold).

Let \( I^u \subset \mathbb{R}^{n_u} \) and \( I^s \subset \mathbb{R}^{n_s} \) be compact \( n_u \)- and \( n_s \)-dimensional disks, respectively. We define the domain of interest as the product set
\begin{equation}\label{eq:R definition}
R = I^u \times I^s \subset \mathbb{R}^n,
\end{equation}
which is topologically a \emph{hyper-cylinder}: a direct product of two compact Euclidean disks. While the general theory allows \( R \) to be any such product domain, in later sections we shall, for concreteness and simplicity of analysis, often take \( R \) to be a \emph{hypercube}—a rectangular box aligned with the coordinate axes.

\begin{definition}[Horizontal Disk]\label{Horizontal disk definition}
For an arbitrary \(s_0 \in I^s\), define the \emph{horizontal disk at level \(s_0\)} by
\[
d(s_0) = I^u \times \{s_0\}
= \bigl\{ (u,s) \in I^u \times I^s \,\big|\, s = s_0 \bigr\}.
\]
Equivalently, the family \(\{d(s_0)\}_{s_0 \in I^s}\) consists of all horizontal cross-sections of \(R\).
\end{definition}

\begin{definition}[Horizontal Slice]\label{Horizontal slice definition}
A set \( h \subset R \) is called a \emph{horizontal slice of \( R \)} if there exists a continuous function
\[
\psi : I^u \to I^s
\]
such that
\[
h = \left\{ (u, s) \in I^u \times I^s \,\middle|\, s = \psi(u) \right\}.
\]
That is, the slice consists of exactly one point in the vertical direction for each horizontal coordinate.
\end{definition}

\begin{definition}[Horizontal Slab]
A set \( H \subset R \) is called a \emph{horizontal slab of \( R \)} if there exists a continuous set-valued map
\[
\Psi : I^u \to \mathcal{C}(I^s),
\]
such that:
\begin{enumerate}
    \item For each \( u \in I^u \), the fiber \( \Psi(u) \subset I^s \) is a nonempty closed set;
    \item The map \( u \mapsto \Psi(u) \) is continuous with respect to the Hausdorff metric topology on \( \mathcal{C}(I^s) \).
\end{enumerate}
Then the slab \( H \) is defined as
\[
H = \left\{ (u, s) \in I^u \times I^s \,\middle|\, s \in \Psi(u) \right\}.
\]

Here, \( \mathcal{C}(I^s) \) denotes the space of all nonempty closed subsets of \( I^s \), equipped with the \emph{Hausdorff metric}. Given two sets \( A, B \in \mathcal{C}(I^s) \), their Hausdorff distance is defined by
\[
d_H(A, B) = \max\left\{ \sup_{a \in A} \inf_{b \in B} \|a - b\|,\ \sup_{b \in B} \inf_{a \in A} \|b - a\| \right\},
\]
which metrizes the convergence of sets. Continuity of \( \Psi \) means that small changes in \( u \) lead to small (Hausdorff) changes in the fiber \( \Psi(u) \).
\end{definition}

\begin{definition}[Vertical Slab]
A set \( V \subset R \) is called a \emph{vertical slab of \( R \)} if there exists a continuous set-valued map
\[
\Phi : I^s \to \mathcal{C}(I^u),
\]
such that:
\begin{enumerate}
    \item For each \( s \in I^s \), the fiber \( \Phi(s) \subset I^u \) is a nonempty closed set;
    \item The map \( s \mapsto \Phi(s) \) is continuous with respect to the Hausdorff metric on \( \mathcal{C}(I^u) \).
\end{enumerate}
Then the slab \( V \) is defined as
\[
V = \left\{ (u, s) \in I^u \times I^s \,\middle|\, u \in \Phi(s) \right\}.
\]
\end{definition}

\begin{proposition}[Semi-conjugacy to the full shift on \( N \) symbols]\label{Semi-conjugacy to the full shift on N symbols}
Let \( f : \mathbb{R}^n \to \mathbb{R}^n \) be a homeomorphism, and let \( R = I^u \times I^s \subset \mathbb{R}^{n_u} \times \mathbb{R}^{n_s} \) be the compact hypercylinder defined in Eq.~(\ref{eq:R definition}). Suppose that for every \( s \in I^s \), the horizontal disk
\[
d(s) := I^u \times \{s\} \subset R
\]
satisfies
\[
f(d(s)) \cap R = h_1(s) \sqcup h_2(s) \sqcup \dots \sqcup h_N(s),
\]
where each \( h_i(s) \subset R \) is a horizontal slice of \( R \).

Then:
\begin{enumerate}
    \item There exists a nonempty compact invariant set \( \Lambda \subset \bigcap_{n \in \mathbb{Z}} f^n(R) \).
    \item The restricted map \( f|_\Lambda \) is semi-conjugate to the full shift on \( N \) symbols.
\end{enumerate}
\end{proposition}

\begin{proof}
For each \( i = 1, \dots, N \), define the horizontal slab
\[
H_i := \bigcup_{s \in I^s} h_i(s) \subset R.
\]
Each \( H_i \) is a horizontal slab of \( R \), as it consists of horizontal slices varying continuously in \( s \), and the collection \( \{H_i\}_{i=1}^N \) is pairwise disjoint.

We now show that the preimage \( V_i := f^{-1}(H_i) \subset R \) is a vertical slab. For each \( s \in I^s \), since \( h_i(s) \subset f(d(s)) \), we have \( f^{-1}(h_i(s)) \subset d(s) \). Each preimage is compact because \( f \) is a homeomorphism and \( h_i(s) \) is compact. As \( s \) varies, the slices \( h_i(s) \) vary continuously in the Hausdorff metric, and so do their preimages. Define
\[
V_i := \bigcup_{s \in I^s} f^{-1}(h_i(s)).
\]
Then \( V_i \) is a vertical slab of \( R \).

Let \( \Sigma_N := \{1, \dots, N\}^\mathbb{Z} \) be the full shift space. For each finite forward sequence \( (i_0, i_1, \dots, i_k) \), define
\[
V_{i_0, i_1, \dots, i_k} := V_{i_0} \cap f^{-1}(V_{i_1}) \cap \dots \cap f^{-k}(V_{i_k}).
\]
We claim that each \( V_{i_0, \dots, i_k} \) is a vertical slab and nonempty. This follows inductively from the following lemma.

\begin{lemma}
Let \( V \subset V_i \) be a vertical slab and \( a \in \{1, \dots, N\} \). Then \( f^{-1}(V) \cap V_a \) is a vertical slab of \( R \).
\end{lemma}

\begin{proof}
Define
\[
K_a(s) := f^{-1}(V \cap h_a(s)).
\]
Since \( h_a(s) \subset f(d(s)) \), the preimage \( f^{-1}(V \cap h_a(s)) \subset d(s) \). Each \( K_a(s) \) is closed, and the map \( s \mapsto K_a(s) \) varies continuously in the Hausdorff topology. Define
\[
\tilde{V}_a := \bigcup_{s \in I^s} K_a(s) = f^{-1}(V) \cap V_a.
\]
Then \( \tilde{V}_a \subset R \) is a vertical slab.
\end{proof}

Using this lemma inductively, all forward intersections \( V_{i_0, \dots, i_k} \) are nonempty vertical slabs.

Now define, for any backward sequence \( (i_{-1}, \dots, i_{-k}) \),
\[
H_{i_{-1}, \dots, i_{-k}} := H_{i_{-1}} \cap f(H_{i_{-2}}) \cap \dots \cap f^{k-1}(H_{i_{-k}}).
\]
Similar derivations show that all such sets are nonempty horizontal slabs of \( R \).

Fix any bi-infinite sequence \( \mathbf{i} = (\dots, i_{-1}, i_0, i_1, \dots) \in \Sigma_N \). For each \( k \geq 0 \), define
\[
\Lambda_k := V_{i_0, \dots, i_k} \cap H_{i_{-1}, \dots, i_{-k}}.
\]
Each \( \Lambda_k \subset R \) is compact and nonempty, and the sequence \( \{\Lambda_k\} \) is nested. Define
\[
\Lambda_{\mathbf{i}} := \bigcap_{k \geq 0} \Lambda_k.
\]
Then \( \Lambda_{\mathbf{i}} \subset R \) is a nonempty compact set. In general, \( \Lambda_{\mathbf{i}} \) may contain more than one point, since no assumptions are made about the widths of the slabs.

Define the invariant set
\[
\Lambda := \bigcup_{\mathbf{i} \in \Sigma_N} \Lambda_{\mathbf{i}} \subset \bigcap_{n \in \mathbb{Z}} f^n(R).
\]
The map
\[
\pi : \Lambda \to \Sigma_N, \quad \pi(x) = \mathbf{i}\quad \mathrm{ if }\quad x \in \Lambda_{\mathbf{i}}
\]
is well-defined and continuous. Moreover,
\[
\pi(f(x)) = \sigma(\pi(x)),
\]
where \( \sigma \) is the left shift on \( \Sigma_N \). Thus, \( f|_\Lambda \) is semi-conjugate to the full shift on \( N \) symbols.
\end{proof}

\section{Paperfolding structures}
\label{Paperfolding structures}

In this section, we introduce the concept of \emph{paperfolding structures}, which will later serve as geometric templates for describing the configuration of horseshoes in higher-dimensional dynamical systems. Paperfolding structures \cite{Dekking82a,Dekking82b,Dekking82c,Dekking12,Ben-Abraham13,Gahler14}, which have appeared in diverse contexts ranging from origami mathematics to symbolic dynamics and discrete geometry, provide a convenient and intuitive language for encoding recursive folding operations in Euclidean space. 

A classical example arises from the process of repeatedly folding an infinitely long paper strip in half, always folding it in the middle, and then unfolding it to observe the resulting pattern of creases. When viewed from a direction perpendicular to the paper, each crease can be classified as either a valley or a crest, depending on its local geometry. This simple procedure generates a rich combinatorial structure, forming the basis of the one-dimensional theory of paperfolding, as developed in detail in \cite{Dekking82a,Dekking82b,Dekking82c}. Natural extensions of this theory consider the recursive folding of two- or higher-dimensional sheets and the resulting crease patterns that emerge upon unfolding, as explored in \cite{Dekking12,Ben-Abraham13,Gahler14}. 

Our focus here is not on the combinatorial or symbolic aspects of paperfolding, but rather on its geometric realization as a sequence of folding transformations applied to an initial domain.

In the classical study of horseshoe dynamics, one typically begins by specifying a fundamental region \( R \subset \mathbb{R}^n \) and examining its deformation under an iterated map \( f: \mathbb{R}^n \to \mathbb{R}^n \). A generic hyperbolic map \( f \) contracts \( R \) along the stable (vertical) directions, expands it along the unstable (horizontal) directions, thereby flattening it into a quasi-\( n_u \)-dimensional sheet, and then folds this sheet along one or more creases before stacking the folded segments along the stable (vertical) directions. These folding operations are the geometric mechanisms underlying the iconic U-shaped configuration of the Smale horseshoe \cite{Smale65,Smale67}, first introduced in the 1960s.

In higher-dimensional systems, the number of available folding and stacking directions increases, allowing for a rich variety of novel horseshoe configurations beyond the classical planar case. To systematically study and classify such configurations, we propose to use higher-dimensional paperfolding structures as geometric templates. These templates capture the essential features of the deformation process—flattening, folding, and stacking—and will serve as the basis for constructing and analyzing generalized horseshoe maps in the subsequent sections.

Rather than presenting a comprehensive and abstract theory of paperfolding structures, which would require substantial formal development, we adopt an example-based approach to introduce the key geometric ideas. In particular, we will describe a sequence of paperfolding scenarios, each corresponding to a specific dimensional setting that naturally arises in the study of higher-dimensional dynamical systems.

We begin with the simplest and most intuitive case: the folding of a two-dimensional sheet embedded in three-dimensional space. This scenario is both familiar and easy to visualize, as it corresponds to common experiences such as folding a piece of paper along a crease. Despite its simplicity, this configuration already captures the essential operations—flattening, folding, and stacking—that underlie the geometric mechanism of horseshoe formation.

Subsequent examples will build upon this foundation by exploring how a two-dimensional sheet may be folded in four-dimensional space. These higher-dimensional folding structures are no longer directly accessible to visual intuition, but they can still be described analytically and serve as rigorous geometric templates for more complex horseshoe dynamics. Through these examples, we aim to convey how paperfolding structures can be systematically generalized to characterize novel horseshoe configurations in phase spaces of dimension three and higher.

\subsection{General assumptions}

Throughout this section, we impose the following simplifying assumptions on all folding scenarios, regardless of dimension:

\begin{enumerate}
    \item \textbf{Canonical crease directions:} Each folding operation occurs along a crease aligned with one of the coordinate axes—specifically, the \( x \)-, \( y \)-, or \( z \)-axis (or higher-dimensional analogs). Mixed or oblique crease directions are excluded from the present analysis and will be reported elsewhere \cite{Fujioka2026}.

    \item \textbf{Canonical stacking directions:} The stacking of folded segments also occurs along one of the coordinate axes. Arbitrary stacking orientations are not considered.

    \item \textbf{Single-fold constraint:} Each folding operation is applied exactly once along a given crease. Repeated foldings along the same crease are not allowed in the present analysis.

    \item \textbf{Consistent stacking orientation:} For each folding operation, the portion of the domain on the positive side of the crease is stacked above the negative side. For example, if a sheet in the \( (x,y) \)-plane is folded along the crease \( x = 0 \) (i.e., the \( y \)-axis) and stacked along the \( z \)-direction, then the \( x > 0 \) half is stacked on top of the \( x < 0 \) half. “On top” is defined with respect to increasing values in the stacking direction.

    \item \textbf{Elasticity and renormalization:} The domain is assumed to be elastic, and its dimensions are renormalized after each folding operation to avoid metric complications. In particular, the resulting object is rescaled to preserve the original dimensions of the unfolded domain, and its center is assumed to remain at the origin. This allows us to focus purely on the geometric configuration of the folded structure, independent of length scales or coordinate shifts.
\end{enumerate}

\subsection{Folding a one-dimensional sheet in two dimensions}
\label{Folding a one-dimensional sheet in two dimensions}

We begin our discussion of paperfolding structures with the simplest possible case: folding a one-dimensional ``sheet'' (i.e., a line segment) embedded in two-dimensional space. Although this example is elementary, it serves to introduce the notational conventions and geometric intuition that will be extended to higher dimensions throughout this manuscript.

\begin{figure}[thbp]
    \centering
    \includegraphics[width=0.8\linewidth]{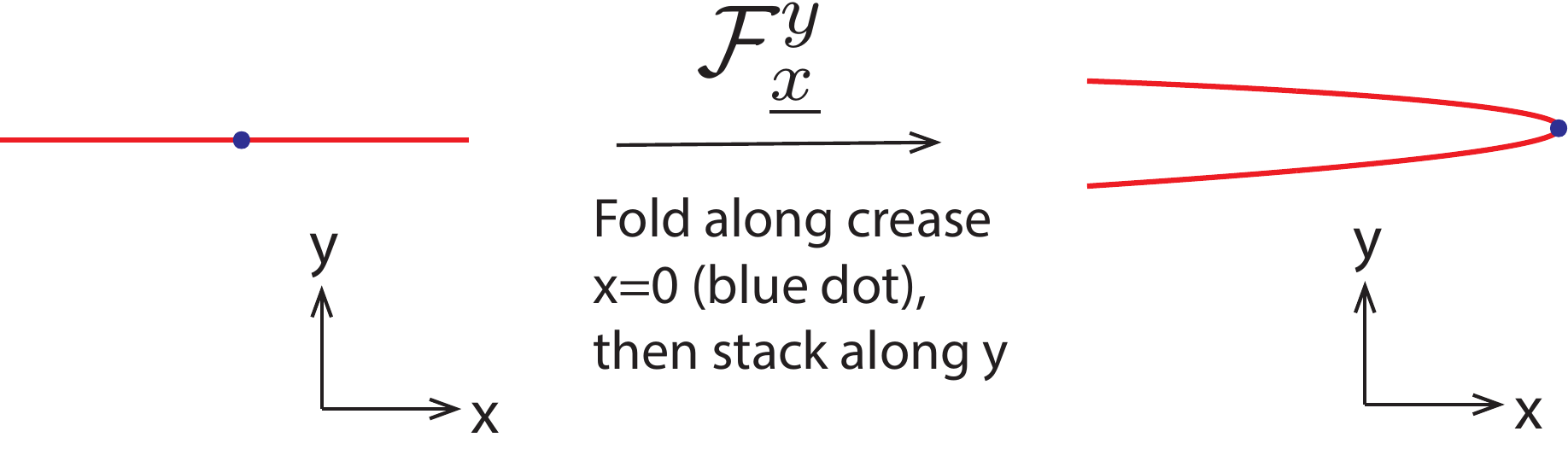} 
   \caption{(Schematic, color online) Folding a one-dimensional ``sheet'' in two dimensions. The sheet (red segment) lies along the \( x \)-axis and is folded about its midpoint (blue dot), which serves as a degenerate form of a crease in higher-dimensional settings. The right half (\( x > 0 \)) is reflected and stacked above the left half (\( x < 0 \)) along the \( y \)-direction.}
    \label{fig:Paperfolding_1D_in_2D}
\end{figure}

In this setting, the sheet is a line segment initially lying along the \( x \)-axis in the \( (x,y) \)-plane. As illustrated in Fig.~\ref{fig:Paperfolding_1D_in_2D}, the folding process consists of folding this segment in half about its midpoint, which we take without loss of generality to be the origin \( x = 0 \). This point represents a degenerate form of a crease in higher-dimensional paperfolding structures. The segment to the right of the crease (\( x > 0 \)) is reflected and stacked on top of the segment to the left (\( x < 0 \)) along the vertical \( y \)-direction.

We denote this folding operation by \( \mathcal{F}^y_{\underline{x}} \). Here, the subscript \( \underline{x} \) indicates that the sheet lies along the \( x \)-axis and that the folding crease is at \( x = 0 \), while the superscript \( y \) specifies that the stacking occurs in the \( y \)-direction. This notational convention—where the subscript indicates the subspace occupied by the sheet and the superscript indicates the stacking direction—will be used consistently throughout the manuscript to describe more general folding operations in higher dimensions.

\subsection{Folding a two-dimensional sheet in three dimensions}
\label{Folding a two-dimensional sheet in three dimensions}

We now turn to the next simplest paperfolding structure: the folding of a two-dimensional sheet embedded in three-dimensional space. This setting captures the essential geometric operations—flattening, folding, and stacking—that underlie the formation of horseshoes in dynamical systems.

Figure~\ref{fig:Paperfolding_2D} illustrates a two-step folding process of a sheet initially lying in the \((x,y)\)-plane, with the \(z\)-axis perpendicular to the sheet. The coordinate frame \((x,y,z)\) is fixed and does not transform with the sheet.

\begin{figure}[thbp]
    \centering
    \includegraphics[width=\linewidth]{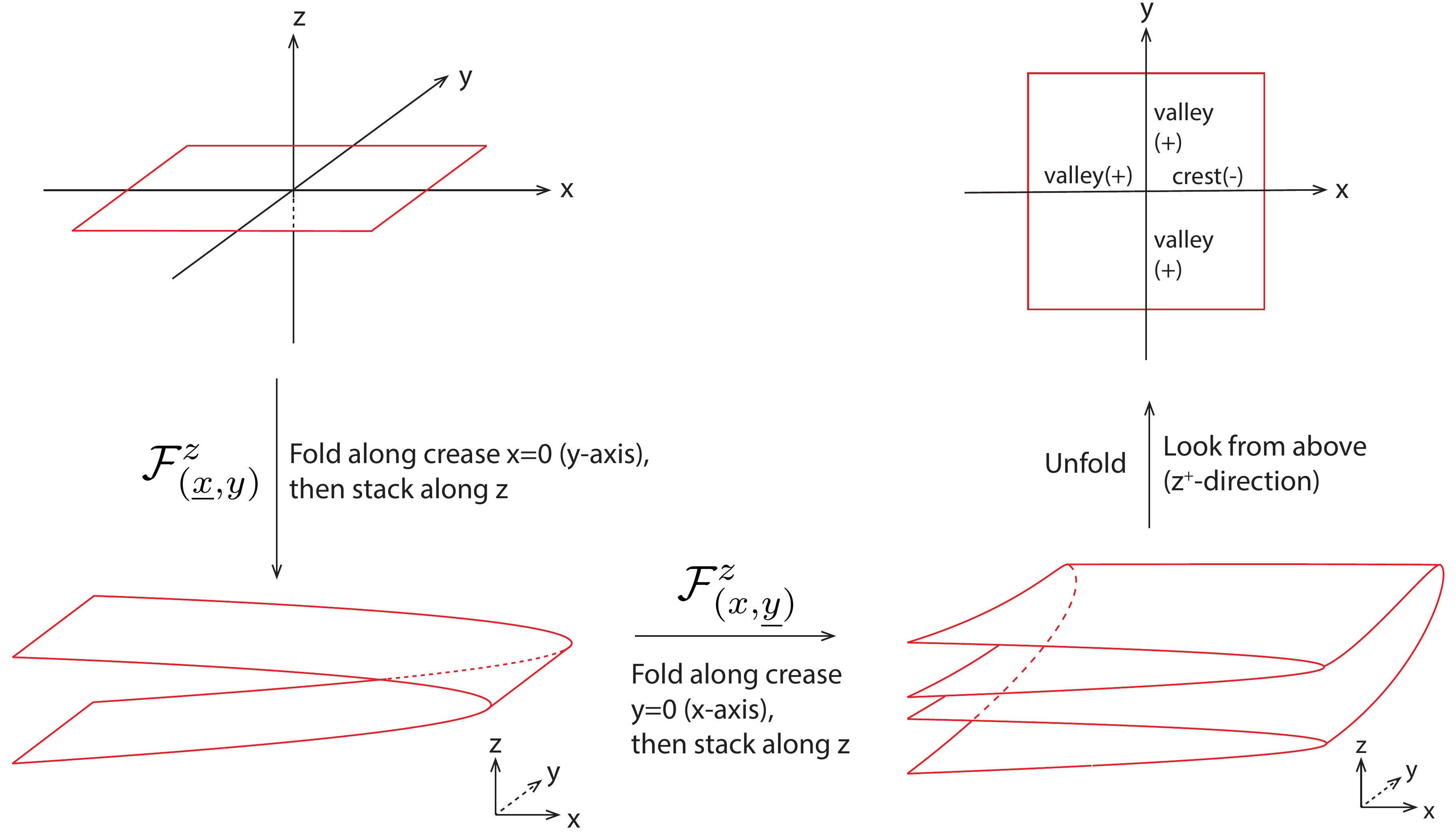} 
    \caption{(Schematic, color online) Folding a two-dimensional sheet in three dimensions. The process consists of two sequential foldings. First, \(\mathcal{F}^z_{(\underline{x},y)}\) folds the sheet with a crease along \(x = 0\) (i.e., the \(y\)-axis), stacking the \(x > 0\) half on top of the \(x < 0\) half along the \(z\)-axis. Second, \(\mathcal{F}^z_{(x,\underline{y})}\) folds the sheet with a crease along \(y = 0\) (i.e., the \(x\)-axis), stacking the \(y > 0\) half above the \(y < 0\) half along \(z\). The resulting configuration, \(\mathcal{F}^z_{(x,\underline{y})} \circ \mathcal{F}^z_{(\underline{x},y)}\), is shown in the lower-right. For visualization purposes, the crease pattern observed upon unfolding and viewed from the \(+z\)-direction is depicted in the upper-right.}
    \label{fig:Paperfolding_2D}
\end{figure}

In the first step, we fold the sheet with a crease along \(x = 0\) (i.e., the \(y\)-axis), stacking the right half (\(x > 0\)) on top of the left half (\(x < 0\)) along the \(z\)-direction. This operation is denoted by \(\mathcal{F}^z_{(\underline{x},y)}\), where the subscript specifies that the sheet lies in the \((x,y)\)-plane and that the folding crease is along the \(x\)-direction, while the superscript \(z\) indicates the stacking direction.

In the second step, the folded configuration is further folded with a crease along \(y = 0\) (i.e., the \(x\)-axis), stacking the top half (\(y > 0\)) over the bottom half (\(y < 0\)) along the same \(z\)-axis. This operation is denoted \(\mathcal{F}^z_{(x,\underline{y})}\). The resulting configuration is given by the composition \(\mathcal{F}^z_{(x,\underline{y})} \circ \mathcal{F}^z_{(\underline{x},y)}\), as shown in the lower-right of Fig.~\ref{fig:Paperfolding_2D}.

Although classical studies of paperfolding (e.g., \cite{Dekking82a,Dekking82b,Dekking82c,Dekking12,Ben-Abraham13,Gahler14}) emphasize the combinatorics of crease patterns obtained upon unfolding, our use of such patterns here is purely illustrative. The crease diagram in the upper-right of Fig.~\ref{fig:Paperfolding_2D}, with valleys marked \(+\) and crests marked \(-\), is included solely to aid visualization and is not the primary object of study in this work.

\subsection{Folding a two-dimensional sheet in four dimensions}
\label{Folding a two-dimensional sheet in four dimensions}

We now introduce a novel paperfolding structure arising in four-dimensional space, namely, a two-dimensional sheet that is folded along two independent crease directions and stacked along two independent stacking directions. To the authors' knowledge, such a configuration—featuring fully independent crease and stacking axes—has not previously appeared in the literature.

\begin{figure}[thbp]
    \centering
    \includegraphics[width=\linewidth]{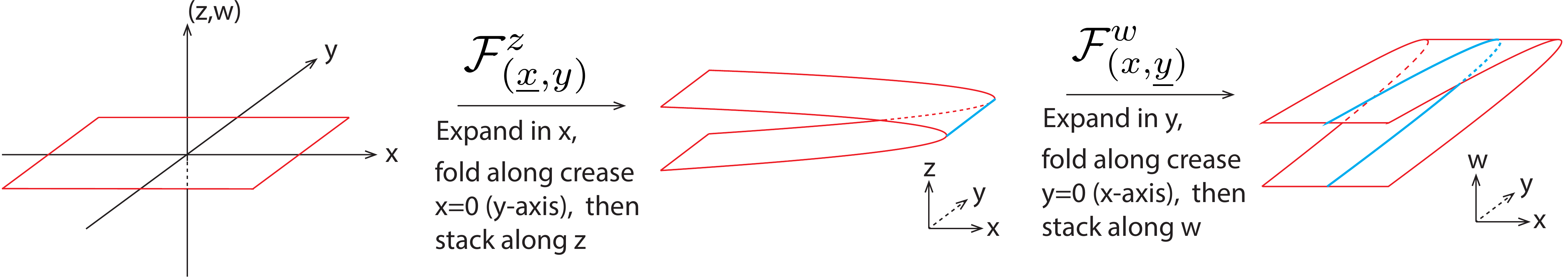} 
    \caption{(Schematic, color online) Folding a two-dimensional sheet in four dimensions. The sheet initially lies in the \( (x,y) \)-plane (the horizontal plane), and the vertical reference plane is taken to be the \( (z,w) \)-plane. The folding process consists of two sequential operations. First, \( \mathcal{F}^z_{(\underline{x},y)} \) folds the sheet with a crease along \( x = 0 \) (i.e., the \( y \)-axis), stacking the \( x > 0 \) half above the \( x < 0 \) half along the \( z \)-axis. Second, \( \mathcal{F}^w_{(x,\underline{y})} \) folds the resulting object with a crease along \( y = 0 \) (i.e., the \( x \)-axis), stacking the \( y > 0 \) half above the \( y < 0 \) half along the \( w \)-axis. The final configuration, denoted by \( \mathcal{F}^w_{(x,\underline{y})} \circ \mathcal{F}^z_{(\underline{x},y)} \), is a doubly folded sheet in four-dimensional space, with independent crease directions and independent stacking directions. In both the middle and right panels, the cyan curve indicates the original crease associated with the first folding operation \( \mathcal{F}^z_{(\underline{x},y)} \).}
    \label{fig:Paperfolding_2D_in_4D}
\end{figure}

Consider a two-dimensional sheet embedded in \(\mathbb{R}^4\), lying in the \((x,y)\)-plane. The horizontal directions are taken to be the \((x,y)\)-plane, while the vertical directions are associated with the \((y,w)\)-plane. The total folding operation consists of two sequential steps, as shown in Fig.~\ref{fig:Paperfolding_2D_in_4D}.

In the first step, the sheet is folded with a crease along \(x = 0\) (i.e., the \(y\)-axis), stacking the right half (\(x > 0\)) above the left half (\(x < 0\)) along the \(z\)-axis. This operation is denoted by \(\mathcal{F}^z_{(\underline{x},y)}\), where the subscript indicates that the sheet lies in the \((x,y)\)-plane and the crease is taken along the \(x\)-direction, and the superscript \(z\) specifies the stacking direction.

In the second step, the resulting configuration is further folded with a crease along \(y = 0\) (i.e., the \(x\)-axis), stacking the upper half (\(y > 0\)) above the lower half (\(y < 0\)) along the \(w\)-axis. This operation is denoted by \(\mathcal{F}^w_{(x,\underline{y})}\).

The final folded configuration is given by the composition \(\mathcal{F}^w_{(x,\underline{y})} \circ \mathcal{F}^z_{(\underline{x},y)}\), as illustrated in Fig.~\ref{fig:Paperfolding_2D_in_4D}. The presence of the additional dimension \(w\) introduces a second, independent stacking direction, thereby enabling a paperfolding structure with two orthogonal crease directions (\(x\) and \(y\)) and two orthogonal stacking axes (\(z\) and \(w\)). Such a configuration is inherently four-dimensional and cannot be realized in spaces of dimension three or lower.

In Sec.~\ref{Type A} we will demonstrate that this type of folding arises naturally as a geometric template for the generalized horseshoe structure of a four-dimensional H\'{e}non-type map.

This four-dimensional folding structure can also be understood from a fiberwise perspective, by decomposing the two-dimensional sheet into one-dimensional fibers and examining how each folding operation acts within lower-dimensional subspaces of $\mathbb{R}^4$. 

More specifically, the first folding operation $\mathcal{F}^z_{(\underline{x},y)}$ acts nontrivially only in the $(x,z)$-subspace, leaving the orthogonal $(y,w)$-subspace unchanged. Therefore, $\mathcal{F}^z_{(\underline{x},y)}$ can be written as a direct product:
\begin{equation}\label{eq:F z xy decomposition}
\mathcal{F}^z_{(\underline{x},y)} = \mathcal{F}^z_{\underline{x}} \times \mathrm{Id}_{(y,w)}\ ,
\end{equation}
where $\mathcal{F}^z_{\underline{x}}$ acts solely on the $(x,z)$-subspace and represents the foloding of each one-dimensional fiber aligned along the $x$-direction, and ${\rm Id}_{(y,w)}$ denotes the identity operator (i.e., no deformation) in the $(y,w)$-subspace. 

Similarly, the second folding operation $\mathcal{F}^w_{(x,\underline{y})}$ acts only in the $(y,w)$-subspace and leaves the $(x,z)$-subspace unaffected. Hence, it can also be decomposed as 
\begin{equation}\label{eq:F w xy decomposition}
\mathcal{F}^w_{(x,\underline{y})} = \mathrm{Id}_{(x,z)} \times \mathcal{F}^w_{\underline{y}}\ ,
\end{equation}
where $\mathcal{F}^w_{\underline{y}}$ acts solely on the $(y,w)$-subspace and represents the folding of each one-dimensional fiber aligned along the $y$-direction, and ${\rm Id}_{(x,z)}$ is the identity on the $(x,z)$-subspace. 

Consequently, the full two-dimensional folding operation can be expressed as the direct product of two one-dimensional foldings:
\begin{equation}\label{eq:folding decomposition into direct products}
\mathcal{F}^w_{(x,\underline{y})} \circ \mathcal{F}^z_{(\underline{x},y)}
= \left( {\rm Id}_{(x,z)} \times \mathcal{F}^w_{\underline{y}} \right) \circ \left( \mathcal{F}^z_{\underline{x}} \times {\rm Id}_{(y,w)} \right)
= \mathcal{F}^z_{\underline{x}} \times \mathcal{F}^w_{\underline{y}} \ .
\end{equation}
In other words, the doubly folded two-dimensional sheet in \( \mathbb{R}^4 \) is equivalent to the Cartesian product of two one-dimensional singly folded sheets: one folded in the \( (x,z) \)-subspace, and the other in the \( (y,w) \)-subspace. 

This direct product decomposition also implies that the two foldings commute:
\begin{equation}\label{eq:foldings commute}
\mathcal{F}^w_{(x,\underline{y})} \circ \mathcal{F}^z_{(\underline{x},y)} = \mathcal{F}^z_{(\underline{x},y)} \circ \mathcal{F}^w_{(x,\underline{y})}\ .
\end{equation}
This is unsurprising, as the two operations act independently on orthogonal subspaces—\( \mathcal{F}^z_{(\underline{x},y)} \) on the \( (x,z) \)-subspace and \( \mathcal{F}^w_{(x,\underline{y})} \) on the \( (y,w) \)-subspace. This perspective will prove useful when analyzing the generalized horseshoe structures of the four-dimensional Hénon-type map in Sec.~\ref{Type A}.

\subsection{Folding a three-dimensional sheet in four dimensions}
\label{Folding a three-dimensional sheet in four dimensions}

\begin{figure}[htbp]
  \centering
  \begin{subfigure}[b]{0.3\textwidth}
    \includegraphics[width=\textwidth]{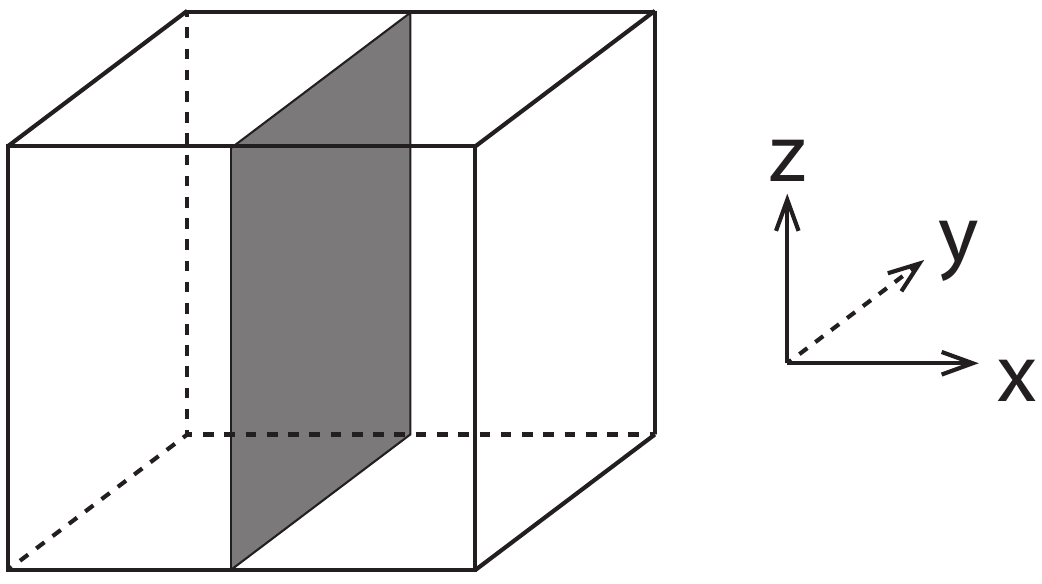}
    \caption{$\mathcal{F}^w_{(\underline{x},y,z)}$}
    \label{fig:3D_Paper_in_4D_1st}
  \end{subfigure}
  \hfill
  \begin{subfigure}[b]{0.3\textwidth}
    \includegraphics[width=\textwidth]{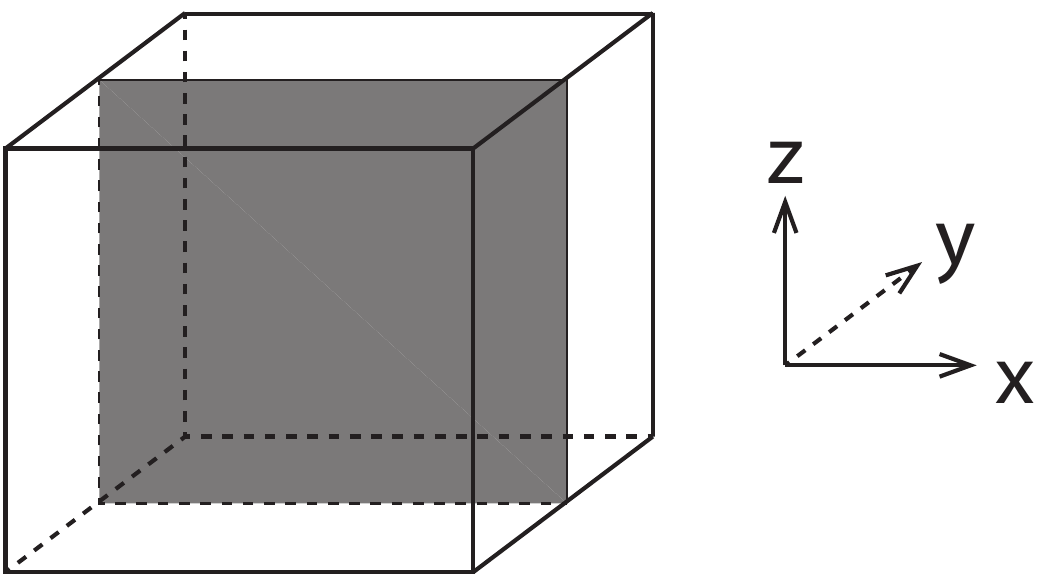}
    \caption{$\mathcal{F}^w_{(x,\underline{y},z)}$}
    \label{fig:3D_Paper_in_4D_2nd}
  \end{subfigure}
  \hfill
  \begin{subfigure}[b]{0.3\textwidth}
    \includegraphics[width=\textwidth]{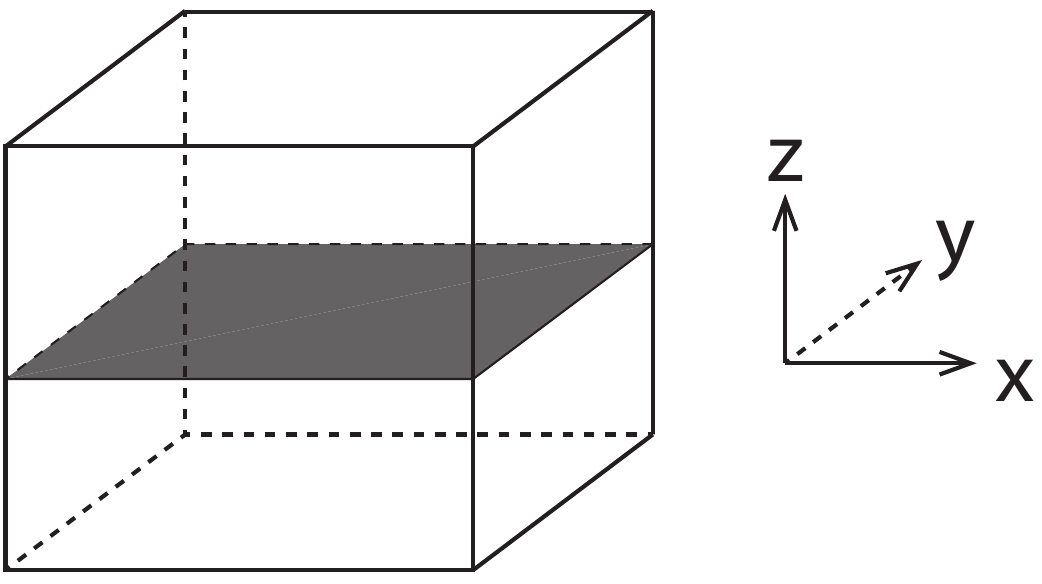}
    \caption{$\mathcal{F}^w_{(x,y,\underline{z})}$}
    \label{fig:3D_Paper_in_4D_3rd}
  \end{subfigure}
  \caption{Folding a three-dimensional sheet (represented as a cube in the $(x,y,z)$-subspace) embedded in four dimensions. The fourth coordinate axis, $w$, is not shown. (a): The operation $\mathcal{F}^w_{(\underline{x},y,z)}$ folds the sheet along the crease $x=0$ (shaded) and stacks the resulting halves along the $w$-direction (not shown). (b): $\mathcal{F}^w_{(x,\underline{y},z)}$ folds the sheet along the crease $y=0$ (shaded) and stacks the resulting halves along $w$. (c): $\mathcal{F}^w_{(x,y,\underline{z})}$ folds the sheet along the crease $z=0$ (shaded) and stacks the resulting halves along $w$. }
  \label{fig:3D_Paper_in_4D}
\end{figure}

In this subsection, we consider a three-dimensional “sheet” lying in the \((x,y,z)\)-subspace and embedded in four-dimensional space \((x,y,z,w)\). In this setting, the only available stacking direction is the \(w\)-axis. We examine a triply folded configuration, illustrated in Fig.~\ref{fig:3D_Paper_in_4D}, consisting of three consecutive folding operations.

In the first step (Fig.~\ref{fig:3D_Paper_in_4D_1st}), the sheet is folded with a crease along \(x = 0\) (i.e., the \((y,z)\)-plane), stacking the \(x > 0\) half above the \(x < 0\) half along the \(w\)-direction. This operation is denoted by \( \mathcal{F}^w_{(\underline{x},y,z)} \), where the subscript indicates that the sheet lies in the \((x,y,z)\)-hyperplane and the folding crease is along the \(x = 0\) plane, while the superscript \(w\) specifies the stacking direction.

In the second step (Fig.~\ref{fig:3D_Paper_in_4D_2nd}), the folded configuration is further folded with a crease along \(y = 0\) (i.e., the \((x,z)\)-plane), stacking the \(y > 0\) half above the \(y < 0\) half along the \(w\)-axis. This operation is denoted by \( \mathcal{F}^w_{(x,\underline{y},z)} \).

In the third step (Fig.~\ref{fig:3D_Paper_in_4D_3rd}), the configuration is folded again, this time with a crease along \(z = 0\) (i.e., the \((x,y)\)-plane), stacking the \(z > 0\) half above the \(z < 0\) half along the \(w\)-axis. This final operation is denoted by \( \mathcal{F}^w_{(x,y,\underline{z})} \).

The complete triply folded configuration is thus given by the composition:
\[
\mathcal{F}^w_{(x,y,\underline{z})} \circ \mathcal{F}^w_{(x,\underline{y},z)} \circ \mathcal{F}^w_{(\underline{x},y,z)}.
\]

Due to the limitations of visualizing four-dimensional structures using three-dimensional projections, we depict only the crease planes for each folding operation in Fig.~\ref{fig:3D_Paper_in_4D}, omitting the stacked layers. This construction yields a triply folded paperfolding structure in four dimensions, characterized by three independent crease directions (the three shaded planes) and a common stacking direction along \(w\). Notably, each crease is now a two-dimensional surface in four-dimensional space—an inherently high-dimensional feature not possible in dimensions less than four.

In Sec.~\ref{f IV}, we will present a concrete example of a four-dimensional Hénon-type map that exhibits a horseshoe structure describable by this triply folded template.

\section{Horseshoe structures in H\'{e}non-type maps}
\label{Horseshoe structures in Henon-type maps}

Before introducing specific Hénon-type maps and their associated invariant sets, we briefly clarify our use of the term \emph{horseshoe} throughout this section. In the classical literature, the term “Smale horseshoe” may refer either to a topologically defined invariant set semi-conjugate to a full shift, or to the geometric image of a square region under a stretching-and-folding map. In this work, we adopt the latter, more geometric perspective: we use the term \emph{horseshoe} to refer to the qualitative shape of \( f(R) \cap R \), where \( R \subset \mathbb{R}^n \) is a chosen fundamental domain and \( f \) is the dynamical map. Specifically, we are interested in whether \( f(R) \) resembles a singly, doubly, or triply folded sheet relative to \( R \), as described by the paperfolding structures introduced in Section~\ref{Paperfolding structures}. The emphasis is thus placed on the configuration of folds—that is, the geometric arrangement of \( f(R) \) with respect to \( R \). The symbolic dynamics and invariant set structure are understood as direct consequences of this underlying geometric structure, which is captured by the corresponding paperfolding templates.

In the following subsections, we introduce a series of H\'{e}non-type maps defined in three- and four-dimensional spaces, each exhibiting a horseshoe structure that corresponds to one of the paperfolding templates developed in Section~\ref{Paperfolding structures}. The simplest three-dimensional example is a trivial extension of the original two-dimensional H\'{e}non map proposed in \cite{Henon76}, while the remaining three- and four-dimensional examples are constructed as compositions or coupled systems derived from the two-dimensional case. These higher-dimensional constructions are designed to realize folding configurations that align with the geometric templates introduced earlier, thereby establishing a direct correspondence between the algebraic structure of the map and the underlying phase-space geometry.

\subsection{Map \( f_{\mathrm{I}} \): singly folded sheet in three dimensions}
\label{f I}

We begin with the simplest higher-dimensional extension of the classical Hénon map: a three-dimensional map \( f_{\mathrm{I}} \) obtained by trivially embedding the original two-dimensional system into \(\mathbb{R}^3\). Despite its simplicity, this map already exhibits a nontrivial horseshoe structure under suitable parameter regimes. We show that the resulting geometry of \( f_{\mathrm{I}}(R) \cap R \) corresponds to a singly folded paperfolding configuration, specifically of the form \( \mathcal{F}^z_{(\underline{x},y)} \), as introduced in Section~\ref{Folding a two-dimensional sheet in three dimensions}.

\begin{figure}[th]
        \centering
        \includegraphics[width=0.6\linewidth]{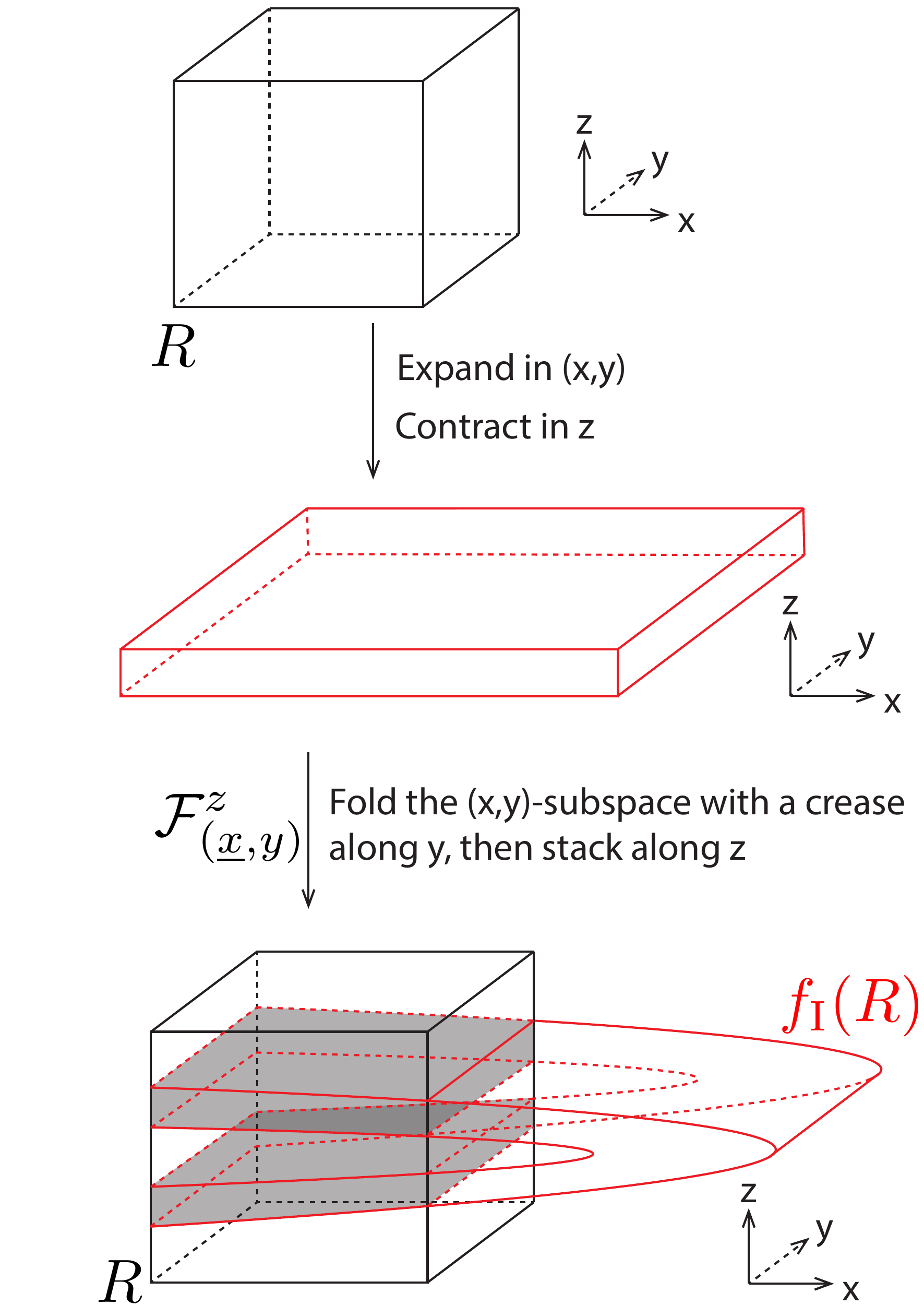} 
        \caption{(Schematic, color online) Horseshoe generated by \( f_{\mathrm{I}} \). Starting from the cube \( R \), the map expands \( R \) in the unstable directions (the \( (x,y) \)-plane) and contracts it in the stable direction (the \( z \)-axis), effectively flattening it into a quasi-two-dimensional sheet lying in the \( (x,y) \)-plane. The sheet is then folded with a crease along \( x = 0 \) (i.e., the \( y \)-direction) and stacked along the \( z \)-axis. The resulting intersection \( R \cap f(R) \) consists of two disjoint horizontal slabs (shaded region).}
 \label{fig:3D_singly_folded_horseshoe}
\end{figure}

Consider the H\'{e}non-type map of the form 
\begin{equation}\label{eq:3D singly folded Henon}
\left( \begin{array}{ccc}
x_{n+1}\\
y_{n+1} \\
z_{n+1} \end{array} \right) = 
f_{\rm I} \left( \begin{array}{ccc}
x_{n}\\
y_{n} \\
z_{n} \end{array} \right) =
\left( \begin{array}{ccc}
a_0 - x^2_n - z_n\\
b y_n \\
x_n \end{array} \right)
\end{equation}
where $(x_n,y_n,z_n)^\mathsf{T}$ denotes the position of the $n$th iteration, and the parameters $a_0$, $b$ satisfy  
\begin{eqnarray}
& a_0  > 5+2\sqrt{5} \nonumber \\
& b > 1   \nonumber
\end{eqnarray}
where the bound on $a_0$ is obtained by Devaney and Nitecki in \cite{Devaney79} to realize the horseshoe in the two-dimensional H\'{e}non map, and the bound on $b$ guarantees uniform expansion in $y$. It is trivial to see that the dynamics in $y$ is a constant uniform expansion uncoupled from $(x,z)$. Therefore, on every $(x,z)$-slice we have the two-dimensional H\'{e}non map
\begin{equation}\label{eq:2D Henon}
\left( \begin{array}{ccc}
x_{n+1}\\
z_{n+1} \end{array} \right) = 
\left( \begin{array}{ccc}
a_0 - x^2_n - z_n\\
x_n
 \end{array} \right)\ ,
\end{equation}
which is the one originally proposed in \cite{Henon76} and studied in detail in \cite{Devaney79}. Generalizing the results established by \cite{Devaney79} to the three-dimensional map $f_{\rm I}$, it is straightforward to see that given $a_0>5+2\sqrt{5}$, a cube $R$ can be identified as
\begin{equation}\label{eq:Devaney V definition}
R = \left\lbrace (x,y,z) \middle|\  |x|,|y|,|z|\leq r \right\rbrace\ , ~~~~{\rm where\ }r=1+\sqrt{1+a_0}
\end{equation}
such that $f_{\rm I}(R)$ gives rise to a horseshoe. As illustrated in Fig.~\ref{fig:3D_singly_folded_horseshoe}, the map expands the region \( R \) in the unstable directions (the \( (x,y) \)-plane), contracts it in the stable direction (the \( z \)-axis), and thereby flattens it into a quasi-two-dimensional sheet lying in the \( (x,y) \)-plane. The sheet is then folded with a crease in the \( y \)-direction and stacked along the \( z \)-axis. The resulting intersection \( R \cap f_{\mathrm{I}}(R) \) consists of two disjoint horizontal slabs. This horseshoe is geometrically described by the paperfolding template \( \mathcal{F}^z_{(\underline{x},y)} \), corresponding to a single fold along the \( x = 0 \) crease and stacking in the \( z \)-direction.

The fact that the horseshoe is templated by \( \mathcal{F}^z_{(\underline{x},y)} \) implies that the invariant set is conjugate to a full shift on two symbols. More generally, once the folding configuration is identified and matched to a known paperfolding template, the symbolic dynamics and topological type of the invariant set follow directly—at least up to topological conjugacy or semi-conjugacy. Thus, the determination of the template type is sufficient to characterize the symbolic structure of the underlying dynamics.

\subsection{Map \( f_{\mathrm{II}} \): doubly folded sheet in three dimensions}
\label{f II}

We now consider a nontrivial three-dimensional generalization of the Hénon map, denoted \( f_{\mathrm{II}} \), designed to exhibit more intricate folding behavior. In particular, we demonstrate that for appropriate parameter choices, the image \( f_{\mathrm{II}}(R) \) undergoes two sequential folding operations, leading to a doubly folded horseshoe structure. This configuration is well-described by the paperfolding template \( \mathcal{F}^z_{(x,\underline{y})} \circ \mathcal{F}^z_{(\underline{x},y)} \), representing a two-step folding process in three-dimensional space.

\begin{figure}
        \centering
        \includegraphics[width=0.7\linewidth]{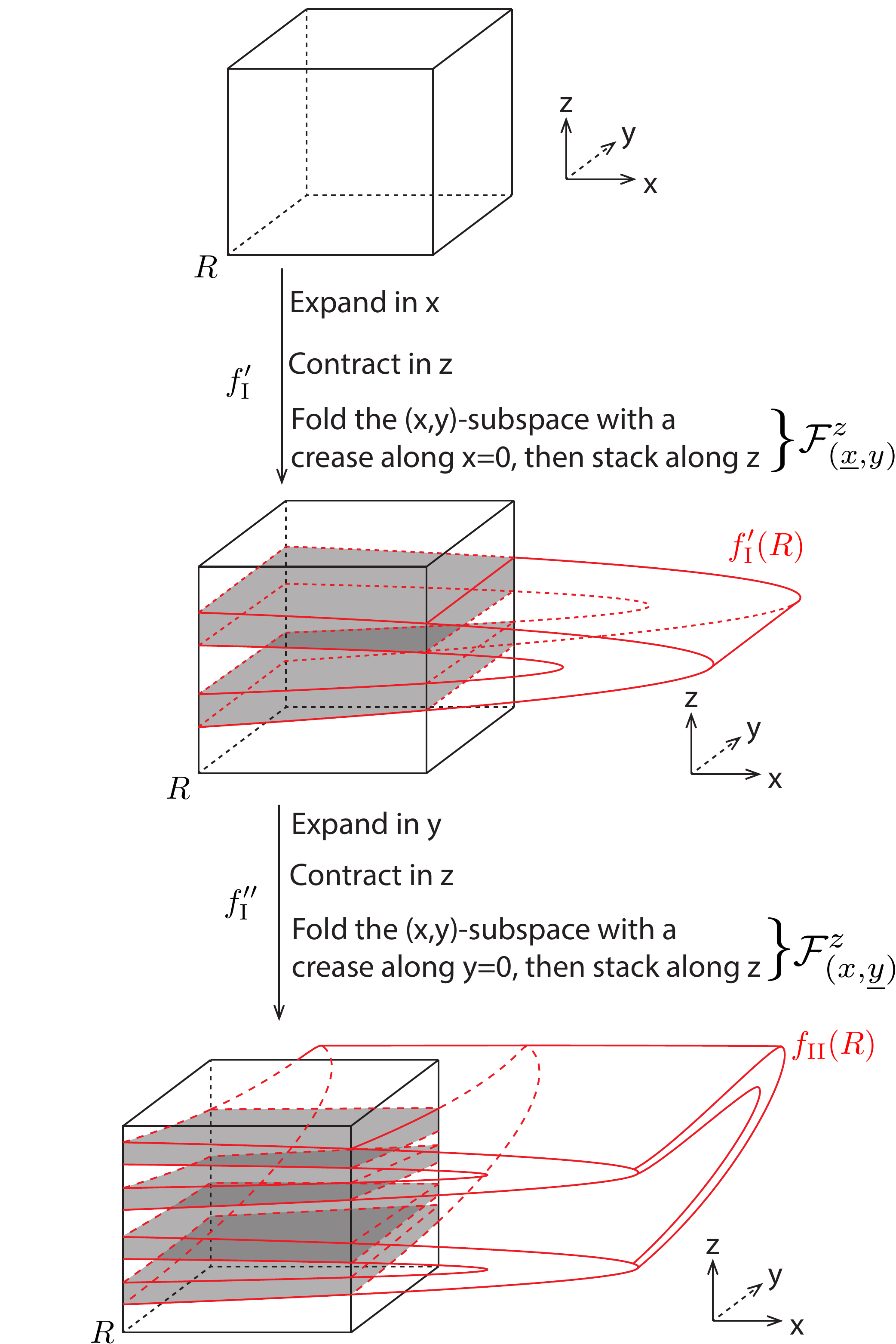} 
       \caption{
(Schematic, color online) Horseshoe generated by \( f_{\rm II} \). Starting from the cube \( R \), the map \( f'_{\rm I} \) expands \( R \) in the \( x \)-direction and contracts it in the \( z \)-direction, effectively flattening it into a quasi-two-dimensional sheet lying in the \( (x,y) \)-plane. It then folds the sheet with a crease along \( x = 0 \) (i.e., the \( y \)-axis) and stacks along the \( z \)-axis. This operation is modeled by \( \mathcal{F}^z_{(\underline{x},y)} \). Subsequently, the map \( f''_{\rm I} \) expands the resulting structure in the \( y \)-direction, contracts it in \( z \), folds it with a crease along \( y = 0 \) (i.e., the \( x \)-axis), and stacks along \( z \) again. This second operation is modeled by \( \mathcal{F}^z_{(x,\underline{y})} \). The total transformation \( f_{\rm II} = f''_{\rm I} \circ f'_{\rm I} \) is thus described by the composition \( \mathcal{F}^z_{(x,\underline{y})} \circ \mathcal{F}^z_{(\underline{x},y)} \), forming a doubly folded structure. The intersection \( R \cap f_{\rm II}(R) \) consists of four disjoint horizontal slabs, as shown in the bottom part of the figure.
}
 \label{fig:3D_doubly_folded_horseshoe}
\end{figure}

Consider the H\'{e}non-type map
\begin{equation}\label{eq:3D doubly folded Henon}
\left( \begin{array}{ccc}
x_{n+1}\\
y_{n+1} \\
z_{n+1} \end{array} \right) = 
f_{\rm II} \left( \begin{array}{ccc}
x_{n}\\
y_{n} \\
z_{n} \end{array} \right) =
\left( \begin{array}{ccc}
a_0 - x^2_n - z_n\\
a_1 - y^2_n - x_n\\
y_n \end{array} \right)
\end{equation}
with parameters $a_0,a_1 > 5+2\sqrt{5}$. The inverse map $f^{-1}_{\rm II}$ is slightly complicated as it involves quartic terms:
\begin{eqnarray}
 \left( \begin{array}{ccc}
x_{n-1}\\
y_{n-1} \\
z_{n-1} \end{array} \right) & = f^{-1}_{\rm II} \left( \begin{array}{ccc}
x_{n}\\
y_{n} \\
z_{n} \end{array} \right) \nonumber \\
& =
\left( \begin{array}{ccc}
-y_n - z^2_n + a_1\\
z_n\\
-x_n - y^2_n -2 y_n z^2_n - z^4_n + a_0 + 2 a_1 y_n + 2 a_1 z^2_n - a^2_1 \end{array} \right). \label{eq:f II inverse}
\end{eqnarray}

In this setting, the horizontal directions are defined by the \( (x,y) \)-subspace, and the vertical direction is the \( z \)-axis. Accordingly, the unstable and stable coordinate disks used in the definition of \( R \) in Eq.~(\ref{eq:R definition}) are given by \( I^u = [-r,r] \times [-r,r] \) and \( I^s = [-r,r] \), respectively, where $r$ is a suitably chosen length parameter to be specified later. The resulting domain $R$ is the three-dimensional cube $[-r,r]^3$. 

The map \( f_{\mathrm{II}} \) can be written as the composition of two mappings, denoted \( f'_{\mathrm{I}} \) and \( f''_{\mathrm{I}} \), 
\begin{equation}\label{eq:3D doubly folded Henon compound}
f_{\rm II} = f^{\prime\prime}_{\rm I} \circ f^{\prime}_{\rm I}\ ,
\end{equation}
where $f^{\prime}_{\rm I}$ is identical to \( f_{\mathrm{I}} \) with the parameter choice \( b = 1 \),
\begin{equation}\label{eq:f prime I}
\left( \begin{array}{ccc}
x^{\prime}\\
y^{\prime} \\
z^{\prime} \end{array} \right) = 
f^{\prime}_{\rm I} \left( \begin{array}{ccc}
x\\
y \\
z \end{array} \right) =
\left( \begin{array}{ccc}
a_0 - x^2 - z\\
y \\
x \end{array} \right)\ .
\end{equation}
Thus $f^{\prime}_{\rm I}$ produces a qualitatively similar deformation of the domain \( R \). 

As illustrated in the middle panel of Fig.~\ref{fig:3D_doubly_folded_horseshoe}, the action of \( f'_{\mathrm{I}} \) is to expand \( R \) along the \( x \)-direction, contract it along the \( z \)-axis, flattening it into a quasi-two-dimensional sheet in the \( (x,y) \)-plane, then fold this sheet with a crease along \( x = 0 \) (i.e., the \( y \)-axis), and stack it along the \( z \)-direction. This transformation is modeled by the paperfolding operation \( \mathcal{F}^z_{(\underline{x},y)} \).

The map \( f''_{\mathrm{I}} \) is obtained from \( f'_{\mathrm{I}} \) by interchanging the roles of the \( x \)- and \( y \)-axes:
\begin{equation}\label{eq:f prime prime I}
\left( \begin{array}{ccc}
x^{\prime\prime}\\
y^{\prime\prime} \\
z^{\prime\prime} \end{array} \right) = 
f^{\prime\prime}_{\rm I} \left( \begin{array}{ccc}
x^{\prime}\\
y^{\prime} \\
z^{\prime} \end{array} \right) =
\left( \begin{array}{ccc}
x^{\prime}\\
a_1 - (y^{\prime})^2 - z^{\prime} \\
y^{\prime} \end{array} \right)\ .
\end{equation}
That is, \( f''_{\mathrm{I}} \) expands along \( y \), contracts along \( z \), folds with a crease along \( y = 0 \) (i.e., the \( x \)-axis), and stacks along \( z \). See the bottom part of Fig.~\ref{fig:3D_doubly_folded_horseshoe} for an illustration. This action corresponds to the paperfolding operation \( \mathcal{F}^z_{(x,\underline{y})} \).

Consequently, the composition \( f_{\mathrm{II}} = f''_{\mathrm{I}} \circ f'_{\mathrm{I}} \) performs two successive folds along orthogonal directions, and is described by the template \( \mathcal{F}^z_{(x,\underline{y})} \circ \mathcal{F}^z_{(\underline{x},y)} \), as illustrated in the lower-right part of Fig.~\ref{fig:Paperfolding_2D}.

Moreover, since every horizontal disk \( d(s) \subset R \) (with \( s \in I^s \)) is parallel to the \( (x,y) \)-plane, its image \( f_{\mathrm{II}}(d(s)) \) undergoes the same qualitative deformation and is likewise modeled by the template \( \mathcal{F}^z_{(x,\underline{y})} \circ \mathcal{F}^z_{(\underline{x},y)} \). This has the important consequence that \( f_{\mathrm{II}}(d(s)) \) must intersect \( R \) along four disjoint horizontal slices, each arising from one of the four vertically stacked layers produced by the doubly folded configuration. These layers correspond to those illustrated in the lower-right part of Fig.~\ref{fig:Paperfolding_2D}. Since \( f_{\mathrm{II}}(R) \) can be viewed as \( f_{\mathrm{II}}(d(s)) \) thickened along the vertical direction, it must likewise intersect \( R \) along four disjoint horizontal slabs, as shown in Fig.~\ref{fig:3D_doubly_folded_horseshoe}.

To place the above qualitative description on a more rigorous foundation, we now proceed to show that \( f_{\mathrm{II}}(d(s)) \) indeed intersects \( R \) in four disjoint horizontal slices. By Proposition~\ref{Semi-conjugacy to the full shift on N symbols}, this guarantees the existence of a topological horseshoe and establishes a semi-conjugacy between the invariant set and the full shift on four symbols.

\begin{theorem}\label{3D Doubly folded horseshoe topology}
Let the parameters of \( f_{\rm II} \) be \( a_0 = a_1 = a > 5 + 2\sqrt{5} \), and define \( r = 1 + \sqrt{1 + a} \). Let \( I^u = [-r, r] \times [-r, r] \) be the horizontal disk in the \( (x, y) \)-subspace, and let \( I^s = [-r, r] \) be the vertical interval in the \( z \)-direction. Define the fundamental region \( R = I^u \times I^s \subset \mathbb{R}^3 \).

Then, for every horizontal disk \( d(s) = I^u \times \{ s \} \), where \( s \in I^s \), the intersection \( f_{\rm II}(d(s)) \cap R \) consists of four disjoint horizontal slices of \( R \), that is,
\begin{equation}
f_{\rm II}(d(s)) \cap R = h_1(s) \sqcup h_2(s) \sqcup h_3(s) \sqcup h_4(s)\ ,
\end{equation}
where each \( h_i(s) \) is a horizontal slice of \( R \).
\end{theorem}

\begin{proof}
Let $d(c)$ denote the horizontal disk of $R$ at height $z=c$:
\begin{equation}\label{eq:horizontal disk f II}
d(c) = \left\lbrace\ (x, y, z) \in \mathbb{R}^3\ \middle|\ |x|, |y| \leq r,\ z = c\ \right\rbrace,
\end{equation}
where $|c|\leq r$. To prove the theorem, it suffices to show that for every $|c|\leq r$, the image $f_{\rm II}(d(c))$ intersects $R$ in the horizontal direction along four disjoint slice; that is, the intersection $f_{\rm II}(d(c))\cap R$ consists of four disjoint horizontal slices of $R$. 

Using the identity $f^{-1}_{\rm II}(f_{\rm II}(d(c)))=d(c)$ we obtain an analytic expression for $f_{\rm II}(d(c))$
\begin{eqnarray}
&f_{\rm II}(d(c)) \nonumber \\ 
&=\left\lbrace (x,y,z)\in \mathbb{R}^3 \middle| \begin{array}{ccc}
|z| \leq r \\
|-y - z^2 + a| \leq r \\
-x - y^2 -2 y z^2 - z^4 + a + 2 a (y + z^2) - a^2 =c  \end{array} \right\rbrace\ . \nonumber \\
\label{eq:f II D} 
\end{eqnarray}

The expression for $R \cap f_{\rm II}(d(c))$ is then obtained trivially by imposing the additional bounds on $x$ and $y$
\begin{eqnarray}
&R\cap f_{\rm II}(d(c)) \nonumber \\ 
&=\left\lbrace (x,y,z)\in \mathbb{R}^3 \middle| \begin{array}{ccc}
|x|,|y|,|z| \leq r \\
|-y - z^2 + a| \leq r \\
-x - y^2 -2 y z^2 - z^4 + a + 2 a (y + z^2) - a^2 =c  \end{array} \right\rbrace\ . \nonumber \\
\label{eq:V intersect f II D} 
\end{eqnarray}

Let $\Sigma^2(x)$ be the $(y,z)$-plane at fixed $x$ where the superscript ``$2$'' indicates the dimensionality of the plane. That is,
\begin{equation}\label{eq: y z slice for f_II}
\Sigma^2(x) = \left\lbrace (x',y',z')\in\mathbb{R}^3 \middle| y',z'\in \mathbb{R},\ x'=x \right\rbrace\ .
\end{equation}

Moreover, define the line segments
\begin{equation}
S^{\pm}_y = \left\lbrace (y,z)\in \mathbb{R}^2 \middle| y=\pm r, |z| \leq r \right\rbrace
\end{equation}
so that $S^+_y$ and $S^-_y$ correspond to the right and left boundaries of $R$, respectively, within each $\Sigma^2(x)$ slice, as illustrated in Fig.~\ref{fig:3D_dooubly_folded_horseshoe_y_z_slice_Gammas}.

To prove that for every \( |c| \leq r \), the intersection \( R \cap f_{\rm II}(d(c)) \) consists of four disjoint horizontal slices of \( R \), it suffices to verify that, for all \( |c|, |x| \leq r \), the restriction \( R \cap f_{\rm II}(d(c))\big|_{\Sigma^2(x)} \) consists of four disjoint horizontal slices of \( R\big|_{\Sigma^2(x)} \). Here, the horizontal direction within \( \Sigma^2(x) \) is the \( y \)-direction, which is inherited from the global horizontal direction—namely, the \((x,y)\)-plane—of the ambient space \( \mathbb{R}^3 \).

\begin{figure}
        \centering
        \includegraphics[width=0.6\linewidth]{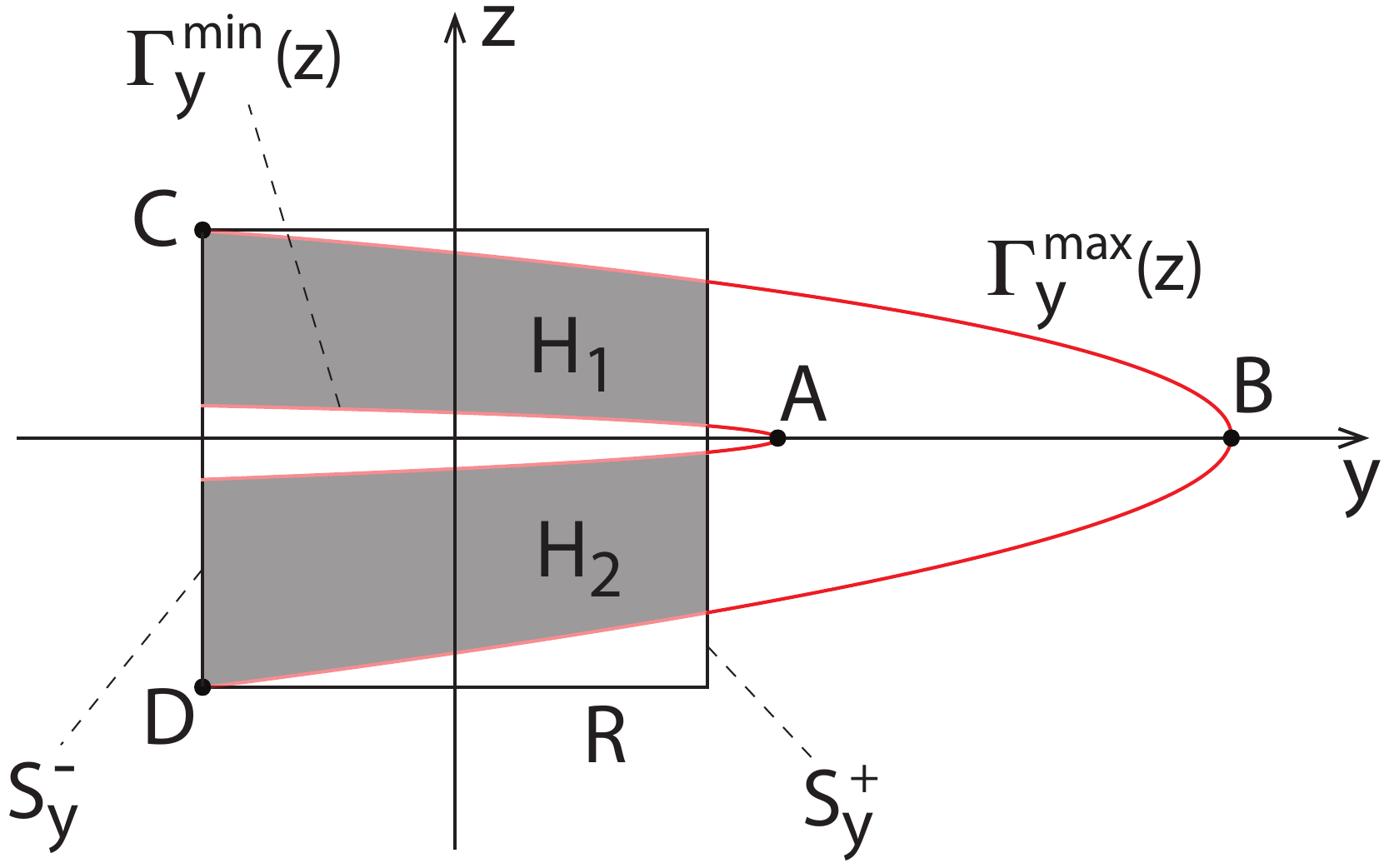} 
        \caption{The region between the graphs \(y=\Gamma_y^{\max}(z)\) and \(y=\Gamma_y^{\min}(z)\) (red) intersects the box \(R\) in two horizontal strips \(H_1\) and \(H_2\) (shaded).}   \label{fig:3D_dooubly_folded_horseshoe_y_z_slice_Gammas}
\end{figure}

The second row of the right-hand side (RHS) of Eq.~(\ref{eq:V intersect f II D}) can be rewritten in the parameterized form
\begin{equation}\label{eq:y z parabola 1 f_II}
y = -z^2 + a + s, ~~~~{\rm where}~ |s| \leq r \ .
\end{equation}
Let $\Gamma_y(z,s)$ be the family of parabolas in $\Sigma^2(x)$
\begin{equation}\label{eq:Gamma y in terms of z s f_II}
\Gamma_y(z,s) = -z^2 + a + s
\end{equation}
where $s$ is viewed as a parameter within range $|s| \leq r$. It is obvious that $\Gamma_y(z,s)$ is bounded by
\begin{equation}\label{eq:Gamma y in terms of z s bounds}
\Gamma^{\min}_y(z) \leq \Gamma_y(z,s) \leq \Gamma^{\max}_y(z)
\end{equation}
with lower and upper bounds
\begin{eqnarray}
& \Gamma^{\min}_y(z) = \Gamma_y(z,s)|_{s=-r} = -z^2 + a -r  \label{eq:Gamma y z min} \\
& \Gamma^{\max}_y(z) =  \Gamma_y(z,s)|_{s=r} = -z^2 + a + r  \label{eq:Gamma y z max}\ .
\end{eqnarray}

The possible location of $R\cap f_{\rm II}(d(c))$ can be narrowed down by establishing the following facts:
\begin{itemize}
\item[(a)] The vertex of $\Gamma^{\min}_y(z)$ is located on the right-hand side of $S^+_y$, as labeled by $A$ in Fig.~\ref{fig:3D_dooubly_folded_horseshoe_y_z_slice_Gammas};
\item[(b)] $\Gamma^{\max}_y(z)$ intersects $S^-_y$ at two points, as labeled by $C$ and $D$ in Fig.~\ref{fig:3D_dooubly_folded_horseshoe_y_z_slice_Gammas}. 
\end{itemize} 

To establish (a), let $A=(y_A,z_A)$. It can be solved easily that
\begin{equation}
z_A=0, ~~~~ y_A = \Gamma^{\min}_y(z_A=0)=a-r. 
\end{equation}
Using the assumption that $a>5+2\sqrt{5}$, it is straightforward to verify that $a-2r>0$, thus $y_A >r$, i.e., $A$ is on the right-hand side of $S^{+}_y$. 

To establish (b), notice that
\begin{equation}
\Gamma^{\max}_y(z=\pm r) = -r^2 +a +r = -r\ ,
\end{equation}
thus $C$ and $D$ are located at
\begin{equation}
C = (-r,r), ~~~~ D=(-r,-r)\ ,
\end{equation}
i.e., $C$ and $D$ are the upper-left and lower-left corners of $R$, respectively, as labeled in Fig.~\ref{fig:3D_dooubly_folded_horseshoe_y_z_slice_Gammas}. Therefore, $\Gamma^{\max}_y(z)$ intersects $S^-_y$ at its two endpoints. 

Combining (a) and (b), we know that the region bounded by $\Gamma^{\max}_y(z)$, $\Gamma^{\min}_y(z)$, and $S^{\pm}_y$ consists of two disjoint horizontal strips of $R|_{\Sigma^2(x)}$, labeled by $H_1$ and $H_2$ in Fig.~\ref{fig:3D_dooubly_folded_horseshoe_y_z_slice_Gammas}. Strictly speaking, both $H_1$ and $H_2$ depend on $x$, i.e., the position of $\Sigma^2(x)$-slice along the $x$-axis, therefore should be written as $H_1(x)$ and $H_2(x)$. However, since the $x$-dependence will not be used for the rest of the proof, we simply omit it and write the horizontal strips without explicit $x$-dependence. When viewed in each $\Sigma^2(x)$, $R \cap f_{\rm II}(d(c))$ can only exist inside $H_1$ and $H_2$:
\begin{equation}\label{eq:V intersects f_II D y z slice first bound}
R \cap f_{\rm II}(d(c)) \Big|_{\Sigma^2(x)} \subset H_1 \cup H_2 \ .
\end{equation}

At this point, let us notice that Eq.~(\ref{eq:V intersects f_II D y z slice first bound}) only makes use of the second row of the RHS of Eq.~(\ref{eq:V intersect f II D}), thus only provides a crude bound for $R \cap f_{\rm II}(d(c)) \Big|_{\Sigma^2(x)}$. Based upon Eq.~(\ref{eq:V intersects f_II D y z slice first bound}), we now further refine the bound for $R \cap f_{\rm II}(d(c)) \Big|_{\Sigma^2(x)}$ by imposing the third row of the RHS of Eq.~(\ref{eq:V intersect f II D}). 

The third row of the RHS of Eq.~(\ref{eq:V intersect f II D}) can be rewritten into the parameterized form
\begin{equation}\label{eq:y z parabola 2 f_II}
-x - y^2 -2 y z^2 - z^4 + a + 2 a (y + z^2) - a^2 = -s, ~~~~{\rm where}~ s=-c\ ,
\end{equation}
from which we solve for $y$ and obtain two branches of solutions:
\begin{equation}\label{eq: y z parabola 2 f_II two branches}
y_{\pm}(z,x,s) = -z^2 + a \pm \sqrt{s-x+a}\ .
\end{equation}
Accordingly, let us define two families of parabolas in $\Sigma^2(x)$, denoted by $\Lambda^{\pm}_y(z,x,s)$, where
\begin{equation}\label{eq:Lambda y in terms of z x s}
\Lambda^{\pm}_y(z,x,s) = -z^2 + a \pm \sqrt{s-x+a}
\end{equation}
where $x$ and $s$ are viewed as parameters with bounds $|x|,|s| \leq r$. When viewed in each $\Sigma^2(x)$ plane, $\Lambda^{+}_y(z,x,s)$ is a family of parabolas parameterized by $s$, bounded by
\begin{equation}\label{eq:Lambda y + lower and upper bounds on each slice}
\Lambda^{+,1}_y(z,x) \leq \Lambda^{+}_y(z,x,s) \leq \Lambda^{+,2}_y(z,x)
\end{equation} 
where the lower and upper bounds are attained at
\begin{eqnarray}
& \Lambda^{+,1}_y(z,x) = \Lambda^+_y(z,x,s)|_{s=-r} = -z^2 + a + \sqrt{a-x-r}  \label{eq:Lambda y + 1} \\
& \Lambda^{+,2}_y(z,x) = \Lambda^+_y(z,x,s)|_{s=r} = -z^2 + a + \sqrt{a-x+r}  \label{eq:Lambda y + 2}\ .
\end{eqnarray}
Similarly, when viewed in each $\Sigma^2(x)$ plane, $\Lambda^{-}_y(z,x,s)$ is a family of parabolas parameterized by $s$, bounded by
\begin{equation}\label{eq:Lambda y - lower and upper bounds on each slice}
\Lambda^{-,1}_y(z,x) \leq \Lambda^{-}_y(z,x,s) \leq \Lambda^{-,2}_y(z,x)
\end{equation} 
where the lower and upper bounds are attained at
\begin{eqnarray}
& \Lambda^{-,1}_y(z,x) = \Lambda^-_y(z,x,s)|_{s=r} = -z^2 + a - \sqrt{a-x+r}  \label{eq:Lambda y - 1} \\
& \Lambda^{-,2}_y(z,x) = \Lambda^-_y(z,x,s)|_{s=-r} = -z^2 + a - \sqrt{a-x-r}  \label{eq:Lambda y - 2}\ .
\end{eqnarray}

\begin{figure}
        \centering
        \includegraphics[width=0.7\linewidth]{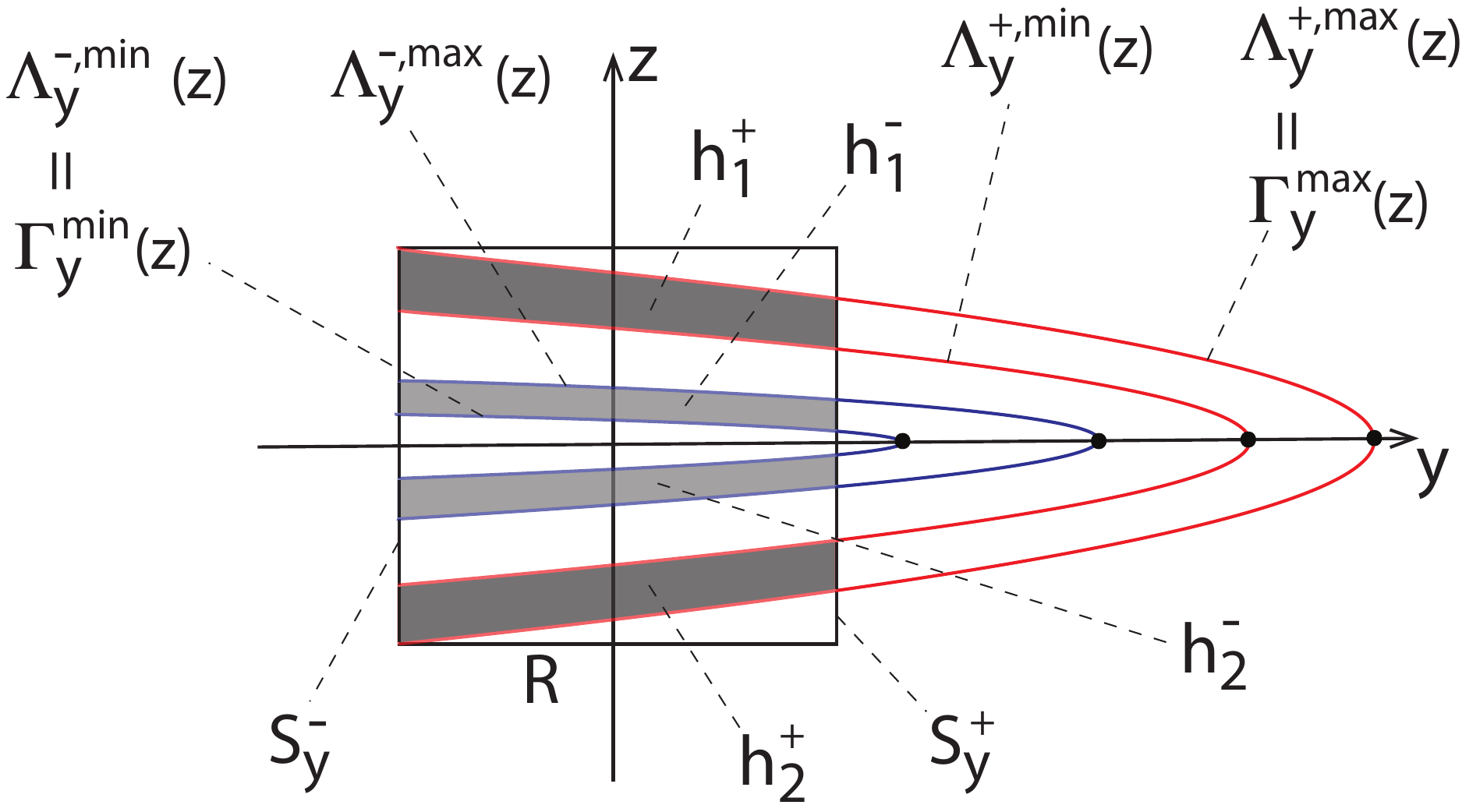} 
	\caption{(Schematic, color online) Visualization of the intersection between \( R \) and the regions bounded by the parabolas \( \Lambda^{\pm,\max}_y(z) \) and \( \Lambda^{\pm,\min}_y(z) \). The region between \( \Lambda^{+,\max}_y(z) \) and \( \Lambda^{+,\min}_y(z) \) intersects \( R \) in two disjoint horizontal slabs, labeled \( h^+_1 \) and \( h^+_2 \) (darker-shaded). Similarly, the region between \( \Lambda^{-,\max}_y(z) \) and \( \Lambda^{-,\min}_y(z) \) intersects \( R \) in two additional disjoint slabs, labeled \( h^-_1 \) and \( h^-_2 \) (lighter-shaded). Note the two pairs of identical parabolas: \( \Lambda^{-,\min}_y(z) = \Gamma^{\min}_y(z) \) and \( \Lambda^{+,\max}_y(z) = \Gamma^{\max}_y(z) \).}
           \label{fig:3D_doubly_folded_horseshoe_y_z_slice_Lambdas}
\end{figure}

It is desirable to get rid of the $x$-dependence in Eqs.~(\ref{eq:Lambda y + lower and upper bounds on each slice}) and (\ref{eq:Lambda y - lower and upper bounds on each slice}). This can be done by obtaining uniform lower and upper bounds for $\Lambda^{\pm}_y(z,x,s)$ with respect to change in $(x,s)$. A simple calculation shows:
\begin{equation}\label{eq:Gamma y +- uniform bounds}
\Lambda^{\pm,\min}_y (z) \leq \Lambda^{\pm}_y(z,x,s) \leq \Lambda^{\pm,\max}_y (z)
\end{equation}
where the bounds are attained at
\begin{eqnarray}
& \Lambda^{+,\min}_y(z) = \Lambda^+_y(z,x,s)|_{(x,s)=(r,-r)} = -z^2 + a + \sqrt{a-2r}  \label{eq:Lambda y + min} \\
& \Lambda^{+,\max}_y(z) = \Lambda^+_y(z,x,s)|_{(x,s)=(-r,r)} = -z^2 + a + \sqrt{a+2r}  \label{eq:Lambda y + max} \\
& \Lambda^{-,\min}_y(z) = \Lambda^-_y(z,x,s)|_{(x,s)=(-r,r)} = -z^2 + a - \sqrt{a+2r}  \label{eq:Lambda y - min} \\
& \Lambda^{-,\max}_y(z) = \Lambda^-_y(z,x,s)|_{(x,s)=(r,-r)} = -z^2 + a - \sqrt{a-2r}  \label{eq:Lambda y - max}\ .
\end{eqnarray}
A schematic illustration of the four parabolas is given in Fig.~\ref{fig:3D_doubly_folded_horseshoe_y_z_slice_Lambdas}. At this point, it is worthwhile checking that since $a > 5 + 2\sqrt{5}$, we have 
\begin{equation}\label{eq:a-2r}
a-2r >0\ ,
\end{equation}
i.e., the square roots in Eqs.~(\ref{eq:Lambda y + min}) and (\ref{eq:Lambda y - max}) are real-valued. Also, it is easy to check that $r = \sqrt{a+2r}$, therefore we obtain the important relations
\begin{eqnarray}\label{eq:Identical parabola bounds}
& \Lambda^{+,\max}_y(z) = \Gamma^{\max}_y(z) \label{eq:Identical parabola bounds max} \\
& \Lambda^{-,\min}_y(z) = \Gamma^{\min}_y(z) \ , \label{eq:Identical parabola bounds min} 
\end{eqnarray}
as indicated by Fig.~\ref{fig:3D_doubly_folded_horseshoe_y_z_slice_Lambdas}. Therefore, conditions (a) and (b) immediately apply to $\Lambda^{-,\min}_y(z)$ and $\Lambda^{+,\max}_y(z)$, respectively. This guarantees that the region bounded between the parabolas \( \Lambda^{+,\max}_y(z) \) and \( \Lambda^{+,\min}_y(z) \) intersects \( R \) in two disjoint horizontal slabs, labeled \( h^+_1 \) and \( h^+_2 \) in Fig.~\ref{fig:3D_doubly_folded_horseshoe_y_z_slice_Lambdas} (darker-shaded regions). Similarly, the region bounded between \( \Lambda^{-,\max}_y(z) \) and \( \Lambda^{-,\min}_y(z) \) intersects \( R \) in two additional disjoint horizontal slabs, labeled \( h^-_1 \) and \( h^-_2 \) in the same figure (lighter-shaded regions).

Furthermore, Eqs.~(\ref{eq:Identical parabola bounds max}) and (\ref{eq:Identical parabola bounds min}) also guarantee that $h^{\pm}_1 \subset H_1$ and $h^{\pm}_2 \subset H_2$. Hence when viewed in each $\Sigma^2(x)$ plane (see Fig.~\ref{fig:3D_doubly_folded_horseshoe_y_z_slice_Lambdas}), $R \cap f_{\rm II}(d(c))\big|_{\Sigma^2(x)}$ consists of four disjoint horizontal slices of $R\big|_{\Sigma^2(x)}$ that lies within the four disjoint horizontal strips 
\begin{equation}\label{eq:V intersects f_II D y z slice second bound}
R \cap f_{\rm II}(d(c)) \Big|_{\Sigma^2(x)} \subset h^{+}_1\cup h^{-}_1 \cup h^{+}_2 \cup h^{-}_2 \subset H_1 \cup H_2 \ . 
\end{equation}
The theorem is thus proved.
\end{proof}

Having established that \( f_{\rm II}(d(s)) \cap R \) consists of four disjoint horizontal slices for every \( s \in I^s \), we may now invoke Proposition~\ref{Semi-conjugacy to the full shift on N symbols}. This immediately implies the existence of a nonempty compact invariant set \( \Omega \subset \bigcap_{n \in \mathbb{Z}} f_{\rm II}^n(R) \), on which the restricted map \( f_{\rm II}|_\Omega \) is semi-conjugate to the full shift on four symbols.

\subsection{Map \( f_{\mathrm{III}} \): doubly and singly folded sheet in four dimensions}
\label{f III}

We now turn to a four-dimensional symplectic H\'{e}non-type map \( f_{\mathrm{III}} \), constructed by coupling two classical H\'{e}non maps—one acting on the \( (x,z) \)-plane and the other on the \( (y,w) \)-plane. This construction introduces a novel degree of geometric flexibility: unlike in three dimensions, where all foldings must stack along a common axis, the four-dimensional setting allows distinct stacking directions, enabling new classes of paperfolding structures. In particular, we identify two parameter regimes that give rise to qualitatively different horseshoe configurations.

In the first regime, the resulting horseshoe is described by the doubly folded template \( \mathcal{F}^w_{(x,\underline{y})} \circ \mathcal{F}^z_{(\underline{x},y)} \), representing a two-dimensional sheet in four-dimensional space with independent crease and stacking directions, as introduced in Section~\ref{Folding a two-dimensional sheet in four dimensions}. Such a configuration is inherently four-dimensional and has no analogue in lower dimensions. In the second regime, the system exhibits a singly folded horseshoe modeled by the template \( \mathcal{F}^z_{(x,\underline{y})} \), in which both folding and stacking occur within the same three-dimensional subspace. The transition between these two regimes implies the existence of a bifurcation—or sequence of bifurcations—in which the doubly folded configuration unfolds into a singly folded one, corresponding to a reduction in the system's topological complexity.

Consider the H\'{e}non-type map in four dimensions, denoted by $f_{\rm III}$
\begin{equation}\label{eq:4D doubly folded Henon}
\left( \begin{array}{ccc}
x_{n+1}\\
y_{n+1} \\
z_{n+1} \\
w_{n+1} \end{array} \right) = 
f_{\rm III} \left( \begin{array}{ccc}
x_{n}\\
y_{n} \\
z_{n}\\
w_{n} \end{array} \right) =
\left( \begin{array}{ccc}
a_0 - x^2_n - z_n + c(x_n - y_n)\\
a_1 - y^2_n - w_n - c(x_n - y_n)\\
x_n \\
y_n \end{array} \right)
\end{equation}
which is the coupled H\'{e}non map studied in \cite{Fujioka25}. The parameters in this system are $a_0$, $a_1$, and $c$. The parameters $a_0$ and $a_1$ govern the rates of expansion in the $x$- and $y$-directions, respectively, while $c$ controls the coupling strength between the dynamics in the $(x,z)$-plane and those in the $(y,w)$-plane. Here we remark that the parameters $a_0$ and $a_1$ are inherited from the original two-dimensional H\'enon map, and the coupling strength is taken to be linear and thus governed by a single proportionality factor $c$. In principle, one could also introduce additional parameters—for example, the “$b$” parameter in \cite{Henon76} that controls dissipation—or allow nonlinear coupling terms between the two two-dimensional subspaces. However, since the aim of this paper is to demonstrate the formation of new horseshoe configurations in higher dimensions, we restrict attention to the simplest case that conserves phase-space area ($b=1$) with linear coupling only.

The inverse of $f_{\rm III}$ is 
\begin{equation}\label{eq:4D doubly folded Henon inverse}
\left( \begin{array}{ccc}
x_{n-1}\\
y_{n-1} \\
z_{n-1} \\
w_{n-1} \end{array} \right) = 
f^{-1}_{\rm III} \left( \begin{array}{ccc}
x_{n}\\
y_{n} \\
z_{n}\\
w_{n} \end{array} \right) =
\left( \begin{array}{ccc}
z_n \\
w_n \\
a_0 - z^2_n - x_n + c(z_n - w_n) \\
a_1 - w^2_n - y_n - c(z_n - w_n)\end{array} \right)\ .
\end{equation}
Notice that by replacement of variables $(x,y,z,w) \mapsto (z,w,x,y)$, $f_{\rm III}$ is transformed into $f^{-1}_{\rm III}$. 

To motivate the emergence of distinct horseshoe geometries in the four-dimensional map \( f_{\rm III} \), we consider two types of Anti-Integrable (AI) limits~\cite{Aubry90} that were proposed and analyzed in detail in our previous work~\cite{Fujioka25}. These limiting scenarios, referred to as \emph{Type A} and \emph{Type B}, provide asymptotic regimes in which the topological structure of the invariant set becomes particularly transparent.

\emph{Type A} arises by taking \( a_0 = a_1 = a \to \infty \) while keeping the coupling parameter \( c \) finite. This leads to infinite expansion within the \( (x,z) \)- and \( (y,w) \)-planes while maintaining finite coupling between them. In contrast, \emph{Type B} is obtained by taking \( a_0 = a_1 = a \to \infty \) while fixing the ratio \( c/\sqrt{a} = \gamma \) fixed with $\gamma>1$, which again induces infinite planar expansion but imposes a coupling strength that scales with \( \sqrt{a} \).

These two limits yield qualitatively different phase-space geometries. Type A results in a doubly folded horseshoe modeled by the template \( \mathcal{F}^w_{(x,\underline{y})} \circ \mathcal{F}^z_{(\underline{x},y)} \), which is topologically equivalent to the direct product of two two-dimensional singly folded horseshoes---one in each plane. Type B, on the other hand, gives rise to a singly folded configuration corresponding to the template \( \mathcal{F}^z_{(x,\underline{y})} \).

The existence of these two structurally distinct limits implies a global deformation between them. Along any path in parameter space \((a_0,a_1,c)\) from the Type~A neighborhood to the Type~B one, the doubly folded horseshoe must be transformed into a singly folded configuration; however, this need not occur in a single codimension-one event. The change may proceed gradually through parameter intervals with intricate intermediate states—e.g., partial unfoldings, temporary additional folds, windows of near-integrable behavior, or sequences of (hetero)clinic tangencies and crises—before the Type~B template is attained. During this passage the invariant set is progressively reorganized, typically with a stepwise (rather than strictly monotone) reduction in topological entropy.

\subsubsection{Type A: Doubly Folded Horseshoe Described by Template \( \mathcal{F}^w_{(x,\underline{y})} \circ \mathcal{F}^z_{(\underline{x},y)} \)}
\label{Type A}

In the neighborhood of the Type A anti-integrable limit, where \( a_0 = a_1 = a \to \infty \) and the coupling parameter \( c \) remains finite, the four-dimensional Hénon-type map \( f_{\mathrm{III}} \) generates a geometric configuration that is well described by the doubly folded paperfolding template \( \mathcal{F}^w_{(x,\underline{y})} \circ \mathcal{F}^z_{(\underline{x},y)} \). This structure arises from two independent folding operations applied to a two-dimensional sheet in four-dimensional space, each with its own crease and stacking direction. In this regime, the invariant set exhibits a product-like geometry, corresponding to the direct product of two classical Smale horseshoes in the \((x,z)\)- and \((y,w)\)-planes, respectively.

To formalize this setting, we adopt the following conventions: the \emph{horizontal directions} are defined to be the \( (x,y) \)-subspace, as these directions are approximately aligned with the expanding directions of the map \( f_{\mathrm{III}} \). The \emph{vertical directions} are taken to be the \( (z,w) \)-subspace, which corresponds to the directions of contraction. Accordingly, we fix the fundamental domain to be the hypercube \(R = I^u \times I^s\), where \( I^u \subset \mathbb{R}^2 \) is a square in the horizontal plane and \( I^s \subset \mathbb{R}^2 \) is a square in the vertical plane. The precise bounds and parameter conditions for \( R \) will be specified in Theorem~\ref{4D Doubly folded horseshoe topology}.

\begin{figure}
        \centering
        \includegraphics[width=0.7\linewidth]{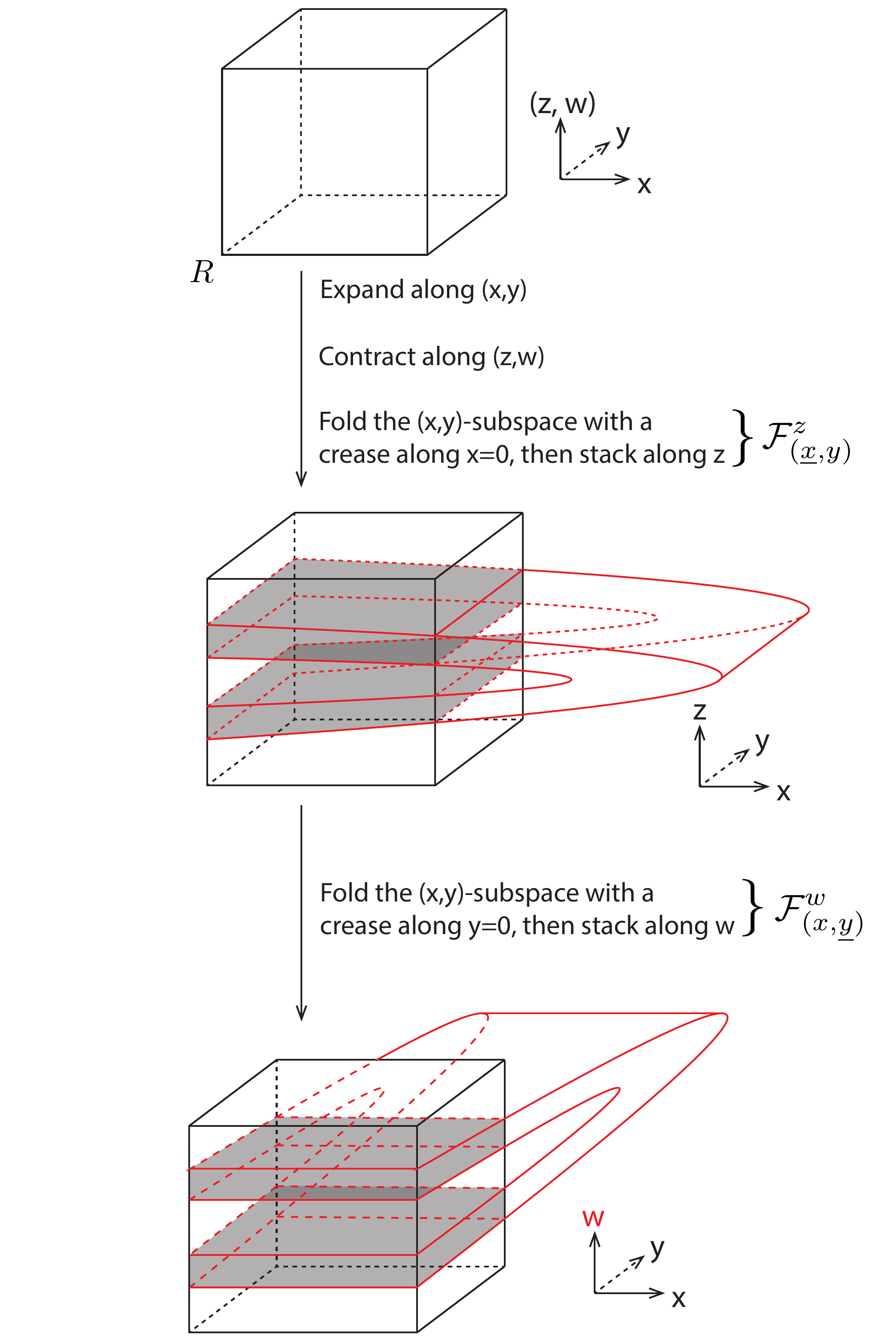} 
        \caption{(Schematic, color online) Horseshoe generated by \( f_{\rm III} \) in a neighborhood of the Type-A AI limit. The qualitative deformation of the cube \( R \) under \( f_{\rm III} \) can be visualized as follows: first, \( R \) is expanded along the \( (x,y) \)-directions and contracted along the \( (z,w) \)-directions, effectively flattening it into a quasi-two-dimensional sheet within the \( (x,y) \)-subspace. This sheet then undergoes two sequential foldings: the first, \( \mathcal{F}^z_{(\underline{x},y)} \), folds along \( x = 0 \) and stacks along the \( z \)-axis; the second, \( \mathcal{F}^w_{(x,\underline{y})} \), folds along \( y = 0 \) and stacks along the \( w \)-axis. The resulting horseshoe is described by the doubly folded template \( \mathcal{F}^w_{(x,\underline{y})} \circ \mathcal{F}^z_{(\underline{x},y)} \). Since the two folding operations act in orthogonal subspaces, their composition is commutative: the structure remains unchanged if the foldings are applied in the reverse order, as expressed in Eq.~(\ref{eq:foldings commute}).
} \label{fig:4D_doubly_folded_horseshoe}
\end{figure}

To better understand the Type-A limit, a natural starting point is the case of zero coupling, \( c = 0 \), in which \( f_{\rm III} \) reduces to the direct product of two decoupled two-dimensional H\'{e}non maps: one acting on the \( (x,z) \)-plane and the other on the \( (y,w) \)-plane.

In particular, when \( c = 0 \) and both parameters \( a_0 \) and \( a_1 \) exceed the classical Devaney--Nitecki threshold---namely, \( a_0, a_1 > 5 + 2\sqrt{5} \)---each two-dimensional component exhibits a Smale horseshoe structure. Consequently, \( f_{\rm III} \) becomes dynamically equivalent to the direct product of two classical Smale horseshoe maps, as illustrated in Fig.~\ref{fig:4D_doubly_folded_horseshoe}.

Geometrically, the action of \( f_{\rm III} \) on the hypercube \( R \) can be understood as follows: the map expands \( R \) along the horizontal directions \( (x,y) \) and contracts it along the vertical directions \( (z,w) \), effectively flattening \( R \) into a quasi-two-dimensional sheet lying in the \( (x,y) \)-subspace. The two-dimensional H\'{e}non map in the \( (x,z) \)-plane acts on each \( x \)-fiber of this sheet (i.e., segments aligned with the \( x \)-axis), folding it along the crease \( x = 0 \) and stacking the two halves along the \( z \)-axis. The resulting configuration of each fiber is described by the one-dimensional folding template \( \mathcal{F}^z_{\underline{x}} \). Similarly, the H\'{e}non map in the \( (y,w) \)-plane folds each \( y \)-fiber of the sheet along the crease \( y = 0 \) and stacks the halves along the \( w \)-axis, yielding the configuration \( \mathcal{F}^w_{\underline{y}} \).

The global image \( f_{\rm III}(R) \) is therefore represented by the Cartesian product \( \mathcal{F}^z_{\underline{x}} \times \mathcal{F}^w_{\underline{y}} \), where each factor corresponds to a one-dimensional folded sheet in a two-dimensional subspace. By Eq.~(\ref{eq:folding decomposition into direct products}), this product is equivalent to the composition
\[
\mathcal{F}^w_{(x,\underline{y})} \circ \mathcal{F}^z_{(\underline{x},y)},
\]
which describes a two-dimensional sheet in four dimensions folded along two orthogonal crease directions and stacked along two independent stacking axes.

A signature observable consequence of the Cartesian product structure is the behavior of individual horizontal disks under the map \( f_{\rm III} \). Specifically, for every horizontal disk \( d(v) \subset R \) with $v\in I^s$, its image \( f_{\rm III}(d(s)) \) intersects \( R \) in four mutually disjoint horizontal slices. This property will be rigorously verified in Theorem~\ref{4D Doubly folded horseshoe topology}, providing strong evidence that the geometry of \( f_{\rm III}(R) \) is accurately described by the doubly folded paperfolding template \( \mathcal{F}^w_{(x,\underline{y})} \circ \mathcal{F}^z_{(\underline{x},y)} \).

\begin{figure}[th]
        \centering
        \includegraphics[width=0.7\linewidth]{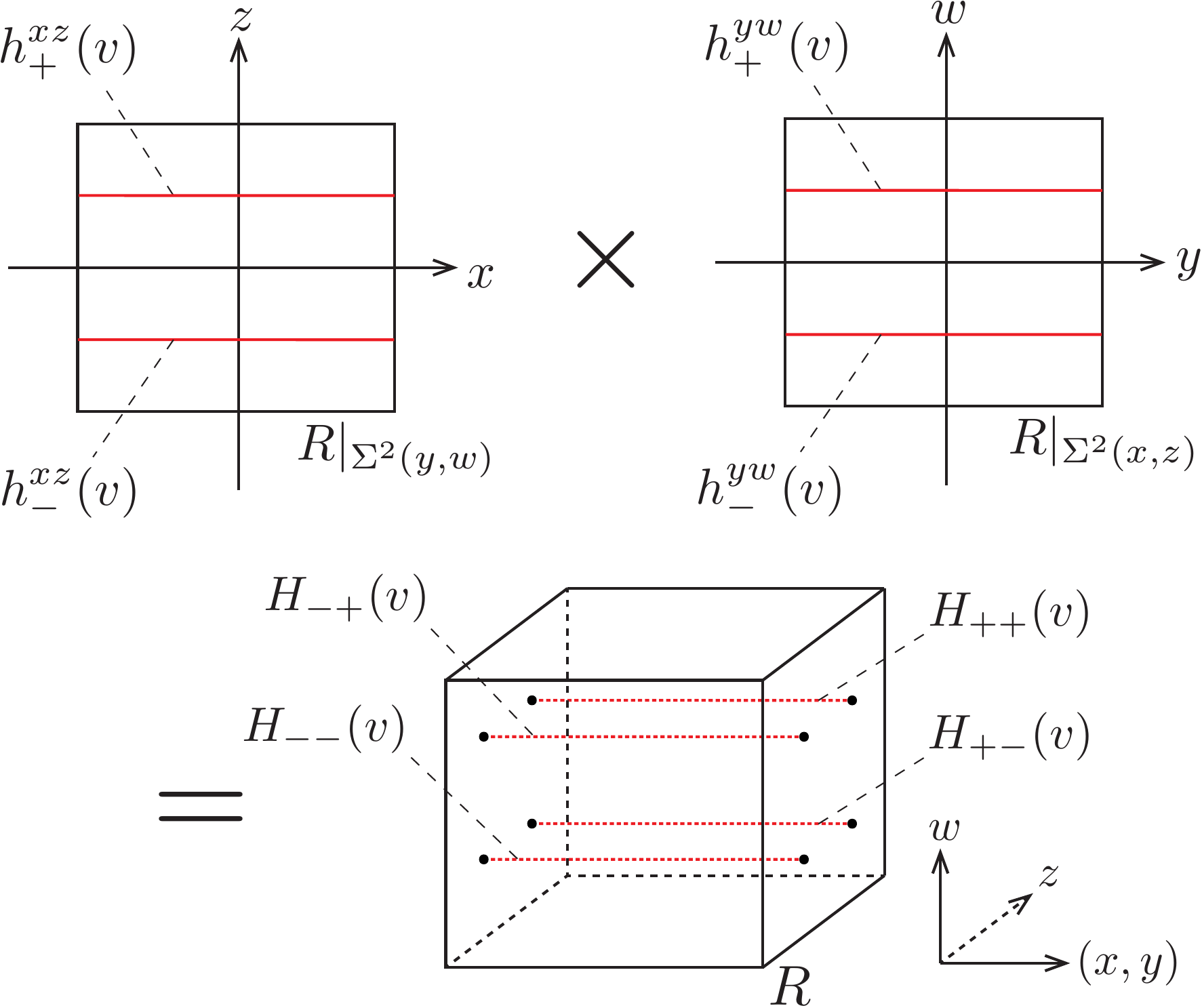} 
        \caption{(Schematic, color online) For \( c = 0 \) and \( a_0, a_1 > 5+2\sqrt{5} \), the intersection \( f_{\rm III}(d(v)) \cap R \) assumes the Cartesian product structure given in Eq.~(\ref{eq:f III four horizontal slices direct prodcut}) for every \( v \in I^s \), consisting of four mutually disjoint horizontal slices of \( R \), labeled \( H_{\pm\pm}(v) \) (dashed red segments). 
} \label{fig:f_III_direct_product}
\end{figure}

This structure arises because the action of \( f_{\rm III} \) decomposes into two independent one-dimensional foldings: every \( x \)-fiber of \( d(v) \) (i.e., a curve obtained by fixing \( y \) in the domain of \( d(v) \)) is folded according to the template \( \mathcal{F}^z_{\underline{x}} \), and every \( y \)-fiber (obtained by fixing \( x \)) is folded according to \( \mathcal{F}^w_{\underline{y}} \).

In the \( \Sigma^2(y,w) \)-plane—defined as the \( (x,z) \)-plane at fixed \( (y,w) \)—the horizontal direction is taken to be the \( x \)-axis, which is approximately aligned with the direction of expansion within the \( (x,z) \)-plane. The image \( f_{\rm III}(d(v)) \cap R \) restricts in this plane to two disjoint horizontal slices:
\[
f_{\rm III}(d(v)) \cap R \big|_{\Sigma^2(y,w)} = h^{xz}_{+}(v) \cup h^{xz}_{-}(v)\ ,
\]
where \( h^{xz}_{\pm}(v) \), as illustrated in Fig.~\ref{fig:f_III_direct_product}, are disjoint horizontal slices of \( R \big|_{\Sigma^2(y,w)} \).

Likewise, in the \( \Sigma^2(x,z) \)-plane—defined as the \( (y,w) \)-plane at fixed \( (x,z) \)—the horizontal direction is taken to be the \( y \)-axis, which is approximately aligned with the direction of expansion within the \( (y,w) \)-plane. In this case, the restriction is
\[
f_{\rm III}(d(v)) \cap R \big|_{\Sigma^2(x,z)} = h^{yw}_{+}(v) \cup h^{yw}_{-}(v)\ ,
\]
where \( h^{yw}_{\pm}(v) \), as illustrated in Fig.~\ref{fig:f_III_direct_product}, are disjoint horizontal slices of \( R \big|_{\Sigma^2(x,z)} \).

Taken together, the total intersection of the image with \( R \) is the Cartesian product
\begin{eqnarray}
f_{\rm III}(d(v)) \cap R &= \left( h^{xz}_{+}(v) \cup h^{xz}_{-}(v) \right) \times \left( h^{yw}_{+}(v) \cup h^{yw}_{-}(v) \right)\nonumber \\
&=H_{++}(v)\cup H_{+-}(v)\cup H_{-+}(v)\cup H_{--}(v), \label{eq:f III four horizontal slices direct prodcut}
\end{eqnarray}
where
\begin{eqnarray}
&H_{++}(v) = h^{xz}_{+}(v)\times h^{yw}_{+}(v),\quad H_{+-}(v) = h^{xz}_{+}(v)\times h^{yw}_{-}(v),\nonumber\\
&H_{-+}(v) = h^{xz}_{-}(v)\times h^{yw}_{+}(v),\quad H_{--}(v) = h^{xz}_{-}(v)\times h^{yw}_{-}(v), \label{eq:f III four horizontal slices direct prodcut definition}
\end{eqnarray}
are mutually disjoint horizontal slices of \( R \) in the full four-dimensional space. This Cartesian product structure is illustrated in Fig.~\ref{fig:f_III_direct_product}. This slicing structure is a direct geometric consequence of the product template \( \mathcal{F}^z_{\underline{x}} \times \mathcal{F}^w_{\underline{y}} \), and hence of the composite template \( \mathcal{F}^w_{(x,\underline{y})} \circ \mathcal{F}^z_{(\underline{x},y)} \).

When \( c > 0 \), it is natural to expect that for sufficiently weak coupling—specifically, when \( a_0, a_1 \gg c > 0 \)—the qualitative geometry of \( f_{\rm III}(R) \cap R \) remains essentially unchanged. In this regime, the dominant expansion within the \( (x,z) \)- and \( (y,w) \)-planes, governed by \( a_0 \) and \( a_1 \), respectively, greatly outweighs the comparatively weak inter-plane coupling controlled by the parameter \( c \). As a result, the leading-order geometry is expected to be unaffected by the coupling, and the configuration of \( f_{\rm III}(R) \cap R \) should closely resemble that of the decoupled case \( c = 0 \), which is well captured by the doubly folded paperfolding template \( \mathcal{F}^w_{(x,\underline{y})} \circ \mathcal{F}^z_{(\underline{x},y)} \).

In particular, we expect that for any horizontal disk \( d(v) \subset R \), although the exact Cartesian product structure in Eq.~(\ref{eq:f III four horizontal slices direct prodcut}) no longer applies, the intersection \( f_{\rm III}(d(v)) \cap R \) still consists of four disjoint horizontal slices of \( R \). This expectation is rigorously confirmed by Theorem~\ref{4D Doubly folded horseshoe topology}, which shows that the local image structure of \( f_{\rm III} \) remains compatible with the paperfolding template \( \mathcal{F}^w_{(x,\underline{y})} \circ \mathcal{F}^z_{(\underline{x},y)} \) even in the presence of weak coupling.

\begin{theorem}\label{4D Doubly folded horseshoe topology}
Let the parameters of \( f_{\rm III} \) satisfy \( A_0 = (a_0 + a_1)/2 \geq -1 \) and define \( r = 2\sqrt{2}\,(1 + \sqrt{1 + A_0}) \). Let \( R = [-r, r]^4 = I^u \times I^s \), where \( I^u = [-r, r]^2 \) is the horizontal square in the \( (x,y) \)-subspace, and \( I^s = [-r, r]^2 \) is the vertical square in the \( (z,w) \)-subspace. Suppose the parameters satisfy the following bounds:
\begin{eqnarray}
& 0 < \frac{1}{4}c^2 + a_i - (c+2) r \ , \quad (i = 0, 1) \label{eq:4D_Doubly_folded_horseshoe_topology_parameter_bound_1} \\
& 0 \leq r^2 - 2(c+1) r - a_i \ . \quad (i = 0, 1) \label{eq:4D_Doubly_folded_horseshoe_topology_parameter_bound_2}
\end{eqnarray}
Then for every horizontal disk $d(v)=I^u \times \lbrace v \rbrace$, where $v\in I^s$, the intersection \( f_{{\rm III}}(d(v))\cap R \) consists of four disjoint horizontal slices of \( R \).
\end{theorem}

\begin{proof}

Fix a vertical coordinate \( v = (z_v, w_v) \in I^s \). Using the identity relation \( f_{\rm III}^{-1}(f_{\rm III}(d(v))) = d(v) \), we obtain the following analytic expression for the image \( f_{\rm III}(d(v)) \):
\begin{eqnarray} 
 f_{\rm III}(d(v))&= 
\left\lbrace (x,y,z,w) \,\middle|\, \begin{array}{l}
|z| \leq r\ , \\
|w| \leq r\ , \\
a_0 - z^2 - x + c(z - w) = z_v\ , \\
a_1 - w^2 - y - c(z - w) = w_v
\end{array} \right\rbrace .\label{eq:f III V}
\end{eqnarray}

The intersection of \( f_{\rm III}(d(v)) \) with \( R \) is obtained by further imposing the bounds on \( x \) and \( y \):
\begin{eqnarray}
f_{\rm III}(d(v)) \cap R
&=&
\left\lbrace (x,y,z,w) \,\middle|\,
\begin{array}{l}
|x|,\,|y|,\,|z|,\,|w| \leq r, \\[4pt]
a_0 - z^{2} - x + c(z - w) = z_v, \\[4pt]
a_1 - w^{2} - y - c(z - w) = w_v
\end{array}
\right\rbrace .
\label{eq:f III V intersect R}
\end{eqnarray}

From Eq.~(\ref{eq:f III V}), we see that the two-dimensional surface \( f_{\mathrm{III}}(d(v)) \) can be parameterized by \((z,w)\), with the \(x\)- and \(y\)-components of points on this surface uniquely determined by \((z,w)\) through the third and fourth rows on the right hand side of Eq.~(\ref{eq:f III V}), respectively. 

We now examine the intersection of \( f_{\rm III}(d(v)) \) with the hypercube \( R \).  
Because this intersection is a two‑dimensional set embedded in four dimensions, it is difficult to visualize directly.  
To aid visualization, we study its intersection with a two‑dimensional surface of section, defined by  
\begin{equation}\label{eq: z w slice}
\Sigma^2(x,y) = \left\lbrace (x',y',z',w') \,\middle|\, x' = x,\; y' = y,\; z',w' \in \mathbb{R} \right\rbrace .
\end{equation}
In other words, \( \Sigma^2(x,y) \) is the \( (z,w) \)-plane at the fixed coordinates \( (x,y) \).

Within this plane, the third row on the right-hand side of Eq.~(\ref{eq:f III V}) can be interpreted as a quadratic function \( w(z;x,z_v) \) that expresses \( w \) in terms of \( z \), with parameters \( x \) and \( z_v \):  
\begin{equation}\label{eq:w in terms of z x z_v}
w(z;x,z_v) = -\frac{z^2}{c} + z + \frac{a_0}{c} - \frac{x+z_v}{c},
\end{equation}
which defines a parabola in \( \Sigma^2(x,y) \).

\begin{figure}[th]
        \centering
        \includegraphics[width=0.75\linewidth]{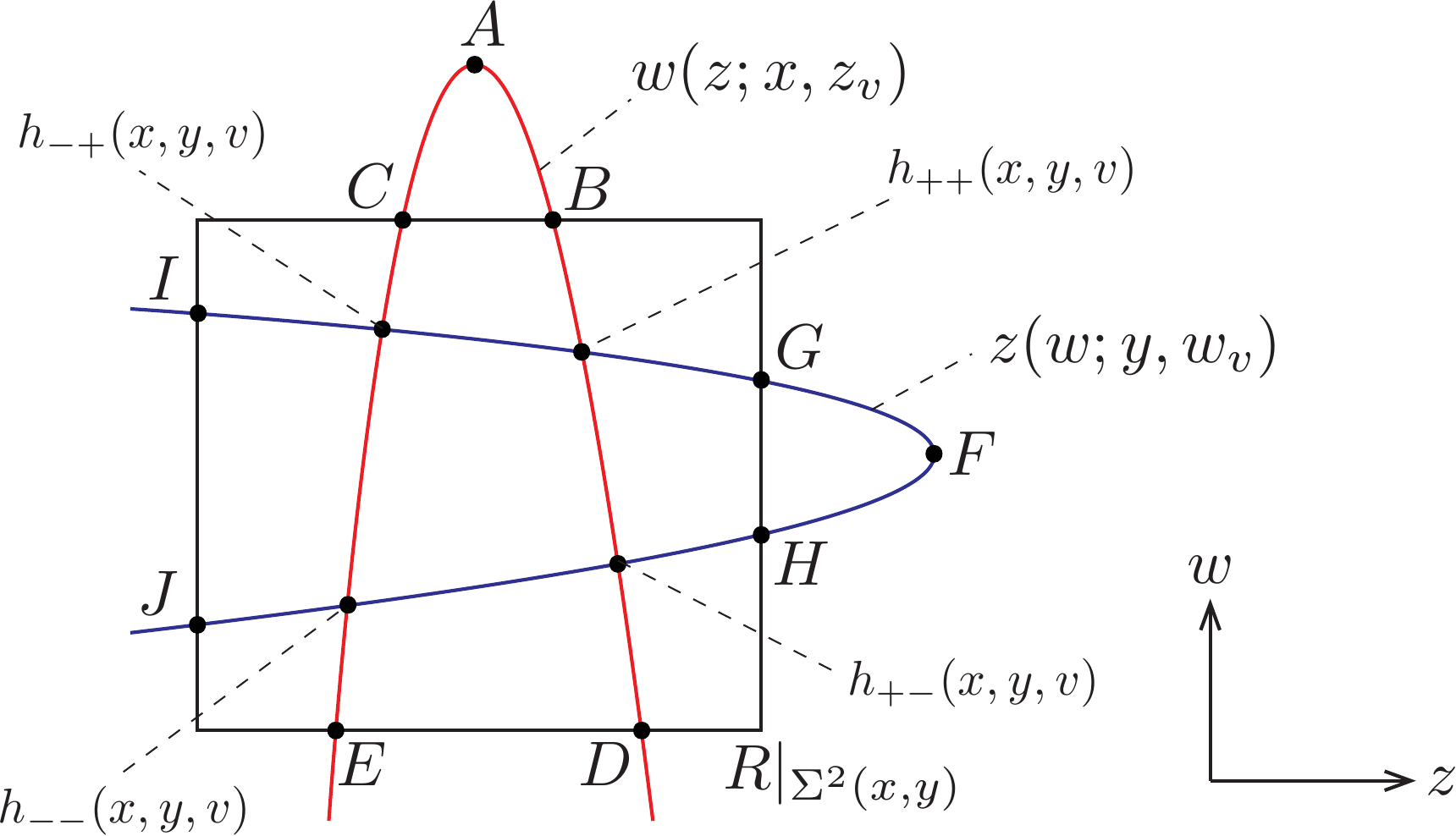} 
        \caption{(Schematic, color online) For every \( (x,y)\in I^u \) and \( v\in I^s \), the graph of \( w(z;x,z_v) \) (red parabola) intersects \( R|_{\Sigma^2(x,y)} \) in two disjoint segments—between \( B \) and \( D \), and between \( C \) and \( E \)—each spanning the entire \( w \)-direction of \( R|_{\Sigma^2(x,y)} \). Likewise, the graph of \( z(w;y,w_v) \) (blue parabola) intersects \( R|_{\Sigma^2(x,y)} \) in two disjoint segments—between \( G \) and \( I \), and between \( H \) and \( J \)—each spanning the entire \( z \)-direction of \( R|_{\Sigma^2(x,y)} \). The two parabolas therefore intersect at four isolated points, labeled \( h_{\pm\pm}(x,y,v) \).
} \label{fig:Two_parabolas_xy_plane}
\end{figure}

Similarly, the fourth row on the right-hand side of Eq.~(\ref{eq:f III V}) can be interpreted as a quadratic function \( z(w;y,w_v) \) that expresses \( z \) in terms of \( w \), with parameters \( y \) and \( w_v \):  
\begin{equation}\label{eq:z in terms of w y w_v}
z(w;y,w_v) = -\frac{w^2}{c} + w + \frac{a_1}{c} - \frac{y+w_v}{c},
\end{equation}
which defines another parabola in \( \Sigma^2(x,y) \).

We now analyze the shapes and positions of the graphs of \( w(z;x,z_v) \) and \( z(w;y,w_v) \) relative to the restriction of \( R \) onto \( \Sigma^2(x,y) \). Our approach follows the well‑known method of Devaney and Nitecki, originally introduced in~\cite{Devaney79}. More precisely, we show that, for every \( (x,y)\in I^u \) and every \( v\in I^s \), the following conditions hold:
\begin{enumerate}
\item[(a)] The parabola given by the graph of \( w(z;x,z_v) \) intersects \( R|_{\Sigma^2(x,y)} \) in two disjoint components, each spanning the entire \( w \)-direction of \( R|_{\Sigma^2(x,y)} \). These components are schematically indicated by the segments on the red curve between \( B \) and \( D \), and between \( C \) and \( E \) in Fig.~\ref{fig:Two_parabolas_xy_plane}.
\item[(b)] The parabola given by the graph of \( z(w;y,w_v) \) intersects \( R|_{\Sigma^2(x,y)} \) in two disjoint components, each spanning the entire \( z \)-direction of \( R|_{\Sigma^2(x,y)} \). These components are schematically indicated by the segments on the blue curve between \( G \) and \( I \), and between \( H \) and \( J \) in Fig.~\ref{fig:Two_parabolas_xy_plane}.
\end{enumerate}

We now proceed to verify conditions (a) and (b) separately.

Consider first the parabola given by the graph of \( w(z;x,z_v) \), whose vertex is denoted by \( A = (z_A, w_A) \) in Fig.~\ref{fig:Two_parabolas_xy_plane}.  
To establish condition (a), it suffices to show that, for every \( (x,y)\in I^u \) and every \( v\in I^s \), the following hold:
\begin{enumerate}
\item[(a.1)] The vertex \( A \) lies outside \( R|_{\Sigma^2(x,y)} \) on the upper side, i.e., \( w_A > r \), as illustrated in Fig.~\ref{fig:Two_parabolas_xy_plane}.
\item[(a.2)] The parabola intersects the lower edge of \( R|_{\Sigma^2(x,y)} \) at two points, labeled \( D \) and \( E \) in Fig.~\ref{fig:Two_parabolas_xy_plane}.
\end{enumerate}

To verify condition (a.1), we first compute  
\begin{equation}
w_A = \frac{c}{4} + \frac{a_0}{c} - \frac{x+z_v}{c}.
\end{equation}
Since \( x, z_v \in [-r,r] \), it follows that  
\begin{equation}
w_A \;\geq\; \frac{c}{4} + \frac{a_0}{c} - \frac{2r}{c}.
\end{equation}
Moreover, by Eq.~(\ref{eq:4D_Doubly_folded_horseshoe_topology_parameter_bound_1}), we have  
\begin{equation}
\frac{c}{4} + \frac{a_0}{c} - \frac{2r}{c} > r,
\end{equation}
and therefore \( w_A > r \).

To verify condition (a.2), it suffices to show that, for every \( x, z_v \in [-r,r] \),
\begin{equation}
w(z;x,z_v)\big|_{z=\pm r} \leq -r.
\end{equation}
Using the expression for \( w(z;x,z_v) \), we obtain
\begin{equation}
w(z;x,z_v)\big|_{z=\pm r} = -\frac{r^{2}}{c} \pm r + \frac{a_0}{c} - \frac{x+z_v}{c}.
\end{equation}
Since \( x, z_v \in [-r,r] \), this quantity is bounded above by
\begin{equation}
w(z;x,z_v)\big|_{z=\pm r} \leq -\frac{r^{2}}{c} + r + \frac{a_0}{c} + \frac{2r}{c}.
\end{equation}
Furthermore, by Eq.~(\ref{eq:4D_Doubly_folded_horseshoe_topology_parameter_bound_2}), we have
\begin{equation}
-\frac{r^{2}}{c} + r + \frac{a_0}{c} + \frac{2r}{c} \leq -r.
\end{equation}
Therefore,
\begin{equation}
w(z;x,z_v)\big|_{z=\pm r} \leq -r
\end{equation}
for all \( x, z_v \in [-r,r] \), and condition (a.2) is verified. Since both conditions (a.1) and (a.2) have been verified, condition (a) is therefore established.

A similar argument shows that condition (b) also holds, and the details are omitted for brevity.

An important consequence of conditions (a) and (b) is that, for every \( (x,y)\in I^u \) and every \( v\in I^s \), the graphs of \( w(z;x,z_v) \) and \( z(w;y,w_v) \) intersect at exactly four isolated points, denoted by \( h_{\pm\pm}(x,y,v) \) in Fig.~\ref{fig:Two_parabolas_xy_plane}.  
From Eq.~(\ref{eq:f III V intersect R}), these points coincide with the restriction of \( f_{\rm III}(d(v)) \cap R \) to the slice \( \Sigma^2(x,y) \):
\begin{eqnarray}
f_{\rm III}&(d(v)) \cap R \big|_{\Sigma^2(x,y)} \nonumber \\
&= h_{++}(x,y,v) \cup h_{+-}(x,y,v) \cup h_{-+}(x,y,v) \cup h_{--}(x,y,v).
\label{eq:f III R intersect R xy plane}
\end{eqnarray}

\begin{figure}[th]
        \centering
        \includegraphics[width=0.7\linewidth]{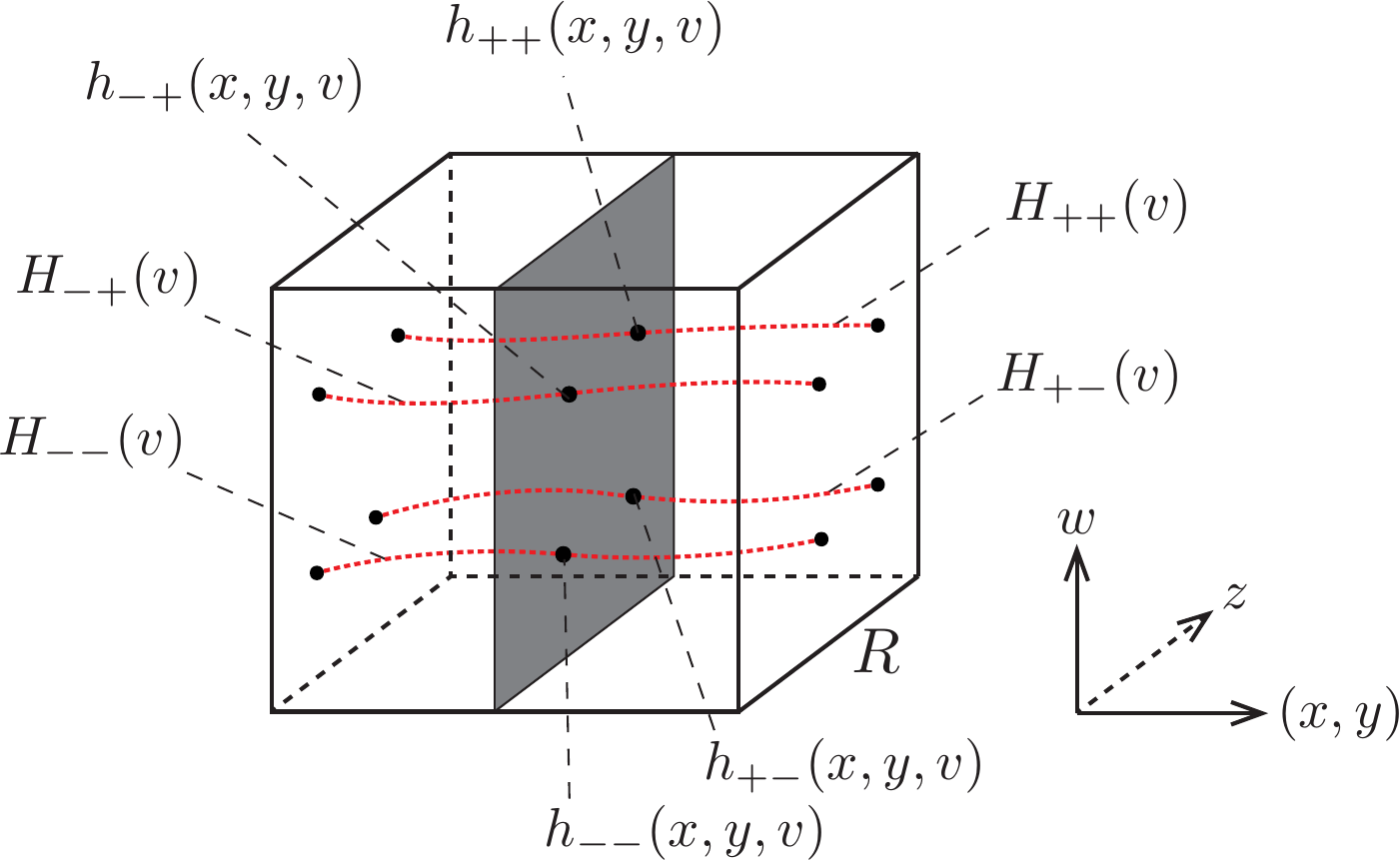} 
        \caption{(Schematic, color online) As expressed in Eq.~(\ref{eq:f III horizontal slices construction}), each \( H_{\pm\pm}(v) \) (shown as red curves) is the union of all points \( h_{\pm\pm}(x,y,v) \) taken over every slice \( \Sigma^2(x,y) \) (indicated in grey).  
Each \( H_{\pm\pm}(v) \) forms a horizontal slice of \( R \), and consequently \( f_{\rm III}(d(v)) \cap R \) is the disjoint union of four such horizontal slices.} \label{fig:f_III_four_horizontal_slices}
\end{figure}

The full intersection \( f_{\rm III}(d(v)) \cap R \) is obtained by taking the union of these restricted intersections over all \( (x,y)\in I^u \):
\begin{eqnarray}
&f_{\rm III}(d(v)) \cap R \nonumber \\
&~~\equiv \bigcup_{(x,y)\in I^u} \big[ h_{++}(x,y,v) \cup h_{+-}(x,y,v) \cup h_{-+}(x,y,v) \cup h_{--}(x,y,v) \big].
\label{eq:f III four horizontal slices}
\end{eqnarray}

Because \( f_{\rm III} \) is continuous, each point \( h_{\pm\pm}(x,y,v) \) depends continuously on \( (x,y) \) (and on \( v \) as well).  
Accordingly, by setting
\begin{equation}
H_{\pm\pm}(v) = \bigcup_{(x,y)\in I^u} h_{\pm\pm}(x,y,v), \label{eq:f III horizontal slices construction}
\end{equation}
we obtain that, for every \( v \in I^s \), each \( H_{\pm\pm}(v) \) constitutes a horizontal slice of \( R \), as depicted in Fig.~\ref{fig:f_III_four_horizontal_slices}. Note the distinction between Eq.~(\ref{eq:f III horizontal slices construction}) and Eq.~(\ref{eq:f III four horizontal slices direct prodcut definition}): owing to the nonzero coupling, the sets $H_{\pm\pm}(v)$ no longer exhibit a Cartesian product structure and must instead be constructed as the union of their intersection points with all $\Sigma^2(x,y)$ planes. 
Moreover,
\begin{equation}
f_{\rm III}(d(v)) \cap R = H_{++}(v) \cup H_{+-}(v) \cup H_{-+}(v) \cup H_{--}(v), \label{eq:f III four horizontal slices}
\end{equation}
that is, \( f_{\rm III}(d(v)) \cap R \) is the union of four mutually disjoint horizontal slices of \( R \), as illustrated in Fig.~\ref{fig:f_III_four_horizontal_slices}. 

\end{proof}

\begin{remark}
From Eq.~(\ref{eq:f III four horizontal slices}) and Proposition~\ref{Semi-conjugacy to the full shift on N symbols}, it follows immediately that there exists a nonempty compact invariant set \( \Lambda \subset \bigcap_{n \in \mathbb{Z}} f^n_{\rm III}(R) \) on which the restricted map \( f_{\rm III}|_{\Lambda} \) is semi-conjugate to the full shift on four symbols. An explicit statement of this result is given in \cite{Fujioka25}.
\end{remark}

\subsubsection{Type B: Singly Folded Horseshoe Described by Template \( \mathcal{F}^Z_{(\underline{X},Y)} \)}
\label{Type B}

In contrast, in the neighborhood of the Type~B anti-integrable limit—defined by \( a_0 = a_1 = a \to \infty \) while keeping the ratio \( c/\sqrt{a} = \gamma \) fixed with \( \gamma > 1 \)—the system produces a fundamentally different geometric configuration. Upon a suitable linear change of coordinates \( (x,y,z,w) \to (X,Y,Z,W) \), the folding operation collapses into a singly folded structure governed by the template \( \mathcal{F}^Z_{(\underline{X},Y)} \). Unlike the doubly folded horseshoe of Type~A, this configuration remains confined to a three-dimensional subspace of the full four-dimensional space and admits a simpler symbolic dynamics.

As shown in \cite{Fujioka25}, the Type‑B AI limit can be studied more conveniently after performing the following change of coordinates:
\begin{eqnarray}
	\left(\begin{array}{c}
	X\\
	Y\\
	Z\\
	W
	\end{array}\right)
	=\frac{1}{2}\left(\begin{array}{c}
	x+y\\
	x-y\\
	z+w\\
	z-w\\
	\end{array}\right).
	\label{eq:change_of_coordinate}
\end{eqnarray}
Under this transformation, Eq.~(\ref{eq:4D doubly folded Henon}) can be rewritten in the form
\begin{eqnarray}
\label{eq:map_tranformed}
	\left(\begin{array}{c}
    X_{n+1}\\
    Y_{n+1}\\
    Z_{n+1}\\
    W_{n+1}
  \end{array}\right)
  =F\left(\begin{array}{c}
    X_n\\
    Y_n\\
    Z_n\\
    W_n
  \end{array}\right)
  =\left(\begin{array}{c}
    A_0-(X_n^2+Y_n^2)-Z_n\\
    A_1-2X_nY_n-W_n+2cY_n\\
    X_n\\
    Y_n
  \end{array}\right)
\end{eqnarray}
where $$\displaystyle A_0=\frac{a_0+a_1}{2}, ~~~~~A_1=\frac{a_0-a_1}{2}. $$
The inverse map $F^{-1}$ is given by
\begin{eqnarray}
\label{eq:map_tranformed_inverse}
  \left(\begin{array}{c}
    X_{n-1}\\
    Y_{n-1}\\
    Z_{n-1}\\
    W_{n-1}
  \end{array}\right)
  =F^{-1}\left(\begin{array}{c}
    X_n\\
    Y_n\\
    Z_n\\
    W_n
  \end{array}\right)
  =\left(\begin{array}{c}
    Z_n\\
    W_n\\
    A_0-(Z_n^2+W_n^2)-X_n\\
   A_1-2Z_nW_n-Y_n+2cW_n\\
  \end{array}\right) .
\end{eqnarray}

To formalize the setting, we take the horizontal directions to be the \((X,Y)\)-subspace, which approximately align with the local expansion directions, and the vertical directions to be the \((Z,W)\)-subspace, which are roughly aligned with the local contraction directions. We define the threshold length \( R = 1 + \sqrt{1 + A_0} \), and introduce a square \( I^u = [-R,R]^2 \) in the horizontal \((X,Y)\)-subspace and a square \( I^s = [-R,R]^2 \) in the vertical \((Z,W)\)-subspace. The hypercube on which we model the action of \( F \) is then given by
\[
R_F = I^u \times I^s,
\]
and we describe its dynamics using paperfolding templates.

As in the Type‑A limit, we begin with a horizontal disk of \( R_F \), namely
\[
d(V) = I^u \times \{ V \},
\]
where \( V = (Z_V,W_V) \in I^s \), and study its image under the mapping \( F \).  
Using the identity \( F^{-1}(F(d(V))) = d(V) \), we obtain the following explicit representation:
\begin{equation}\label{eq:F d(V)}
F(d(V)) =
\left\lbrace (X,Y,Z,W) \,\middle|\, 
\begin{array}{l}
(Z,W) \in I^s,\\[4pt]
A_0 - (Z^2+W^2) - X = Z_V,\\[4pt]
A_1 - 2ZW - Y + 2cW = W_V
\end{array} \right\rbrace.
\end{equation}
From this expression, we see that \( F(d(V)) \) is a two‑dimensional surface parameterized by \((Z,W)\), with \((X,Y)\) determined uniquely from \((Z,W)\) by the second and third relations on the right-hand side of Eq.~(\ref{eq:F d(V)}).

The intersection of \( F(d(V)) \) with \( R_F \) is obtained by imposing additional bounds on \( X \) and \( Y \):
\begin{equation}\label{eq:F d(V) intersect R_F}
F(d(V))\cap R_F =
\left\lbrace (X,Y,Z,W) \,\middle|\, 
\begin{array}{l}
(X,Y,Z,W) \in R_F,\\[4pt]
A_0 - (Z^2 + W^2) - X = Z_V,\\[4pt]
A_1 - 2ZW - Y + 2cW = W_V
\end{array} \right\rbrace.
\end{equation}

Let \( I_{XZ} \) denote the square \( [-R, R]^2 \) in the \( (X,Z) \)-subspace, and \( I_{YW} \) the square \( [-R, R]^2 \) in the \( (Y,W) \)-subspace. By treating the second and third rows of the right-hand side of Eq.~(\ref{eq:F d(V) intersect R_F}) as independent constraints and taking their intersection, the set \( F(d(V)) \cap R_F \) can be expressed as the intersection of two sets, denoted by \( M(V) \) and \( N(V) \), which are defined as follows:
\begin{equation}\label{eq:M def short}
M(V) = \bigcup_{(X,Z)\in I_{XZ}} \left\lbrace (X,Z) \right\rbrace \times L(X,Z;V),
\end{equation}
\begin{equation}\label{eq:N def short}
N(V) = \bigcup_{(Y,W)\in I_{YW}} \left\lbrace (Y,W) \right\rbrace \times H(Y,W;V),
\end{equation}
where
\begin{equation}\label{eq:L}
L(X,Z;V) = \left\lbrace (Y,W) \,\middle|\, 
\begin{array}{l}
(Y,W) \in I_{YW},\\
A_1 - 2ZW - Y + 2cW = W_V
\end{array} \right\rbrace,
\end{equation}
and
\begin{equation}\label{eq:H}
H(Y,W;V) = \left\lbrace (X,Z) \,\middle|\, 
\begin{array}{l}
(X,Z) \in I_{XZ},\\
A_0 - (Z^2 + W^2) - X = Z_V
\end{array} \right\rbrace.
\end{equation}
Note that, by the definitions in Eqs.~(\ref{eq:M def short}) and~(\ref{eq:N def short}), we have the identity
\begin{eqnarray}
F(d(V)) \cap R_F \;&=\; M(V) \cap N(V)\nonumber \\ 
&=\left\lbrace (X,Y,Z,W) \in R_F \,\middle|\, 
\begin{array}{l}
(Y,W) \in L(X,Z;V),\\
(X,Z) \in H(Y,W;V)
\end{array} \right\rbrace. \label{eq:M intersect N}
\end{eqnarray}

Eq.~(\ref{eq:M intersect N}) makes it clear that the structure of \( F(d(V)) \cap R_F \) in the full four-dimensional space can be analyzed by studying the behavior of the sets \( L(X,Z;V) \) and \( H(Y,W;V) \) within their respective two-dimensional subspaces.

We now propose the following three bounds on the parameters \( A_0 \), \( A_1 \), and \( c \) of the map \( F \) (defined in Eq.~(\ref{eq:map_tranformed})). These bounds effectively specify a neighborhood of the Type-B AI limit in the parameter space \( (a_0, a_1, c) \):
\begin{eqnarray}
\label{eq:typeB_sufficient1}
&A_1 \le R < c,\\
\label{eq:typeB_sufficient2}
&R < A_0 - (W^*)^2 - R,\\
\label{eq:typeB_sufficient3}
&W^* \le R,
\end{eqnarray}
where
\[
W^* = \max\left( \left| \frac{2R - A_1}{2(c - R)} \right|,\ \left| \frac{-2R - A_1}{2(c - R)} \right| \right).
\]
We now show that, under these bounds, the image \( F(d(V)) \) exhibits a geometry that can be described by the paperfolding template \( \mathcal{F}^Z_{(X,\underline{Y})} \): a two-dimensional sheet lying in the \( (X,Y) \)-subspace, embedded in four dimensions, folded along the crease \( Y = 0 \), and stacked along the \( Z \)-direction. This structure becomes evident by analyzing the properties of the sets \( L(X,Z;V) \) and \( H(Y,W;V) \) within their respective subspaces. We begin with \( L(X,Z;V) \).

The second row of Eq.~(\ref{eq:L}) can be rewritten as
\begin{equation}\label{eq:W in terms of Y Z s}
W = \frac{Y - A_1 - s}{2(c - Z)},
\end{equation}
where \( s = -W_V \) is treated as a parameter constrained by \( |s| \leq R \). Condition~(\ref{eq:typeB_sufficient1}) ensures that the denominator remains nonzero for all $|Z|\leq R$. For clarity, define the linear function
\begin{equation}
\Gamma_W(Y;Z,s) = \frac{Y - A_1 - s}{2(c - Z)},
\end{equation}
where \( Z \) and \( s \) are regarded as parameters within bounds \( |Z|, |s| \leq R \). The set \( L(X,Z;V) \) is determined by the graph of \( \Gamma_W(Y;Z,s) \) in the \( (Y,W) \)-plane. Since \( c - Z > 0 \) for all $|Z|\leq R$ under the standing assumption~(\ref{eq:typeB_sufficient1}), the graph of \( \Gamma_W \) is a straight line with positive slope.

The values of \( \Gamma_W \) are uniformly bounded across all admissible triples \( (Y,Z,s) \in [-R,R]^3 \), with extremal values given by
\begin{eqnarray}
W^{\max} &=& \Gamma_W(Y;Z,s)\Big|_{(Y,Z,s) = (R,R,-R)} = \frac{2R - A_1}{2(c - R)},\\
W^{\min} &=& \Gamma_W(Y;Z,s)\Big|_{(Y,Z,s) = (-R,R,R)} = \frac{-2R - A_1}{2(c - R)}.
\end{eqnarray}
Hence,
\begin{equation}\label{eq:Y W uniform bounds}
W^{\min} \leq \Gamma_W(Y;Z,s) \leq W^{\max}
\end{equation}
for all \( (Y,Z,s) \in [-R,R]^3 \). 

\begin{figure}[th]
\centering
\includegraphics[width=0.5\linewidth]{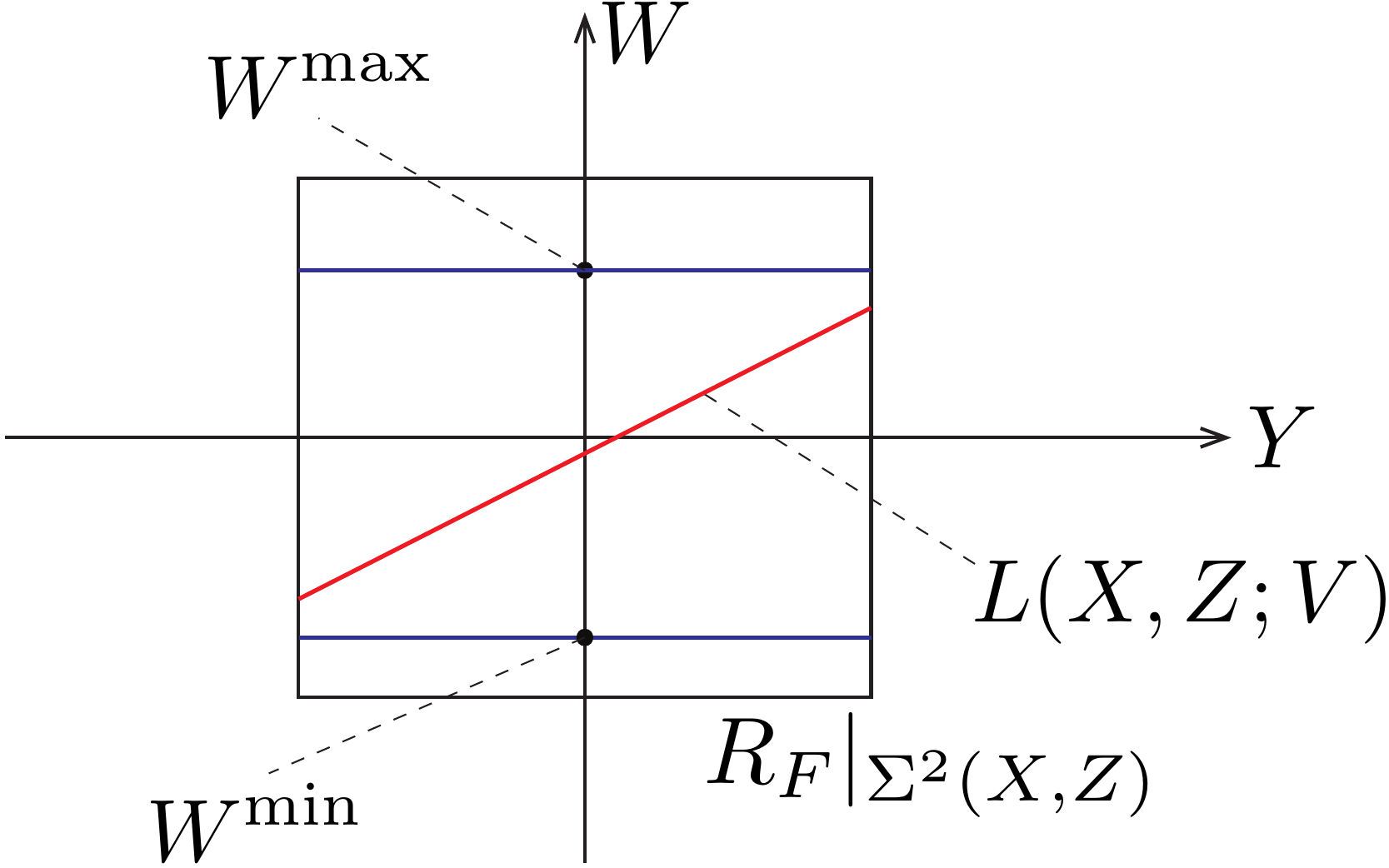}
\caption{(Schematic) For every \( (X,Z) \in [-R,R]^2 \) and every \( V \in I^s \), the set \( L(X,Z;V) \) (red) consists of a positively sloped straight-line segment bounded between two horizontal lines (blue) at heights \( W^{\max} \) and \( W^{\min} \), both of which satisfy \( |W^{\max}|, |W^{\min}| \leq R \). Consequently, \( L(X,Z;V) \) forms a horizontal slice of the square \( R_F|_{\Sigma^2(X,Z)} \).}
\label{fig:4D_singly_folded_horseshoe_Y_W}
\end{figure}

Condition~(\ref{eq:typeB_sufficient3}) implies that
\[
W^* = \max\left( |W^{\min}|, |W^{\max}| \right) \leq R,
\]
so that the bounds in Eq.~(\ref{eq:Y W uniform bounds}) can be further extended as
\begin{equation}\label{eq:Y W uniform bounds further}
-R \leq -W^* \leq W^{\min} \leq \Gamma_W(Y;Z,s) \leq W^{\max} \leq W^* \leq R.
\end{equation}
In other words, for every \( (X,Z) \in [-R,R]^2 \) and every \( s = -W_V \in [-R,R] \), the graph of \( \Gamma_W(Y;Z,s) \) intersects the vertical edges of the square \( R_F|_{\Sigma^2(X,Z)} \), forming a straight line across its entire \( Y \)-width. As shown in Fig.~\ref{fig:4D_singly_folded_horseshoe_Y_W}, the set \( L(X,Z;V) \) forms a horizontal slice of \( R_F|_{\Sigma^2(X,Z)} \) for each such \( (X,Z) \) and every \( V \in I^s \); that is, \( L(X,Z;V) \) spans the full \( Y \)-width of the square \( R_F|_{\Sigma^2(X,Z)} \). This property will be useful in later analysis.

Next, we analyze the geometric structure of the set \( H(Y,W;V) \), which, according to Eq.~(\ref{eq:F d(V) intersect R_F reexpress}), governs the behavior of \( F(d(V)) \cap R_F \) within the \( (X,Z) \)-subspace.

The second row of Eq.~(\ref{eq:H}) can be rewritten as
\begin{equation}\label{eq:X Z relation}
X = -Z^2 - W^2 + A_0 - Z_V.
\end{equation}
For clarity, define the function
\begin{equation}
\Gamma_X(Z;W,Z_V) = -Z^2 - W^2 + A_0 - Z_V,
\end{equation}
with \( W \) and \( Z_V \) considered as parameters satisfying \( |W|, |Z_V| \leq R \). The set \( H(Y,W;V) \) is thus described by the graph of \( \Gamma_X(Z;W,Z_V) \) in the \( (X,Z) \)-plane.

We now analyze the family of parabolas given by the graphs of \( \Gamma_X(Z;W,Z_V) \), parameterized by \( (W,Z_V) \in [-R,R]^2 \), and demonstrate that, under the constraints specified in Eqs.~(\ref{eq:typeB_sufficient1})--(\ref{eq:typeB_sufficient3}), each member of this family intersects the square \( R_F|_{\Sigma^2(Y,W)} \) in two disjoint horizontal slices.

\begin{figure}[th]
\centering
\includegraphics[width=0.7\linewidth]{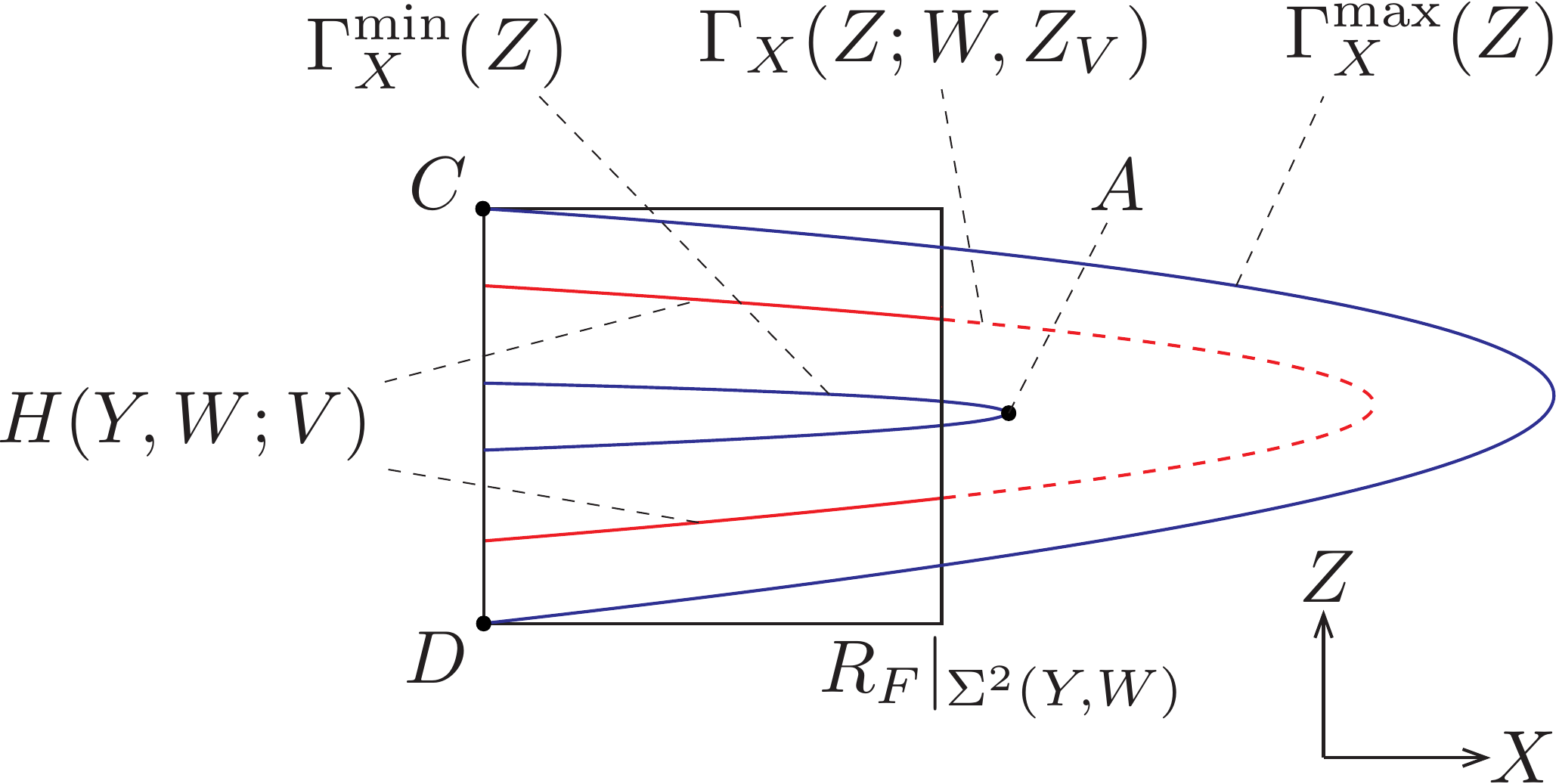}
\caption{(Schematic) The family of parabolas \( \Gamma_X(Z;W,Z_V) \) (solid and dashed red) is bounded on the left by \( \Gamma_X^{\min}(Z) \) (blue) and on the right by \( \Gamma_X^{\max}(Z) \) (blue). The constraints in Eqs.~(\ref{eq:typeB_sufficient1})--(\ref{eq:typeB_sufficient3}) ensure that the vertex \( A \) of \( \Gamma_X^{\min}(Z) \) lies outside the square \( R_F|_{\Sigma^2(Y,W)} \) on its right-hand side, while \( \Gamma_X^{\max}(Z) \) intersects the square at its two corners, labeled \( C \) and \( D \). As a result, each \( \Gamma_X(Z;W,Z_V) \) intersects \( R_F|_{\Sigma^2(Y,W)} \) in two disjoint horizontal slices, whose union forms the set \( H(Y,W;V) \) (solid red).}
\label{fig:4D_singly_folded_horseshoe_X_Z}
\end{figure}

The leftmost and rightmost parabolas in the family, denoted \( \Gamma_X^{\min} \) and \( \Gamma_X^{\max} \), respectively, are given by:
\begin{eqnarray}
\Gamma_X^{\min}(Z) &=& \Gamma_X(Z;W,Z_V)\big|_{(W,Z_V) = (W^*, R)} = -Z^2 - (W^*)^2 + A_0 - R, \label{eq:Gamma X min} \\
\Gamma_X^{\max}(Z) &=& \Gamma_X(Z;W,Z_V)\big|_{(W,Z_V) = (0, -R)} = -Z^2 + A_0 + R. \label{eq:Gamma X max}
\end{eqnarray}

We now establish that, under the constraints specified by Eqs.~(\ref{eq:typeB_sufficient1})--(\ref{eq:typeB_sufficient3}), the following two properties are satisfied:
\begin{enumerate}
\item[(a)] The vertex of \( \Gamma_X^{\min}(Z) \), labeled \( A \) in Fig.~\ref{fig:4D_singly_folded_horseshoe_X_Z}, lies outside the square \( R_F|_{\Sigma^2(Y,W)} \), on its right-hand side;
\item[(b)] The parabola \( \Gamma_X^{\max}(Z) \) intersects the left edge of \( R_F|_{\Sigma^2(Y,W)} \) at two distinct points, labeled \( C \) and \( D \) in Fig.~\ref{fig:4D_singly_folded_horseshoe_X_Z}.
\end{enumerate}

\textbf{Property (a):} The vertex \( A \) of \( \Gamma_X^{\min}(Z) \) is given by
\begin{equation}\label{eq:Vertex Gamma X min}
(X_A, Z_A) = \left( A_0 - (W^*)^2 - R,\ 0 \right).
\end{equation}
By Eq.~(\ref{eq:typeB_sufficient2}), we immediately find that \( X_A > R \), meaning that the vertex \( A \) lies outside the square \( R_F|_{\Sigma^2(Y,W)} \), on its right-hand side.

\textbf{Property (b):} This is equivalent to requiring that \( \Gamma_X^{\max}(Z = \pm R) \leq -R \). Indeed,
\begin{equation}\label{eq:Gamma X max S X - intersections}
\Gamma_X^{\max}(Z = \pm R) = -R^2 + A_0 + R = -R \leq -R,
\end{equation}
where the second equality follows from the identity \( R = 1 + \sqrt{1 + A_0} \). This confirms property (b). In fact, Eq.~(\ref{eq:Gamma X max S X - intersections}) shows that the parabola \( \Gamma_X^{\max}(Z) \) intersects \( R_F|_{\Sigma^2(Y,W)} \) precisely at its two corner points, labeled \( C \) and \( D \) in Fig.~\ref{fig:4D_singly_folded_horseshoe_X_Z}.

Therefore, properties (a) and (b) hold under the assumptions given in Eqs.~(\ref{eq:typeB_sufficient1})--(\ref{eq:typeB_sufficient3}). Since each member of the family \( \Gamma_X(Z;W,Z_V) \) lies between the bounding curves \( \Gamma_X^{\min}(Z) \) and \( \Gamma_X^{\max}(Z) \),
\begin{equation}
\Gamma_X^{\min}(Z) \leq \Gamma_X(Z;W,Z_V) \leq \Gamma_X^{\max}(Z),
\end{equation}
it follows from the geometry illustrated in Fig.~\ref{fig:4D_singly_folded_horseshoe_X_Z} that every \( \Gamma_X(Z;W,Z_V) \) intersects the square \( R_F|_{\Sigma^2(Y,W)} \) in two disjoint horizontal slices. The union of these slices forms the set \( H(Y,W;V) \).

We now synthesize the above analysis. Equation~(\ref{eq:F d(V) intersect R_F reexpress}) characterizes \( F(d(V)) \cap R_F \) as the set of all points \( (X,Y,Z,W) \in R_F \) satisfying the pair of constraints
\[
(Y,W) \in L(X,Z;V) \quad {\rm and} \quad (X,Z) \in H(Y,W;V),
\]
simultaneously. Geometrically, this means that each point in \( F(d(V)) \cap R_F \) lies at the intersection of a horizontal slice \( L(X,Z;V) \subset \Sigma^2(X,Z) \) and two disjoint horizontal slices \( H(Y,W;V) \subset \Sigma^2(Y,W) \) arising from a folded one-dimensional sheet, with each structure encoded in its respective two-dimensional subspace.

In particular, for each fixed \( (X,Z) \in I_{XZ} \), the set \( L(X,Z;V) \subset \Sigma^2(X,Z) \) is a straight segment spanning the entire \( Y \)-width of the square \( R_F|_{\Sigma^2(X,Z)} \), with no folding. Thus, the projection of \( F(d(V)) \) onto the \( (Y,W) \)-subspace yields a one-dimensional sheet aligned along the \( Y \)-axis, intersecting each \( R_F|_{\Sigma^2(X,Z)} \) in a single horizontal slice, as shown in Fig.~\ref{fig:4D_singly_folded_horseshoe_Y_W}. In contrast, for each fixed \( (Y,W) \in I_{YW} \), the set \( H(Y,W;V) \subset \Sigma^2(Y,W) \) consists of two disjoint horizontal slices of \( R_F|_{\Sigma^2(Y,W)} \), reflecting a fold along the \( X \)-axis, with stacking directed along the \( Z \)-axis, as shown schematically in Fig.~\ref{fig:4D_singly_folded_horseshoe_X_Z}.

Together, these observations imply that the four-dimensional image \( F(d(V)) \) is modeled pointwise by the Cartesian product of a straight segment in the \( Y \)-direction and a folded curve in the \( X \)-direction. This configuration is precisely described by the paperfolding template \( \mathcal{F}^Z_{(\underline{X},Y)} \), in which a two-dimensional sheet lying in the \( (X,Y) \)-subspace is embedded in four dimensions, folded along the crease \( X = 0 \), and stacked in the \( Z \)-direction. The coordinate \( Y \) spans the unfolded direction of the sheet, while the fourth coordinate \( W \) remains unaffected by the folding.

Therefore, under the parameter bounds~(\ref{eq:typeB_sufficient1})--(\ref{eq:typeB_sufficient3}), the image \( F(d(V)) \) is geometrically described by the template \( \mathcal{F}^Z_{(\underline{X},Y)} \), completing the characterization of the Type‑B singly folded horseshoe.

We can also infer properties of the symbolic dynamics from the paperfolding structure. The template \( \mathcal{F}^Z_{(\underline{X},Y)} \) indicates that the image \( F(d(V)) \) consists of two folded halves, each spanning the entire \( (X,Y) \)-directions of \( R_F \), and approximately stacked along the \( Z \)-direction. This implies that \( F(d(V)) \) intersects \( R_F \) in two disjoint horizontal slices of \( R_F \) that are approximately aligned along the \( (X,Y) \)-plane and approximately stacked in the \( Z \)-direction. We remind the reader that the horizontal direction in the full four-dimensional space \( (X,Y,Z,W) \) is defined to be the \( (X,Y) \)-subspace.

Finally, by invoking Proposition~\ref{Semi-conjugacy to the full shift on N symbols}, we conclude that there exists a nonempty compact invariant set \( \Lambda_F \subset \bigcap_{n\in \mathbb{Z}} F^n(R_F) \), on which the restricted map \( F|_{\Lambda_F} \) is semi-conjugate to the full shift on two symbols \cite{Fujioka25}.

\subsection{Map \( f_{\mathrm{IV}} \): triply folded three-dimensional sheet in four dimensions}
\label{f IV}

Finally, we introduce a four-dimensional Hénon-type map \( f_{\mathrm{IV}} \) that generates a highly nontrivial geometric structure: a three-dimensional sheet undergoing three successive foldings along orthogonal planes, all stacked along the \( w \)-direction. The resulting horseshoe configuration is described by the composition \( \mathcal{F}^w_{(x,y,\underline{z})} \circ \mathcal{F}^w_{(x,\underline{y},z)} \circ \mathcal{F}^w_{(\underline{x},y,z)} \), as introduced in Section~\ref{Folding a three-dimensional sheet in four dimensions}. This triply folded structure is intrinsically four-dimensional and exemplifies the kind of complex folding behavior made possible by additional phase-space dimensions.

$f_{\rm IV}$ is a H\'{e}non-type map in four dimensions defined by
\begin{equation}\label{eq:4D triply folded Henon}
\left( \begin{array}{ccc}
x_{n+1}\\
y_{n+1} \\
z_{n+1} \\
w_{n+1} \end{array} \right) = 
f_{\rm IV} \left( \begin{array}{ccc}
x_{n}\\
y_{n} \\
z_{n} \\
w_{n} \end{array} \right) =
\left( \begin{array}{ccc}
a_0 - x^2_n - w_n\\
a_1 - y^2_n - x_n\\
a_2 - z^2_n -y_n\\
z_n \end{array} \right)
\end{equation}
with parameters $a_0,a_1,a_2 > 5+2\sqrt{5}$. $f_{\rm IV}$ can be written as the compound mapping of three successive maps, denoted by $f_1$, $f_2$, and $f_3$, respectively:
\begin{equation}\label{eq:4D triply folded Henon compound}
f_{\rm IV} =f_3 \circ f_2 \circ f_1\ ,
\end{equation}
where $f_1$ is expressed as
\begin{equation}\label{eq:f_1}
\left( \begin{array}{ccc}
x^{\prime}\\
y^{\prime} \\
z^{\prime} \\
w^{\prime} \end{array} \right) = 
f_1 \left( \begin{array}{ccc}
x\\
y \\
z \\
w \end{array} \right) =
\left( \begin{array}{ccc}
a_0 - x^2 - w\\
y \\
z \\
x \end{array} \right)\ ,
\end{equation}
$f_2$ is expressed as
\begin{equation}\label{eq:f_2}
\left( \begin{array}{ccc}
x^{\prime}\\
y^{\prime} \\
z^{\prime} \\
w^{\prime} \end{array} \right) = 
f_2 \left( \begin{array}{ccc}
x\\
y \\
z \\
w \end{array} \right) =
\left( \begin{array}{ccc}
x\\
a_1 - y^2 - w\\
z \\
y \end{array} \right)\ ,
\end{equation}
and $f_3$ is expressed as
\begin{equation}\label{eq:f_3}
\left( \begin{array}{ccc}
x^{\prime}\\
y^{\prime} \\
z^{\prime} \\
w^{\prime} \end{array} \right) = 
f_3 \left( \begin{array}{ccc}
x\\
y \\
z \\
w \end{array} \right) =
\left( \begin{array}{ccc}
x\\
y\\
a_2 - z^2 - w\\
z \end{array} \right)\ .
\end{equation}

The inverse map, $f^{-1}_{\rm IV}$, is expressed as
\begin{eqnarray}\label{eq:4D triply folded Henon inverse}
\left( \begin{array}{ccc}
x_{n}\\
y_{n} \\
z_{n} \\
w_{n} \end{array} \right) &= 
f^{-1}_{\rm IV} \left( \begin{array}{ccc}
x_{n+1}\\
y_{n+1} \\
z_{n+1} \\
w_{n+1} \end{array} \right) \nonumber \\
&=
\left( \begin{array}{ccc}
a_1 - y_{n+1} - (a_2 - z_{n+1} - w^2_{n+1})^2\\
a_2 - z_{n+1} - w^2_{n+1}\\
w_{n+1}\\
a_0 - x_{n+1} - \left[ a_1 - y_{n+1} - (a_2 - z_{n+1} - w^2_{n+1})^2 \right]^2 \end{array} \right).
\end{eqnarray}

\begin{figure}[thbp]
\includegraphics[width=1\linewidth]{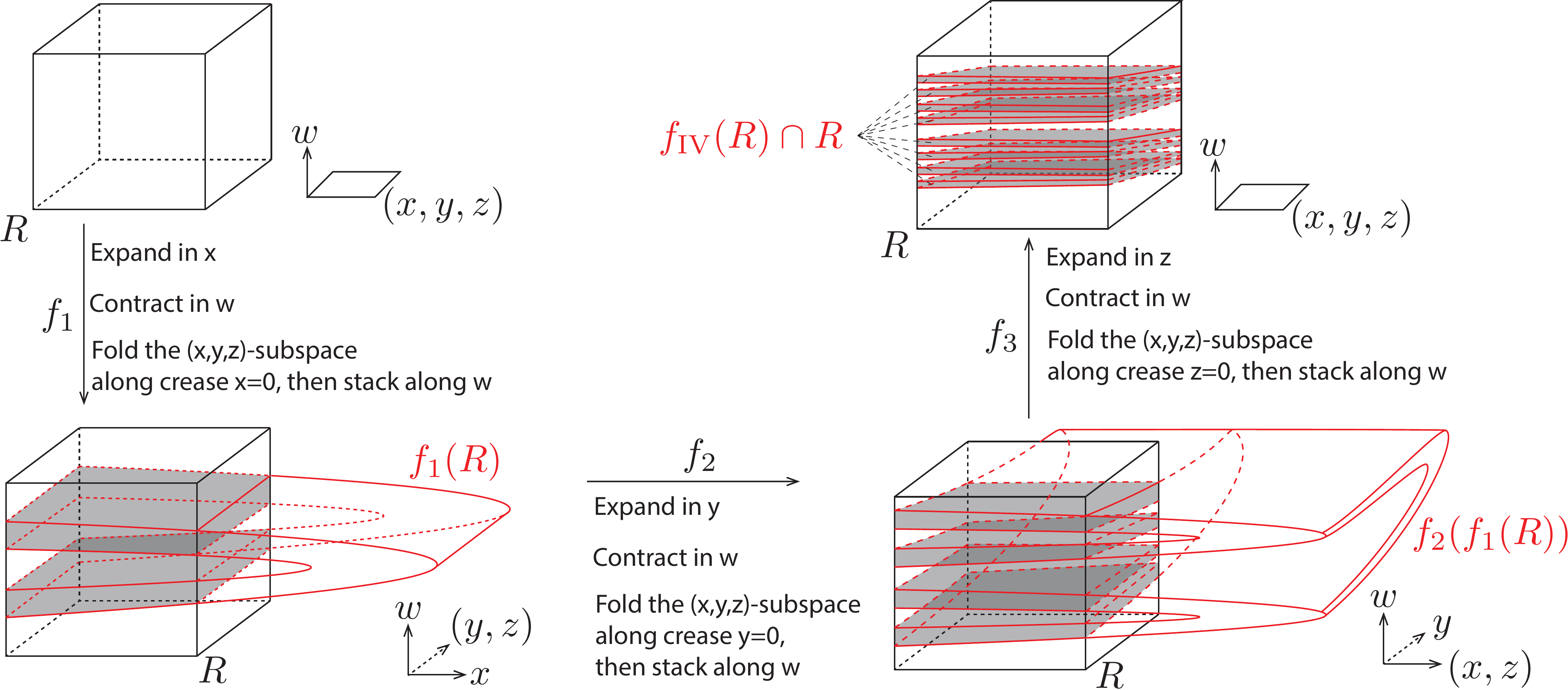} 
\caption{(Schematic, color online) \textbf{Upper left:} A specifically chosen hypercube \( R \) of initial conditions. \textbf{Lower left:} The map \( f_1 \) contracts \( R \) along the \( w \)-direction to ``flatten'' it into a quasi-three-dimensional sheet in the \((x, y, z)\)-subspace, expands the sheet along \( x \), and folds it along the crease \( x = 0 \) (the \((y, z)\)-plane), with the folded halves stacked along the \( w \)-direction. The intersection \( f_1(R) \cap R \) consists of two horizontal slabs (shaded). \textbf{Lower right:} The map \( f_2 \) contracts the resulting structure along \( w \), expands it along \( y \), and folds it along the crease \( y = 0 \) (the \((x, z)\)-plane), again stacking the folded halves along \( w \). The intersection \( f_2(f_1(R)) \cap R \) consists of four horizontal slabs (shaded). \textbf{Upper right:} The map \( f_3 \) contracts the structure along \( w \), expands it along \( z \), and folds it along the crease \( z = 0 \) (the \((x, y)\)-plane), with the folded halves stacked along \( w \). The final intersection \( f_{\rm IV}(R) \cap R \) consists of eight horizontal slabs (shaded).} \label{fig:f_IV_schematic}
\end{figure}

Evidently, the restrictions of $f_1$, $f_2$, and $f_3$ to the $(x,w)$-, $(y,w)$-, and $(z,w)$-subspaces, respectively, are just two-dimensional H\'{e}non maps well-studied in the literature. Fig.~\ref{fig:f_IV_schematic} qualitatively illustrates the action of the map $f_{\rm IV}=f_3\circ f_2\circ f_1$ on a four-dimensional hypercube $R$ (to be specified later). The map $f_1$ first expands $R$ along the $x$-direction and contracts it along the $w$-direction, effectively ``flattening" it into a quasi-three-dimensional sheet lying in the $(x,y,z)$-hyperplane and embedded in four-dimensional space. The $w$-axis thus serves as the normal or ``stacking" direction. Next, $f_1$ folds the flattened structure along the crease $x=0$ (i.e., the $(y,z)$-hyperplane), stacking the resulting halves along the $w$-axis while leaving the $(y,z)$-components unchanged. From the discussion in Sec.~\ref{Folding a three-dimensional sheet in four dimensions}, it is evident that the deformation of $R$ under $f_1$ can be described qualitatively by the paperfolding template $\mathcal{F}^w_{(\underline{x},y,z)}$. As illustrated in the lower-left portion of Fig.~\ref{fig:f_IV_schematic}, this operation creates two horizontal slabs of $R$ that span the $(x,y,z)$-subspace and are stacked vertically along the $w$-direction.

The map $f_2$ then acts on the folded structure: it expands it along $y$, contracts it along $w$, and folds it along the crease $y=0$ (the $(x,z)$-hyperplane), again stacking the folded halves along $w$ and leaving the $(x,z)$-components unchanged. From the discussion in Sec.~\ref{Folding a three-dimensional sheet in four dimensions}, it is evident that the deformation of $R$ under $f_2$ can be described qualitatively by the paperfolding template $\mathcal{F}^w_{(x,\underline{y},z)}$. As illustrated in the lower-right portion of Fig.~\ref{fig:f_IV_schematic}, $f_2$ doubles the number of horizontal slabs already present in $f_1(R)\cap R$, resulting in four horizontal slabs of $R$ that span the $(x,y,z)$-subspace and are stacked vertically along the $w$-direction.

Finally, $f_3$ expands the resulting structure along $z$, contracts it along $w$, and folds it along the crease $z=0$ (the $(x,y)$-hyperplane), stacking the halves along $w$ and preserving the $(x,y)$-components. The deformation of \( R \) under \( f_3 \) can be described qualitatively by the paperfolding template \( \mathcal{F}^w_{(x,y,\underline{z})} \), while the overall deformation of \( R \) under the composite map \( f_{\mathrm{IV}} = f_3 \circ f_2 \circ f_1 \) corresponds to the composition of paperfolding operations  
\begin{equation}
\mathcal{F}^w_{(x,y,\underline{z})} \circ \mathcal{F}^w_{(x,\underline{y},z)} \circ \mathcal{F}^w_{(\underline{x},y,z)}.
\end{equation}
As illustrated in the upper-right portion of Fig.~\ref{fig:f_IV_schematic}, the action of \( f_3 \) doubles the number of horizontal slabs already present in \( f_2(f_1(R)) \cap R \), producing eight disjoint horizontal slabs of \( R \) that span the \((x,y,z)\)-subspace and are stacked vertically along the \( w \)-direction. Here we emphasize that these eight disjoint horizontal slabs can be regarded as the result of the paperfolding operation \(\mathcal{F}^w_{(x,y,\underline{z})} \circ \mathcal{F}^w_{(x,\underline{y},z)} \circ \mathcal{F}^w_{(\underline{x},y,z)}\). Therefore, verifying their existence---at least in the strongly chaotic regime where the parameters satisfy
\[
a_0 = a_1 = a_2 = a > 5 + 2\sqrt{5},
\]
will provide strong evidence for the validity of the paperfolding template in this regime. This is precisely what we establish in the next theorem.

To formalize the setting, we take the horizontal directions to be the \((x,y,z)\)-subspace, which is approximately aligned with the local expansion directions, and the vertical direction to be the \(w\)-axis, which is roughly aligned with the local contraction direction.  
Let  
\[
r = 1 + \sqrt{1+a}
\]  
and define \(I^u = [-r,r]^3\) as the cube of side length \(2r\) in the \((x,y,z)\)-subspace, and \(I^s = [-r,r]\) as the interval of length \(2r\) in the \(w\)-subspace.  
The four-dimensional hypercube \(R = I^u \times I^s\) specifies the local region whose deformation under \( f_{\mathrm{IV}} \) is of interest.  

For \(v \in I^s\), let  
\[
d(v) = I^u \times \{ v \} \subset R
\]  
be a horizontal disk of \(R\).  
Geometrically, \(d(v)\) is a three-dimensional ``sheet'' embedded in four dimensions, and we examine how this sheet is deformed under \( f_{\mathrm{IV}} \).

\begin{theorem}\label{4D triply folded horseshoe topology}
In the parameter regime $a_0 = a_1 = a_2 = a > 5 + 2\sqrt{5}$, with $r$ and $R$ defined as above, the following holds: for every horizontal disk \( d(v) \subset R \) with \( v \in I^s \), the intersection \( f_{\mathrm{IV}}(d(v)) \cap R \) consists of eight disjoint horizontal slices of \( R \).
\end{theorem}

\begin{proof}
Let $v=(x_v,y_v,z_v,w_v)$ be an arbitrary point in $I^s$. Using the identify $f^{-1}_{\rm IV}(f_{\rm IV}(d(v)))=d(v)$, we obtain an expression for $f_{\rm IV}(d(v))$
\begin{equation}\label{eq:f IV V}
f_{\rm IV}(d(v)) = 
\left\lbrace (x,y,z,w) \middle| \begin{array}{ccc}
|w| \leq r \\
| a - z - w^2 | \leq r \\
| a - y - (a - z - w^2)^2 | \leq r \\
a - x - \left[ a - y - (a - z - w^2)^2 \right]^2  = w_v \end{array} \right\rbrace. 
\end{equation}
The expression for $f_{\rm IV}(d(v))\cap R$ is obtained by imposing the additional bounds on $x$, $y$, and $z$:
\begin{equation}\label{eq:V intersect f IV V}
f_{\rm IV}(d(v))\cap R = 
\left\lbrace (x,y,z,w) \middle| \begin{array}{ccc}
|x|,|y|,|z|,|w| \leq r \\
| a - z - w^2 | \leq r \\
| a - y - (a - z - w^2)^2 | \leq r \\
a - x - \left[ a - y - (a - z - w^2)^2 \right]^2  = w_v \end{array} \right\rbrace. 
\end{equation}

Let \(\Sigma^3(x)\) denote the \((y,z,w)\)-hyperplane at fixed \(x\), where the superscript “3” indicates the dimensionality of the hyperplane. Explicitly,  
\begin{equation}\label{eq: y z w slice for f_IV}
\Sigma^3(x) = \left\lbrace (x',y',z',w') \,\middle|\, y',z',w' \in \mathbb{R},\ x' = x \right\rbrace.
\end{equation}
Similarly, let \(\Sigma^2(x,y)\) denote the \((z,w)\)-plane at fixed \( (x,y) \),  
\begin{equation}\label{eq: z w slice for f_IV}
\Sigma^2(x,y) = \left\lbrace (x',y',z',w') \,\middle|\, z',w' \in \mathbb{R},\ (x',y') = (x,y) \right\rbrace.
\end{equation}

To prove Theorem~\ref{4D triply folded horseshoe topology}, it suffices to consider the intersection of $f_{\mathrm{IV}}(d(v)) \cap R$ with the section hyperplane $\Sigma^3(x)$ for $x \in [-r,r]$, and to show that, for each such $x$, the restriction $f_{\mathrm{IV}}(d(v)) \cap R|_{\Sigma^3(x)}$ consists of eight disjoint horizontal slices of $R|_{\Sigma^3(x)}$, where the horizontal directions within $\Sigma^3(x)$ are the $(y,z)$-directions. Furthermore, this reduces to verifying that, for every $(x,y) \in [-r,r]^2$, the restriction $f_{\mathrm{IV}}(d(v)) \cap R|_{\Sigma^2(x,y)}$ consists of eight disjoint horizontal slices of $R|_{\Sigma^2(x,y)}$, where the horizontal direction within $\Sigma^2(x,y)$ is the $z$-direction. We now proceed to establish this property.

The second row of Eq.~(\ref{eq:f IV V}) can be rewritten in the parameterized form
\begin{equation}\label{eq:z w parabola 1 f_V}
z = -w^2 + a + s, ~~~~{\rm where}~ |s| \leq r \ .
\end{equation}
Let $\Gamma_z(w,s)$ be the family of parabolas in $\Sigma^2(x,y)$
\begin{equation}\label{eq:Gamma z in terms of w s f_V}
\Gamma_z(w,s) = -w^2 + a + s
\end{equation}
where $s$ is viewed as a parameter within range $|s| \leq r$. It is obvious that $\Gamma_z(w,s)$ is bounded by
\begin{equation}\label{eq:Gamma z in terms of w s bounds}
\Gamma^{\min}_z(w) \leq \Gamma_z(w,s) \leq \Gamma^{\max}_z(w)
\end{equation}
with lower and upper bounds
\begin{eqnarray}
& \Gamma^{\min}_z(w) = \Gamma_z(w,s)|_{s=-r} = -w^2 + a -r  \label{eq:Gamma z w min} \\
& \Gamma^{\max}_z(w) =  \Gamma_z(w,s)|_{s=r} = -w^2 + a + r  \label{eq:Gamma z w max}\ .
\end{eqnarray}

In each $\Sigma^2(x,y)$ slice, the restriction $f_{\mathrm{IV}}(d(v)) \cap R|_{\Sigma^2(x,y)}$ is contained within the gap between the two parabolas $\Gamma^{\min}_z(w)$ and $\Gamma^{\max}_z(w)$.

\begin{figure}[th]
        \centering
        \includegraphics[width=0.6\linewidth]{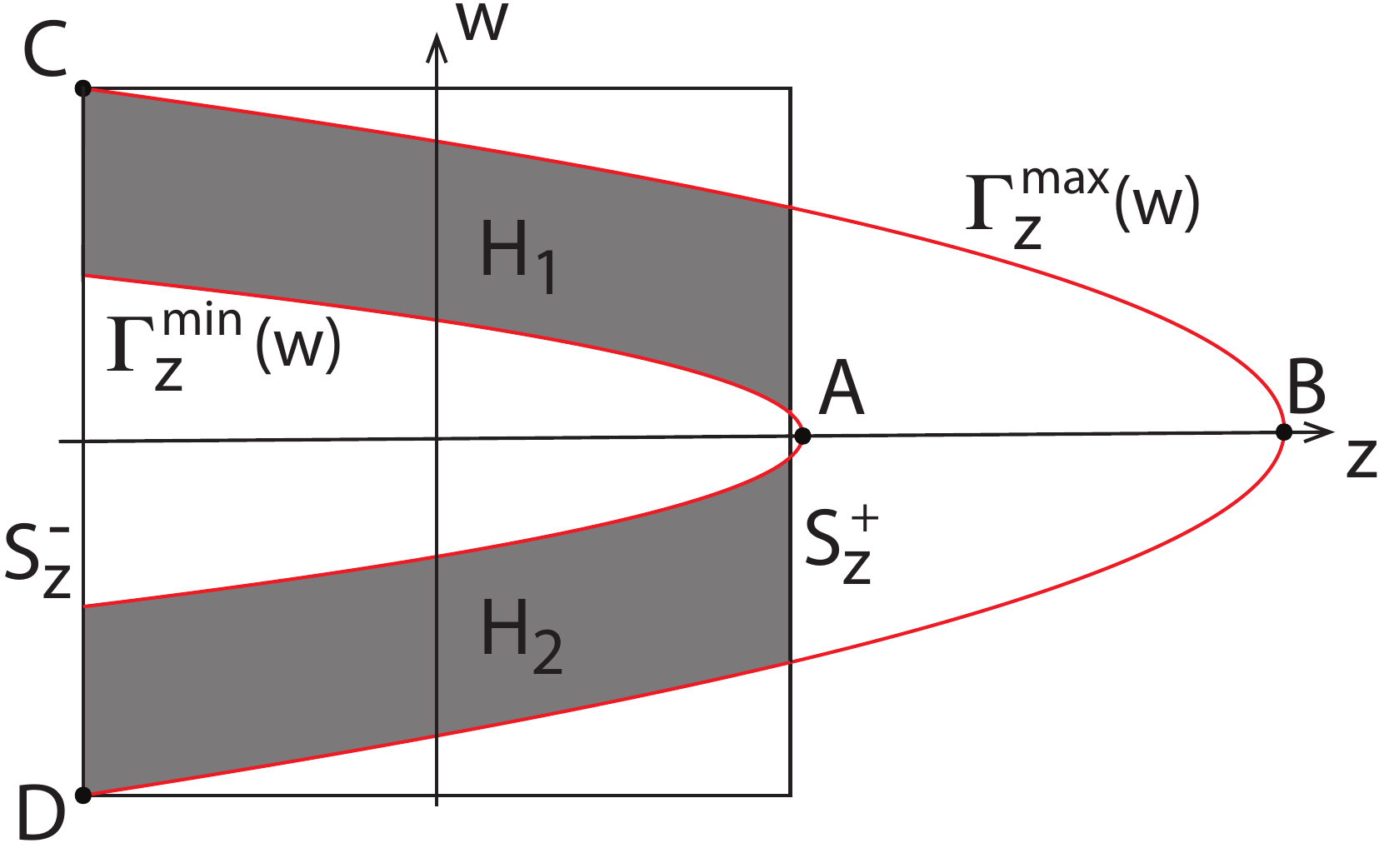} 
        \caption{(Schematic, color online) The gap bounded by $\Gamma^{\min}_z(w)$, $\Gamma^{\max}_z(w)$ , and $S^{\pm}_z$ consists of two disjoint horizontal strips $H_1$ and $H_2$.}   \label{fig:4D_triply_folded_horseshoe_z_w_slice_Gammas}
\end{figure}

Let $S^{\pm}_z = \lbrace (z,w) | z=\pm r, |w|\leq r \rbrace$. The possible location of the two parabolas can be further narrowed down by proving the following facts:
\begin{itemize}
\item[(a)] The vertex of $\Gamma^{\min}_z(w)$ is located on the right-hand side of $S^+_z$, as labeled by $A$ in Fig.~\ref{fig:4D_triply_folded_horseshoe_z_w_slice_Gammas};
\item[(b)] $\Gamma^{\max}_z(w)$ intersects $S^-_z$ at two points, as labeled by $C$ and $D$ in Fig.~\ref{fig:4D_triply_folded_horseshoe_z_w_slice_Gammas}. 
\end{itemize} 

To prove (a), let $A=(z_A,w_A)$. It can be solved easily that
\begin{equation}
w_A=0, ~~~~ z_A = \Gamma^{\min}_z(w_A=0)=a-r. 
\end{equation}
Using the assumption that $a>5+2\sqrt{5}$, it is straightforward to verify that $a-2r>0$, thus $z_A >r$, i.e., $A$ is on the right-hand side of $S^{+}_z$. 

To prove (b), notice that
\begin{equation}
\Gamma^{\max}_z(w=\pm r) = -r^2 +a +r = -r\ ,
\end{equation}
thus $C$ and $D$ are located at
\begin{equation}
C = (-r,r), ~~~~ D=(-r,-r)\ ,
\end{equation}
i.e., $C$ and $D$ are the upper-left and lower-left corners of $V$, respectively, as labeled in Fig.~\ref{fig:4D_triply_folded_horseshoe_z_w_slice_Gammas}. Therefore, $\Gamma^{\max}_z(w)$ intersects $S^-_z$ at its two endpoints. 

Combining (a) and (b), we know that the region bounded by $\Gamma^{\max}_z(w)$, $\Gamma^{\min}_z(w)$, and $S^{\pm}_z$ consists of two disjoint horizontal slabs, labeled by $H_1$ and $H_2$ in Fig.~\ref{fig:4D_triply_folded_horseshoe_z_w_slice_Gammas}. Strictly speaking, both $H_1$ and $H_2$ depend on $x$ and $y$, i.e., the position of $(z,w)$-slice in the $(x,y)$-subspace, therefore should be written as $H_1(x,y)$ and $H_2(x,y)$. However, since the $(x,y)$-dependence will not be used for the rest of the proof, we simply omit it and write the horizontal slabs without explicit $(x,y)$-dependence. When viewed in each $\Sigma^2(x,y)$, the restriction $f_{\rm IV}(d(v))\cap R|_{\Sigma^2(x,y)}$ lies within $H_1$ and $H_2$:
\begin{equation}\label{eq:V intersects f V V z w slice first bound}
f_{\rm IV}(d(v))\cap R \Big|_{\Sigma^2(x,y)} \subset H_1 \cup H_2 \ .
\end{equation}

At this point, let us notice that Eq.~(\ref{eq:V intersects f V V z w slice first bound}) only makes use of the second row of Eq.~(\ref{eq:f IV V}), thus only provides a crude bound for $f_{\rm IV}(d(v))\cap R|_{\Sigma^2(x,y)}$. Based upon Eq.~(\ref{eq:V intersects f V V z w slice first bound}), we now further refine the bound for $f_{\rm IV}(d(v))\cap R|_{\Sigma^2(x,y)}$ by imposing the third row of Eq.~(\ref{eq:f IV V}). 

The third row of Eq.~(\ref{eq:f IV V}) can be expressed in the parameterized form  
\begin{equation}\label{eq:z w parabola 2 f_V}
a - y - (a - z - w^2)^2 = -s, 
\qquad |s| \leq r,
\end{equation}
from which solving for \( z \) yields two branches:  
\begin{equation}\label{eq: z w parabola 2 f_V two branches}
z_{\pm}(w,y,s) = -w^2 + a \pm \sqrt{s - y + a}.
\end{equation}
We therefore define two families of parabolas in \(\Sigma^2(x,y)\), denoted \(\Lambda^{\pm}_z(w,y,s)\), by  
\begin{equation}\label{eq:Lambda z in terms of w y s}
\Lambda^{\pm}_z(w,y,s) = -w^2 + a \pm \sqrt{s - y + a},
\end{equation}
where \( y \) and \( s \) are treated as parameters satisfying \( |y|, |s| \leq r \).

In each \(\Sigma^2(x,y)\) slice, \(\Lambda^{+}_z(w,y,s)\) forms a family of parabolas parameterized by \(s\), bounded by  
\begin{equation}\label{eq:Lambda z + lower and upper bounds on each slice}
\Lambda^{+,1}_z(w,y) \leq \Lambda^{+}_z(w,y,s) \leq \Lambda^{+,2}_z(w,y),
\end{equation}  
where the lower and upper bounds are attained, respectively, at  
\begin{eqnarray}
& \Lambda^{+,1}_z(w,y) = \Lambda^+_z(w,y,s)\big|_{s=-r} = -w^2 + a + \sqrt{a-y-r}, \label{eq:Lambda z + 1} \\
& \Lambda^{+,2}_z(w,y) = \Lambda^+_z(w,y,s)\big|_{s=r}  = -w^2 + a + \sqrt{a-y+r}. \label{eq:Lambda z + 2}
\end{eqnarray}  

Similarly, in each \(\Sigma^2(x,y)\) slice, \(\Lambda^{-}_z(w,y,s)\) constitutes a family of parabolas parameterized by \(s\), bounded by  
\begin{equation}\label{eq:Lambda z - lower and upper bounds on each slice}
\Lambda^{-,1}_z(w,y) \leq \Lambda^{-}_z(w,y,s) \leq \Lambda^{-,2}_z(w,y),
\end{equation}  
with the lower and upper bounds attained, respectively, at  
\begin{eqnarray}
& \Lambda^{-,1}_z(w,y) = \Lambda^-_z(w,y,s)\big|_{s=r}  = -w^2 + a - \sqrt{a-y+r}, \label{eq:Lambda z - 1} \\
& \Lambda^{-,2}_z(w,y) = \Lambda^-_z(w,y,s)\big|_{s=-r} = -w^2 + a - \sqrt{a-y-r}. \label{eq:Lambda z - 2}
\end{eqnarray}  

\begin{figure}
        \centering
        \includegraphics[width=0.7\linewidth]{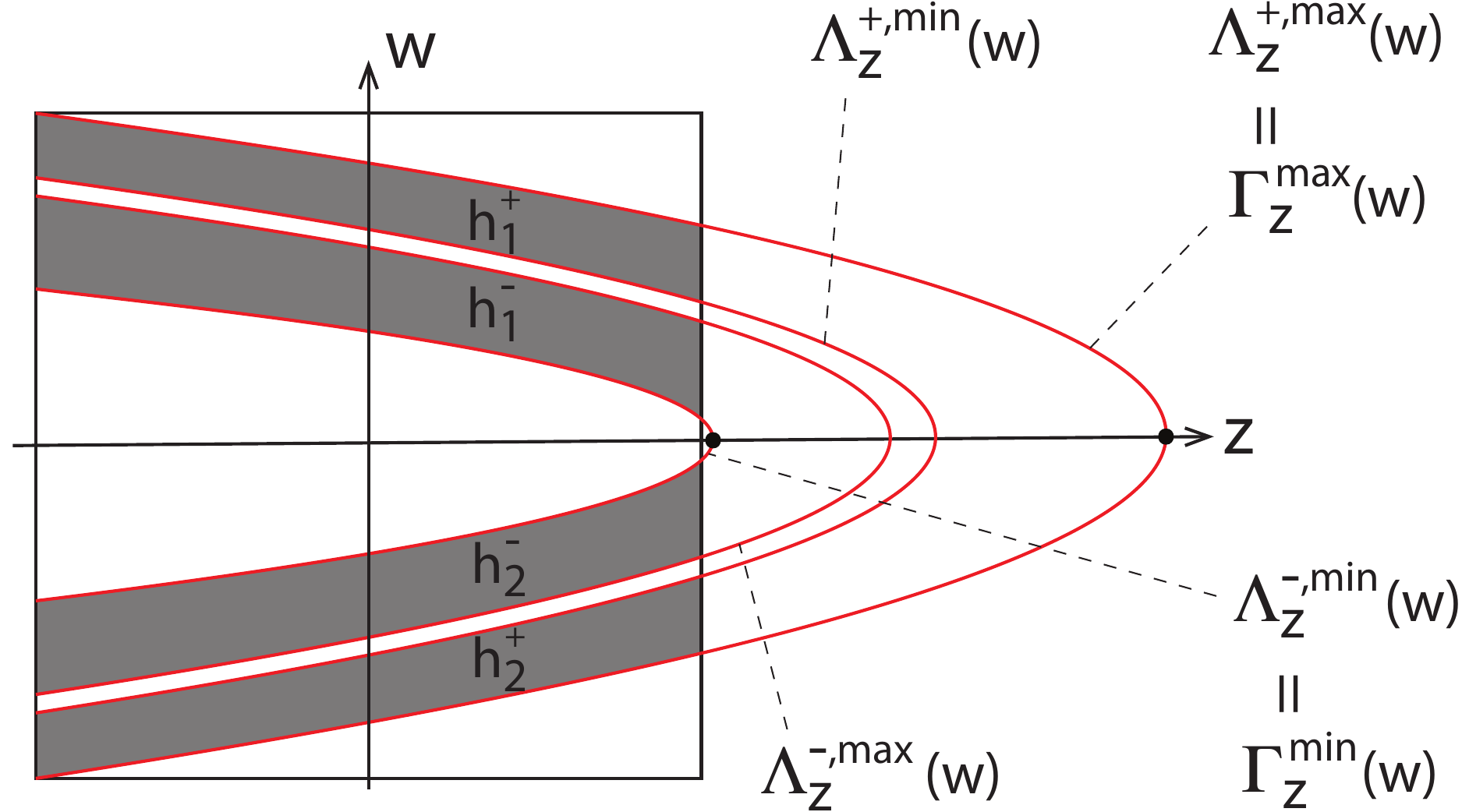} 
        \caption{(Schematic, color online) In each $\Sigma^2(x,y)$ slice, the gap between $\Lambda^{+,\max}_z(w)$ and $\Lambda^{+,\min}_z(w)$ intersects $R|_{\Sigma^2(x,y)}$ in two disjoint horizontal slabs ($h^+_1$ and $h^+_2$). Likewise, the gap between $\Lambda^{-,\max}_z(w)$ and $\Lambda^{-,\min}_z(w)$ produces another two disjoint horizontal slabs ($h^-_1$ and $h^-_2$). Note the two pairs of identical parabolas: $\Lambda^{-,\min}_z(w) = \Gamma^{\min}_z(w)$ and $\Lambda^{+,\max}_z(w) = \Gamma^{\max}_z(w)$.}
        \label{fig:4D_triply_folded_horseshoe_z_w_slice_Lambdas}
\end{figure}

It is convenient to remove the $y$-dependence from Eqs.~(\ref{eq:Lambda z + lower and upper bounds on each slice}) and (\ref{eq:Lambda z - lower and upper bounds on each slice}) by introducing uniform bounds for $\Lambda^{\pm}_z(w,y,s)$ over all $(y,s)$. A straightforward calculation yields
\begin{equation}\label{eq:Gamma z +- uniform bounds}
\Lambda^{\pm,\min}_z(w) \leq \Lambda^{\pm}_z(w,y,s) \leq \Lambda^{\pm,\max}_z(w),
\end{equation}
where the extremal values occur at
\begin{eqnarray}
& \Lambda^{+,\min}_z(w) = \Lambda^+_z(w,y,s)\big|_{(y,s)=(r,-r)} = -w^2 + a + \sqrt{a-2r}, \label{eq:Lambda z + min} \\
& \Lambda^{+,\max}_z(w) = \Lambda^+_z(w,y,s)\big|_{(y,s)=(-r,r)} = -w^2 + a + \sqrt{a+2r}, \label{eq:Lambda z + max} \\
& \Lambda^{-,\min}_z(w) = \Lambda^-_z(w,y,s)\big|_{(y,s)=(-r,r)} = -w^2 + a - \sqrt{a+2r}, \label{eq:Lambda z - min} \\
& \Lambda^{-,\max}_z(w) = \Lambda^-_z(w,y,s)\big|_{(y,s)=(r,-r)} = -w^2 + a - \sqrt{a-2r}. \label{eq:Lambda z - max}
\end{eqnarray}

The four parabolas in Eqs.~(\ref{eq:Lambda z + min})–(\ref{eq:Lambda z - max}) are shown schematically in Fig.~\ref{fig:4D_triply_folded_horseshoe_z_w_slice_Lambdas}.  
Since $a > 5 + 2\sqrt{5}$, we have
\begin{equation}\label{eq:a-2r}
a - 2r > 0,
\end{equation}
ensuring that the square roots in Eqs.~(\ref{eq:Lambda z + min}) and (\ref{eq:Lambda z - max}) are real.  
Moreover, because $r = \sqrt{a+2r}$, it follows that
\begin{eqnarray}
& \Lambda^{+,\max}_z(w) = \Gamma^{\max}_z(w), \label{eq:Identical parabola bounds 2 max} \\
& \Lambda^{-,\min}_z(w) = \Gamma^{\min}_z(w), \label{eq:Identical parabola bounds 2 min}
\end{eqnarray}
as indicated in Fig.~\ref{fig:4D_triply_folded_horseshoe_z_w_slice_Lambdas}.  
Thus, conditions (a) and (b) apply directly to $\Lambda^{-,\min}_z(w)$ and $\Lambda^{+,\max}_z(w)$, respectively.  
This guarantees that the gap between $\Lambda^{+,\max}_z(w)$ and $\Lambda^{+,\min}_z(w)$ intersects $R|_{\Sigma^2(x,y)}$ in two disjoint horizontal slabs ($h^+_1$ and $h^+_2$), and the gap between $\Lambda^{-,\max}_z(w)$ and $\Lambda^{-,\min}_z(w)$ produces another two disjoint horizontal slabs ($h^-_1$ and $h^-_2$). These four slabs are mutually disjoint.

Finally, Eqs.~(\ref{eq:Identical parabola bounds 2 max})–(\ref{eq:Identical parabola bounds 2 min}) imply that $h^{\pm}_1 \subset H_1$ and $h^{\pm}_2 \subset H_2$.  
Therefore, in each $\Sigma^2(x,y)$ slice (see Fig.~\ref{fig:4D_triply_folded_horseshoe_z_w_slice_Lambdas}),
\begin{equation}\label{eq:V intersects f_V V z w slice second bound}
f_{\rm IV}(d(v)) \cap R|_{\Sigma^2(x,y)} 
\subset h^{+}_1 \cup h^{-}_1 \cup h^{+}_2 \cup h^{-}_2 
\subset H_1 \cup H_2.
\end{equation}

At this point, let us observe that Eq.~(\ref{eq:V intersects f_V V z w slice second bound}) relies solely on the first three rows of Eq.~(\ref{eq:f IV V}), and thus provides yet again a crude estimate of $f_{\rm IV}(d(v))\cap R|_{\Sigma^2(x,y)}$. In the following, we demonstrate that by additionally imposing the fourth row of Eq.~(\ref{eq:f IV V}), the location of $f_{\rm IV}(d(v))\cap R|_{\Sigma^2(x,y)}$ can be verified---that is, it consists of eight disjoint horizontal slices of $R$, as asserted by the theorem, each of which can be shown to lie within one of eight disjoint horizontal slabs, as indicated by the eight shaded strips in Fig.~\ref{fig:4D_triply_folded_horseshoe_z_w_slice_Thetas}. 

The fourth row of Eq.~(\ref{eq:f IV V}) can be rewritten in the parameterized form  
\begin{equation}\label{eq:z w parabola 3 f_V}
a - x - [\,a - y - (a - z - w^2)^2\,]^2 = -s,
\end{equation}  
where \( s = -w_v \in [-r,r] \) is treated as a bounded parameter determined by the initial choice of \( v \in I^s \).

Solving for \( z \) yields four solution branches:  
\begin{eqnarray}\label{eq:z w parabola 3 f_V four branches}
z_{++} &= -w^2 + a + \sqrt{a-y + \sqrt{s-x+a}}, \label{eq:z w parabola 3 f_V four branches ++} \\
z_{+-} &= -w^2 + a + \sqrt{a-y - \sqrt{s-x+a}}, \label{eq:z w parabola 3 f_V four branches +-} \\
z_{-+} &= -w^2 + a - \sqrt{a-y + \sqrt{s-x+a}}, \label{eq:z w parabola 3 f_V four branches -+} \\
z_{--} &= -w^2 + a - \sqrt{a-y - \sqrt{s-x+a}}. \label{eq:z w parabola 3 f_V four branches --}
\end{eqnarray}
Since \( |s|, |x|, |y| \leq r \) and \( a - 2r > 0 \), it follows that \( s - x + a > 0 \) and \( a - y \pm \sqrt{s-x+a} > 0 \).  
Thus, all square roots in Eqs.~(\ref{eq:z w parabola 3 f_V four branches ++})--(\ref{eq:z w parabola 3 f_V four branches --}) are real-valued.

Accordingly, we define four families of parabolas in \( \Sigma^2(x,y) \), denoted by \( \Theta_z^{\pm\pm}(w,x,y,s) \):  
\begin{eqnarray}
\Theta_z^{++}(w,x,y,s) &= -w^2 + a + \sqrt{a-y + \sqrt{s-x+a}}, \label{eq:Theta z in terms of w x y s ++} \\
\Theta_z^{+-}(w,x,y,s) &= -w^2 + a + \sqrt{a-y - \sqrt{s-x+a}}, \label{eq:Theta z in terms of w x y s +-} \\
\Theta_z^{-+}(w,x,y,s) &= -w^2 + a - \sqrt{a-y + \sqrt{s-x+a}}, \label{eq:Theta z in terms of w x y s -+} \\
\Theta_z^{--}(w,x,y,s) &= -w^2 + a - \sqrt{a-y - \sqrt{s-x+a}}, \label{eq:Theta z in terms of w x y s --}
\end{eqnarray}
where \( x \), \( y \), and \( s = -w_v \) are treated as parameters satisfying \( |x| \leq r \), \( |y| \leq r \), and \( |s| \leq r \).

For a given \( (x,y) \), the lower and upper bounds of \( \Theta^{\pm\pm}_z(w,x,y,s) \) with respect to variations in \( s \) are denoted by \( \Theta^{\pm\pm,1}_z(w,x,y) \) and \( \Theta^{\pm\pm,2}_z(w,x,y) \), respectively.  
They satisfy  
\begin{equation}\label{eq:lower upper bounds Theta in z w slice}
\Theta^{\pm\pm,1}_z(w,x,y) \leq \Theta^{\pm\pm}_z(w,x,y,s) \leq \Theta^{\pm\pm,2}_z(w,x,y),
\end{equation}
and are given explicitly by  
\begin{eqnarray}
\Theta_z^{++,1}(w,x,y) &= -w^2 + a + \sqrt{a-y + \sqrt{-r-x+a}}, \label{eq:Theta z ++ 1} \\
\Theta_z^{++,2}(w,x,y) &= -w^2 + a + \sqrt{a-y + \sqrt{r-x+a}}, \label{eq:Theta z ++ 2} \\
\Theta_z^{+-,1}(w,x,y) &= -w^2 + a + \sqrt{a-y - \sqrt{r-x+a}}, \label{eq:Theta z +- 1} \\
\Theta_z^{+-,2}(w,x,y) &= -w^2 + a + \sqrt{a-y - \sqrt{-r-x+a}}, \label{eq:Theta z +- 2} \\
\Theta_z^{-+,1}(w,x,y) &= -w^2 + a - \sqrt{a-y + \sqrt{r-x+a}}, \label{eq:Theta z -+ 1} \\
\Theta_z^{-+,2}(w,x,y) &= -w^2 + a - \sqrt{a-y + \sqrt{-r-x+a}}, \label{eq:Theta z -+ 2} \\
\Theta_z^{--,1}(w,x,y) &= -w^2 + a - \sqrt{a-y - \sqrt{-r-x+a}}, \label{eq:Theta z -- 1} \\
\Theta_z^{--,2}(w,x,y) &= -w^2 + a - \sqrt{a-y - \sqrt{r-x+a}}. \label{eq:Theta z -- 2}
\end{eqnarray}

\begin{figure}
        \centering
        \includegraphics[width=0.8\linewidth]{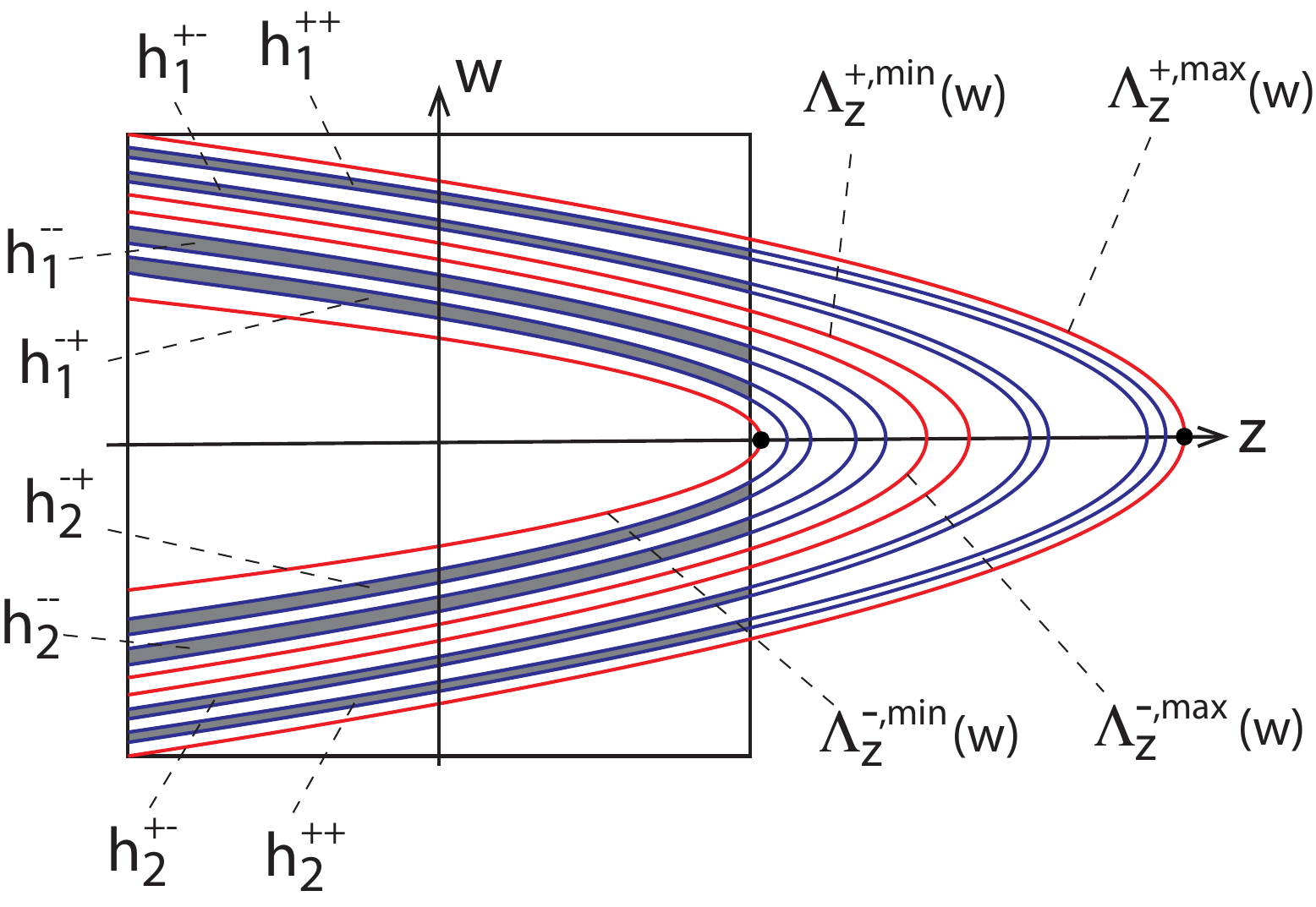} 
        \caption{(Schematic, color online) The eight parabolas in Eq.~(\ref{eq:Theta z pm pm 1 2 ordering}) (blue curves) are arranged from right to left. The gap between each successive pair—(\(1^{\rm st}\), \(2^{\rm nd}\)), (\(3^{\rm rd}\), \(4^{\rm th}\)), (\(5^{\rm th}\), \(6^{\rm th}\)), and (\(7^{\rm th}\), \(8^{\rm th}\))—intersects \( R|_{\Sigma^2(x,y)} \) in two disjoint horizontal slabs, yielding eight mutually disjoint slabs (shaded) labeled \( h^{\pm\pm}_1 \) and \( h^{\pm\pm}_2 \). The restriction \( f_{\rm IV}(d(v)) \cap R|_{\Sigma^2(x,y)} \) therefore consists of eight disjoint horizontal slices, each contained in one of these slabs.
}   \label{fig:4D_triply_folded_horseshoe_z_w_slice_Thetas}
\end{figure}

For any fixed \( (w,x,y) \), the ordering of the eight parabolas is  
\begin{eqnarray}
&\Theta^{++,2}_z(w,x,y) > \Theta^{++,1}_z(w,x,y) > \Theta^{+-,2}_z(w,x,y) > \Theta^{+-,1}_z(w,x,y) \nonumber \\
&> \Theta^{--,2}_z(w,x,y) > \Theta^{--,1}_z(w,x,y) > \Theta^{-+,2}_z(w,x,y) > \Theta^{-+,1}_z(w,x,y) \label{eq:Theta z pm pm 1 2 ordering}
\end{eqnarray}
showing that they are positioned from right to left without intersections, as illustrated in Fig.~\ref{fig:4D_triply_folded_horseshoe_z_w_slice_Thetas}.

Having determined the relative ordering of the eight parabolas, we next bound their positions.  
In particular, we will show that the four parabolas in the first row of Eq.~(\ref{eq:Theta z pm pm 1 2 ordering}) lie between \( \Lambda^{+,\max}_z(w) \) and \( \Lambda^{+,\min}_z(w) \), while those in the second row lie between \( \Lambda^{-,\max}_z(w) \) and \( \Lambda^{-,\min}_z(w) \), as shown in Fig.~\ref{fig:4D_triply_folded_horseshoe_z_w_slice_Thetas}.

To prove this, we derive uniform lower and upper bounds for each of Eqs.~(\ref{eq:Theta z in terms of w x y s ++})–(\ref{eq:Theta z in terms of w x y s --}) as the parameters \( x \), \( y \), and \( s \) vary:  
\begin{equation}\label{eq:Thetas pm uniform bounds}
\Theta^{\pm\pm,\min}_z(w) \leq \Theta^{\pm\pm}_z(w,x,y,s) \leq \Theta^{\pm\pm,\max}_z(w),
\end{equation} 
where the bounds are given explicitly by  
\begin{eqnarray}
\Theta^{++,\min}_z(w) &= \Theta^{++}_z(w,x,y,s)|_{(x,y,s)=(r,r,-r)} \nonumber \\
&= -w^2 + a + \sqrt{a-r + \sqrt{a-2r}}, \label{eq:Theta z in terms of w x y s ++ min} \\
\Theta^{++,\max}_z(w) &= \Theta^{++}_z(w,x,y,s)|_{(x,y,s)=(-r,-r,r)} \nonumber \\
&= -w^2 + a + \sqrt{a+r + \sqrt{a+2r}} = \Lambda^{+,\max}_z(w), \label{eq:Theta z in terms of w x y s ++ max} \\
\Theta^{+-,\min}_z(w) &= \Theta^{+-}_z(w,x,y,s)|_{(x,y,s)=(-r,r,r)} \nonumber \\
&= -w^2 + a + \sqrt{a-r - \sqrt{a+2r}} = \Lambda^{+,\min}_z(w), \label{eq:Theta z in terms of w x y s +- min} \\
\Theta^{+-,\max}_z(w) &= \Theta^{+-}_z(w,x,y,s)|_{(x,y,s)=(r,-r,-r)} \nonumber \\
&= -w^2 + a + \sqrt{a+r - \sqrt{a-2r}}, \label{eq:Theta z in terms of w x y s +- max} \\
\Theta^{-+,\min}_z(w) &= \Theta^{-+}_z(w,x,y,s)|_{(x,y,s)=(-r,-r,r)} \nonumber \\
&= -w^2 + a - \sqrt{a+r + \sqrt{a+2r}} = \Lambda^{-,\min}_z(w), \label{eq:Theta z in terms of w x y s -+ min} \\
\Theta^{-+,\max}_z(w) &= \Theta^{-+}_z(w,x,y,s)|_{(x,y,s)=(r,r,-r)} \nonumber \\
&= -w^2 + a - \sqrt{a-r + \sqrt{a-2r}}, \label{eq:Theta z in terms of w x y s -+ max} \\
\Theta^{--,\min}_z(w) &= \Theta^{--}_z(w,x,y,s)|_{(x,y,s)=(r,-r,-r)} \nonumber \\
&= -w^2 + a - \sqrt{a+r - \sqrt{a-2r}}, \label{eq:Theta z in terms of w x y s -- min} \\
\Theta^{--,\max}_z(w) &= \Theta^{--}_z(w,x,y,s)|_{(x,y,s)=(-r,r,r)} \nonumber \\
&= -w^2 + a - \sqrt{a-r - \sqrt{a+2r}} = \Lambda^{-,\max}_z(w), \label{eq:Theta z in terms of w x y s -- max}
\end{eqnarray}
where the final equalities in Eqs.~(\ref{eq:Theta z in terms of w x y s ++ max}), (\ref{eq:Theta z in terms of w x y s +- min}), (\ref{eq:Theta z in terms of w x y s -+ min}), and (\ref{eq:Theta z in terms of w x y s -- max}) follow from \( \sqrt{a+2r} = r \).  

Combining Eqs.~(\ref{eq:Theta z pm pm 1 2 ordering}), (\ref{eq:Thetas pm uniform bounds}), (\ref{eq:Theta z in terms of w x y s ++ max}), (\ref{eq:Theta z in terms of w x y s +- min}), (\ref{eq:Theta z in terms of w x y s -+ min}), and (\ref{eq:Theta z in terms of w x y s -- max}) yields  
\begin{eqnarray}
&\Lambda^{+,\max}_z(w) \geq \Theta^{++,2}_z(w,x,y), \label{eq:Theta z final bounds 1} \\
&\Theta^{+-,1}_z(w,x,y) \geq \Lambda^{+,\min}_z(w), \label{eq:Theta z final bounds 2} \\
&\Lambda^{-,\max}(w) \geq \Theta^{--,2}_z(w,x,y), \label{eq:Theta z final bounds 3} \\
&\Theta^{-+,1}_z(w,x,y) \geq \Lambda^{-,\min}_z(w). \label{eq:Theta z final bounds 4}
\end{eqnarray}
These relations confirm that the four parabolas in the first row of Eq.~(\ref{eq:Theta z pm pm 1 2 ordering}) lie between \( \Lambda^{+,\max}_z(w) \) and \( \Lambda^{+,\min}_z(w) \), while the four in the second row lie between \( \Lambda^{-,\max}_z(w) \) and \( \Lambda^{-,\min}_z(w) \), as illustrated in Fig.~\ref{fig:4D_triply_folded_horseshoe_z_w_slice_Thetas}.

Therefore, the eight parabolas in Eq.~(\ref{eq:Theta z pm pm 1 2 ordering}) are arranged so that the gap between each successive pair—(first, second), (third, fourth), (fifth, sixth), and (seventh, eighth)—intersects \( R|_{\Sigma^2(x,y)} \) in two disjoint horizontal slabs.  
This yields a total of eight mutually disjoint horizontal slabs, labeled \( h^{\pm\pm}_1 \) and \( h^{\pm\pm}_2 \) in Fig.~\ref{fig:4D_triply_folded_horseshoe_z_w_slice_Thetas} (shaded).  
Hence, in each \( \Sigma^2(x,y) \) slice, \( f_{\rm IV}(d(v)) \cap R|_{\Sigma^2(x,y)} \) consists of eight disjoint horizontal slices of \( R|_{\Sigma^2(x,y)} \), each lying within one of these slabs:  
\begin{eqnarray}\label{eq:V intersects f_V V z w slice final bound}
f_{\rm IV}(d(v))\cap R \Big|_{\Sigma^2(x,y)} 
&\subset \bigcup_{i\in \{1,2\}} \left( h^{++}_i \cup h^{+-}_i \cup h^{-+}_i \cup h^{--}_i \right) \nonumber \\
&\subset h^{+}_1 \cup h^{-}_1 \cup h^{+}_2 \cup h^{-}_2 \subset H_1 \cup H_2 \, .
\end{eqnarray}
The theorem is thus proved.

\end{proof}

\begin{remark}
By Theorem~\ref{4D triply folded horseshoe topology} and Proposition~\ref{Semi-conjugacy to the full shift on N symbols}, there exists a nonempty compact invariant set 
\(\Lambda \subset \bigcap_{n\in\mathbb{Z}} f_{\rm IV}^n(R)\) 
such that the restricted map \( f_{\rm IV}|_{\Lambda} \) is semi-conjugate to the full shift on eight symbols.
\end{remark}

\section{Conclusion}\label{Conclusion}

We have constructed explicit three- and four-dimensional H\'{e}non-type maps that realize horseshoe structures with folding topologies unattainable in two dimensions. By tailoring each map to implement specific combinations of crease orientations and stacking directions, we obtained a variety of qualitatively distinct configurations, each naturally captured by an associated paperfolding template. 

These examples represent only a small subset of the folding–stacking arrangements possible in higher-dimensional phase spaces. Many additional topologies could emerge by varying the dimensionality of the folded sheet, the embedding space, and the orientations of creases and stacking directions. A natural next step is to determine which of these configurations occur generically in broad classes of multidimensional maps and to assess their implications for the corresponding symbolic dynamics.

In the symplectic setting (map~$f_{\rm III}$), our earlier work~\cite{Fujioka25} demonstrated that parameter regions associated with the Type-A and Type-B AI limits support both a topological horseshoe and uniform hyperbolicity, ensuring conjugacy to a full shift. Extending such results to the other maps introduced here remains an open problem. Computer-assisted methods~\cite{arai2007hyperbolic} offer a promising means to rigorously delineate parameter regions of uniform hyperbolicity and to investigate how distinct horseshoe topologies coexist and bifurcate within parameter space. Such studies would advance the systematic classification of chaotic structures in multidimensional dynamical systems.

\section*{Acknowledgement}
J.L. thanks Steven Tomsovic for many insightful and inspiring discussions.
J.L. and A.S. acknowledge support from the Japan Society for the Promotion of Science (JSPS) through the JSPS Postdoctoral Fellowships for Research in Japan (Standard).
This work was also supported by JSPS KAKENHI (Grant Nos.~17K05583, 19F19315, and 22H01146) and by the Japan Science and Technology Agency (JST) under the program “Establishment of University Fellowships toward the Creation of Science and Technology Innovation” (Grant No.~JPMJFS2139).

\section*{References}
\bibliographystyle{iopart-num}
\bibliography{classicalchaos,paperfolding}
\end{document}